\documentclass{lmcs} 
\pdfoutput=1

\usepackage{lastpage}
\lmcsdoi{17}{4}{26}
\lmcsheading{}{\pageref{LastPage}}{}{}%
{Nov.~10,~2020}{Dec.~30,~2021}{}


\theoremstyle{plain} 

\usepackage[utf8]{inputenc}
\usepackage{shuffle}
\usepackage{tikz}
\usetikzlibrary{arrows,shapes.geometric,positioning,automata}

\setlist{itemsep=0.2cm}

\allowdisplaybreaks
\newcommand{\B}{{\mathcal{B}}}

\newcommand{\e}{{\tilde{\varphi}}}
\newcommand{\x}{{\tilde{\xi}}}
\newcommand{\ps}{{\tilde{\psi}}}

\begin{document}

\title[Modelling of parametric architectures]{Architectures in parametric component-based systems: \\ Qualitative and quantitative modelling\rsuper*}
\titlecomment{{\lsuper*}A preliminary version of this paper appeared in \cite{Pi:Ar}.}

\author[M.~Pittou]{Maria Pittou$^1$}	
\thanks{$^1$\protect\includegraphics[height=0.3cm]{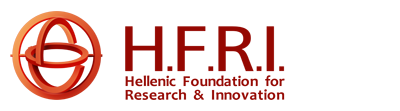}The research work was supported by the Hellenic Foundation for Research and Innovation (HFRI) under the HFRI PhD Fellowship grant (Fellowship Number: 1471).}	

\author[G.~Rahonis]{George Rahonis}	
\address{Department of Mathematics\\Aristotle University of Thessaloniki\\54124 Thessaloniki, Greece}	
\email{\{mpittou,grahonis\}@math.auth.gr}  





\begin{abstract}
  \noindent One of the key aspects in component-based design is specifying the software architecture
that characterizes the topology and the permissible interactions of the components of a system. To achieve well-founded design there is need to address both the qualitative and non-functional aspects of architectures. In this paper we study the qualitative and quantitative formal modelling of architectures applied on parametric component-based systems, that consist of an unknown number of instances
of each component. Specifically, we introduce an extended propositional interaction logic and
investigate its first-order level which serves as a formal language for the interactions of parametric systems. Our logics achieve to encode the execution order of interactions, which is a main feature in several important architectures, as well as to model recursive interactions. Moreover, we prove the decidability of equivalence, satisfiability, and validity of first-order extended interaction logic formulas, and provide several examples of formulas describing well-known architectures. We show the robustness of our theory by effectively extending our results for parametric weighted architectures. 
 For this, we study the weighted counterparts of our logics over a commutative semiring, and we apply them for modelling the quantitative aspects of concrete architectures. Finally, we prove that the equivalence problem of weighted first-order extended interaction logic formulas is decidable in a large class of semirings, namely the class (of subsemirings) of skew fields.
\end{abstract}

\maketitle


\section{Introduction}
Well-founded design is a key principle for complex systems in order to guarantee correctness and performance.
Rigorous formalisms in systems engineering are mainly component-based that allow reconfigurability and validation \cite{Bl:Al}.
Component-based design lies in constructing multiple
components which coordinate through their interfaces in order to generate the global model for a system \cite{Gi:Re}. 
In such a setting, one of the main issues in the modelling process is defining the communication patterns of systems.
Coordination principles among components can be specified by means of architectures, which 
characterize the permissible interactions and their implementation order, as well as the underlying topology of the systems \cite{Ma:Co,Ma:Sp}. 
Architectures have been proved important in the modelling of systems since they enforce design rules on the components, and hence
ensure correctness by construction with respect to basic properties such as mutual exclusion \cite{Bl:Ve,Bo:Ch,Ki:An,Ma:Co}. 

The formal modelling of architectures is a well-known problem of component-based systems and has been investigated with versatile approaches.
Some existing frameworks focus on modelling only the components' connections \cite{Ma:Co,Pa:On}, while other integrate the specification of components' behavior and communication in order  to express requirements \cite{Bo:Ch,Ma:Sp}. In some cases the proposed formal approaches are further supported by graphical languages for characterizing architectures or tools for the verification of component-based systems and their architectures (cf. \cite{Ki:An,Ko:Pa,Ma:Ar,Me:Cl}).  
Although, there is several work on the qualitative modelling of architectures, this has not been the case for their quantitative aspects. 
In this paper we propose a formal approach for the characterization of architectures 
where components' communication is expressed without incorporating their behavior. 
Moreover, our framework addresses not only the qualitative but also the weighted setting for specifying architectures, which is the main novelty of our contribution. The consideration of the quantitative properties of architectures is crucial for the efficient design of component-based systems.
 Cross-cutting concerns involve timing constraints, available resources, energy consumption, probabilities, etc., for implementing the
communication among the systems' components. Consequently, in order to model such optimization issues
there is need to study architectures in the weighted setup \cite{An:Pa,Be:Pa,Pa:On}.

The current paper introduces a generic framework for the qualitative and quantitative modelling of architectures applied to a wide class of large systems,
and specifically to parametric systems. Parametric systems are component-based systems 
found in several applications including communication protocols, concurrent processes, and distributed algorithms \cite{Ab:Pa,Bl:De,De:Pa}. 
Parametric systems are constructed by a finite number of component types whose number of 
instances is unknown, i.e., is a parameter for the system. Hence, specifying the architecture of parametric systems
is challenging since their components' number is not known in advance, 
affecting in turn the interactions that define their communication. On the other hand, efficient 
characterization of architectures is important in parametric systems' design for identifying decidable classes of their verification problem, that
is undecidable in general if unbounded data is exchanged \cite{Am:Pa,Bl:De,Es:Po}. 

The formal characterization of parametric architectures proposed in this work is logic-based.
In particular, we firstly introduce and investigate a first-order extended interaction logic, namely FOEIL, that models efficiently the qualitative characteristics of parametric architectures. In turn we consider its weighted counterpart, weighted FOEIL over a commutative semiring $K$, in order to address the corresponding non-functional aspects.
In contrast to existing logics for modelling parametric architectures \cite{Bo:Ch,Bo:St,Ko:Pa,Ma:Co}, our FOEIL achieves to encode the execution order of the interactions imposed by the corresponding architecture as well as to describe recursive interactions.
Specifically, several parametric architectures found in applications,
including Publish/Subscribe and Request/Response, impose restrictions on the order of
the permissible interactions \cite{Da:Se,Zh:To}. 
For instance, in a Publish/Subscribe architecture a subscriber cannot receive a message from a topic of its interest if beforehand 
a publisher has not transferred the message to the particular topic. On the other hand, distinct
subscribers may express their interest to the same topic in any order. Our logic is proved sufficient enough to 
model such order restrictions. Moreover, we allow recursion in the components' interactions which in turn implies that
we can characterize the subsequent implementation of an architecture during a parametric system's operation. Then, our weighted FOEIL maintains the qualitative attributes of FOEIL, and also models the quantitative properties of parametric architectures, such as the total cost of the interactions or the probability of the implementation of concrete interactions.

In our setting we model 
components with the standard formalism
of labelled transitions systems (cf. \cite{Al:Th,Am:RE,Bl:Al}). 
Then, communication of components is performed by their associated set of labels, called ports. In turn, architectures are modelled by FOEIL
formulas encoding the order and recursion of the respective allowed interactions, which are defined by nonempty sets of ports.
The presented logic-based modelling framework of architectures does not require the knowledge of
the actual transitions of the components since it is constructed on the corresponding set of ports of the given system. 
As a result our theory does not consider the components' behavior and can be applied to every
component-based framework where the system's interface is described by a set of ports. For the weighted setup
we associate each port with a weight that represents the `cost' of its participation
in an interaction. The weights of the ports range over a
commutative semiring $K$ and we formalize the quantitative aspects of parametric architectures by weighted
FOEIL formulas over $K$. In particular, the contributions of the current paper are the following:

(1) We introduce \emph{Extended Propositional Interaction Logic} (EPIL for short) over a finite set of ports,
which augments PIL from \cite{Ma:Co} with three operators namely the concatenation operator $*$, the shuffle operator $\shuffle$, and the 
iteration operator $^{+}$. In contrast to classical PIL,
where the satisfaction relation is checked against interactions (nonempty sets of ports),
the formulas of EPIL are interpreted over finite words whose letters are interactions over the given set of ports. Intuitively, the semantics of concatenation
operator specifies consecutive interactions while the semantics of shuffle operator encodes interleaving, i.e., all possible orders for the execution of permissible interactions in
the system. Moreover, the iteration operator serves for modelling recursive interactions in the architecture. We apply EPIL formulas for formalizing the architectures of component-based systems with ordered interactions,
and specifically we present three examples with 
the architectures \emph{Blackboard} \cite{Co:Bl}, \emph{Request/Response} \cite{Da:Se}, and \emph{Publish/Subscribe} \cite{Eg:Pu}.

(2) We introduce  the first-order level of EPIL, namely \emph{First-Order Extended Interaction Logic} (FOEIL for short),
 as a modelling language for the architectures of parametric systems. The syntax of FOEIL is equipped with the syntax of EPIL, the common existential and universal quantifiers,
 and four new quantifiers, namely existential and universal concatenation
and shuffle quantifiers. The new quantifiers achieve to
encode the partial and whole participation of component instances in sequential and interleaved interactions of parametric architectures. 
For the semantics of FOEIL we consider triples consisting of 
a mapping defining the number of component instances in the parametric system,
an assignment that attributes a
unique identifier to the ports of each component instance, and a finite word of
interactions.

(3) We show the expressiveness of FOEIL by providing several examples
for architectures of parametric component-based systems.
 In particular we consider the architectures \emph{Blackboard}, \emph{Request/Response} and \emph{Publish/Subscribe}
 which impose orders on the implementation of their interactions, as well as \emph{Master/Slave} \cite{Ma:Co}, \emph{Star} \cite{Ma:Co}, \emph{Repository}
 \cite{Cl:Do}, and \emph{Pipes/Filters} \cite{Ga:An}  whose interactions may be executed in arbitrary order. 

(4) We study decidability results related with FOEIL. For this, we firstly 
state an effective translation of FOEIL formulas to  finite automata. The best case run time of our translation
 algorithm is polynomial and the worst case is exponential. In turn, we obtain the decidability of equivalence and validity
 for FOEIL sentences in doubly exponential time, and the decidability of satisfiability for FOEIL sentences in exponential time.

(5) We introduce \emph{weighted Extended Propositional Interaction Logic} (wEPIL for short) over a finite set of ports and a commutative semiring $K$ for representing the weights. wEPIL extends \emph{weighted Propositional Interaction Logic}
 (wPIL for short) from \cite{Pa:On,Pa:We} with three new weighted operators, namely the weighted concatenation
 operator $\odot$, the weighted shuffle operator $ \varpi$, and the weighted iteration operator $^{\oplus}$.
 Intuitively these operators allow to encode the weight of consecutive, interleaving, and recursive
 interactions in weighted component-based systems, respectively. We interpret wEPIL formulas as series defined over
 finite words and $K$. In turn, we study the first-order level of wEPIL, namely \emph{weighted First-Order Extended Interaction Logic} (wFOEIL for short) 
over $K$ for modelling parametric weighted architectures. wFOEIL is equipped with the corresponding operators and quantifiers of FOEIL in the weighted setup, and is interpreted against series over $K$.

(6) We provide examples which show that wEPIL and wFOEIL serve sufficiently for modelling the quantitative aspects of architectures with 
ordered and recursive interactions. Then, for different instantiations of the semiring $K$ in our examples we derive alternative interpretations
for the resulting cost of the allowed interactions, that corresponds to some quantitative characteristic. Similarly to the unweighted setting, we show that we can also apply wFOEIL, and hence wEPIL, for expressing weighted architectures whose interactions may be executed in arbitrary order. 

(7) Finally, we establish a translation of wFOEIL formulas to weighted automata. We show that the worst case run time for our translation algorithm is doubly exponential and the best case run time is exponential. Then, we prove
 the decidability of equivalence of wFOEIL sentences in doubly exponential
time over a large class of semirings, namely (subsemirings of) skew fields. Therefore, the complexity remains
 the same with the one for the decidability of equivalence of FOEIL sentences, which depicts the robustness of our theory.

A preliminary version of this paper appeared in \cite{Pi:Ar}. The present version extends the work of \cite{Pi:Ar} as follows:

- We equip EPIL and FOEIL with the iteration operator $^{+}$, and hence we model the repetition of the interactions in parametric architectures. Therefore, the overall results in the paper have been modified accordingly to incorporate recursion. 

- We present extensive descriptions for the examples considered in \cite{Pi:Ar} and we
provide several new examples of architectures modelled by FOEIL sentences. Moreover, we present the detailed proofs for the decidability results of FOEIL.

-We introduce a whole new section that actually comprises half of the current contribution, and addresses the quantitative formal modelling of parametric architectures. 
In particular, we provide the weighted counterparts of EPIL and FOEIL, we apply them on concrete weighted architectures, and we investigate the corresponding decidability results. 

The structure of the paper is as follows. In Section $2$ we discuss related work and in Section $3$ 
we recall the basic notions for component-based systems and interactions. Then, in Section $4$ 
we introduce the syntax and semantics of EPIL and present examples of architectures defined
 by EPIL formulas. In Section $5$ we introduce the syntax and semantics of our FOEIL and provide
examples of FOEIL sentences describing concrete parametric architectures. 
Section $6$ deals with the decidability results for FOEIL sentences. Then, in Section $7$, we introduce and study wEPIL and wFOEIL over a commutative semiring. Finally, in Conclusion, we present open problems and future work.

\section{Related work}
\label{sec_rel}
Existing work in parametric architectures has investigated their formal modelling mainly in the qualitative setting. Among the several instantiations of the problem is
included the study of static \cite{Bo:Ch, Ko:Pa} or dynamic architectures \cite{Bo:Mo,Ci:Fo,Ma:Sp}, of architectures with data \cite{Bl:Ve,Hu:Fo} as well as of the architectures' composition problem \cite{At:Ge,Bl:Ve}. On the other hand, some work focus on the modelling of architectures without considering the underlying behavior of the components \cite{Ko:Pa,Ma:Co,Pa:On}, while other deal with the specification both of components' behavior and communication \cite{Bo:Ch,Ma:Sp}. In some cases the proposed formal approaches are supported by graphical languages or tools for characterizing and verifying architectures (cf. \cite{Ki:An,Ko:Pa,Ma:Ar,Me:Cl}).  
In the sequel we present an extensive description of some work, in the qualitative and in the weighted setup, that is closely related with our framework. Our modelling approach differs from each of the work discussed below in at least one of following directions: 1) it is logic-based, 2) it 
deals with ordering constraints and recursion in the parametric communication, 3) it addresses the quantitative modelling of parametric architectures that 
has not been considered in most of the existing frameworks. 

In \cite{Ma:Co} the authors introduced a Propositional Configuration Logic (PCL for short) as a modelling language for the description of architectures. A configuration was defined as a set of interactions over a given set of ports. PCL formulas were obtained by the formulas of PIL using the operations of union, intersection and complementation, as well as a new operator, called coalescing operator, for expressing combination of configuration sets. PCL was interpreted over configurations and the authors proved the decidability of equivalence of its formulas in an automated way using the Maude language. In \cite{Ma:Co} it was also studied the first- and second-order configuration logic for describing parametric architectures (called styles of architectures in that paper). The first-order level of PCL was applied for modelling Master/Slave, Star, Pipes/Filters, Repository, Blackboard, and Request/Response architectures, while the second-order level of PCL encoded Ring, Grid, and Linear architectures. In the subsequent work \cite{Ma:Ar} the authors supported their framework with a graphical language, based on architecture diagrams, for describing architecture styles. Both our FOEIL and the first-order of PCL can express sufficiently the same parametric architectures, while the study of second-order level of EPIL is work in progress. Our contribution with respect to the work of \cite{Ma:Co} is that our logics achieve to model the order of implementing components' interactions, that cannot be captured by PCL. We further clarify this by Remark \ref{remark} presented in Section \ref{examples_EPIL} that discusses an example on a Blackboard architecture. Another difference is that we also consider recursion in architectures, and mainly that we investigate the quantitative modelling of parametric architectures that has not been addressed in \cite{Ma:Co}. 

In \cite{Ko:Pa} the first-order level of
propositional interaction logic with arithmetics, namely FOIL, was introduced to describe finitely many interactions for parametric systems
in BIP (Behavior-Interaction-Priority) framework (cf. \cite{Bl:Al}). The syntax of FOIL was built on a set of ports (of the parametric components), the common disjunction and conjunction operators, the existential and universal quantifiers, 
and formulas in Presburger arithmetic. FOIL formulas modelled the permissible interactions of parametric architectures as sets. In \cite{Ko:Pa} the authors provided also a framework and a prototype with FOIL templates for identifying the architecture of a random parametric BIP model, which was
then associated  with a known verification method  for the particular type of communication. 
FOIL served sufficiently for expressing several classical architectures including token-passing rings, rendezvous cliques, broadcast
cliques, and rendezvous stars. However, in contrast to our logic, FOIL cannot encode order restrictions imposed on the execution of interactions since its formulas return sets of interactions. For this, FOIL was not applied for modelling complex architectures such as Request/Response or Publish/Subscribe. On the other hand, the use of arithmetic constraints in FOIL formulas allowed the logic to describe token-passing ring architectures, that in our setting could be expressed by the second-order level of EPIL.  

In \cite{Bo:Ch} the authors introduced Monadic Interaction Logic (MIL for short)
as an alternative logic for 
interactions of parametric systems.
MIL syntax was equipped with first-order variables referring to instances of the parametric components, equality checking of variables, 
 state and port symbols, the disjunction and negation operations, as well as with the common existential quantifier of first order variables. Although, MIL achieved to encode rendezvous and broadcast communication in parametric systems, it was not considered for modelling
architectures since it did not include cardinality constraints. The authors applied in turn MIL for the development of an automated method for detecting deadlocks in parametric systems. In the same line,  in \cite{Bo:St} the authors and introduced Interaction Logic with One Successor (IL1S for short) 
for describing rendezvous and broadcast communication as well as several architectures
in parametric systems. In particular, IL1S preserved the syntax and semantics of MIL with the difference that it also involved a cyclic modulo-type successor function, and hence achieved to model
 architectures of ring, linear, and pipeline
type, as well as architectures with tree-like structures. 
IL1S was proved to be decidable and was used to define parametric invariants for checking correctness of safety properties in parametric systems. 

Our work is positioned with respect to the logics presented in \cite{Bo:Ch,Bo:St} as follows. FOEIL achieves to model architectures in contrast to MIL that described 
only the communication type among interactions. On the other hand, IL1S described architectures that in our framework can be expressed with the second-order level of
EPIL, which we study in a forthcoming paper.  However, the work of \cite{Bo:St} did not include the modelling of more involved architectures such as Request/Response or Publish/Subscribe. Moreover, the semantics of MIL and IL1S formulas did not consider recursive interactions and their execution order imposed by each architecture, which is one of our main contributions. Extending our results
for the verification of parametric systems against safety properties, following a methodology similar to the work of \cite{Bo:Ch,Bo:St}, is left as future work.

BIP is a modelling framework that supports the rigorous design of behavior and coordination of component-based systems. One of the main features in BIP framework is `priorities among interactions' in a component-based system (cf. for instance \cite{Bl:Al}). A priority system is determined by a strict partial order $\prec$ among the set of permitted interactions. Hence  if $a \prec a'$ for two interactions $a$ and $a'$, then  $a'$ must be implemented before $a$ since it has higher priority. It should be clear that the priority system of BIP cannot describe the required orders of interactions in complex architectures. Specifically, the set of strings of interactions satisfying an EPIL sentence containing a shuffle operator, cannot be obtained by any strict partial order among the set of interactions.

Hennessy and Milner introduced in 1985 (cf. \cite{He:Al}) a logic, called HML, as a calculus for the specification of concurrent programs and their properties. In \cite{Fr:Mo} the authors investigated $\mu$HML, i.e., HML with recursive formulas expressing least and greatest fixpoints  and focused on a fragment of that logic which can be monitored for runtime verification of programs' execution. That logic succeeded to describe simple client/server type processes but it is far from describing complex architectures. Specifically our shuffle operator cannot be described by means of $\mu$HML.

Coordination of distributed systems was also investigated in \cite{Gu:Re} (cf. also \cite{Tu:Se}). Specifically, the authors
developed a theoretical framework and a prototype for studying the realizability problem of coordination in terms of pomsets. 
In that setting the repetition and order of input and output components' actions were considered. Though, due to the imposed orders of pomsets,
 our shuffle operator cannot be sufficiently described in this framework. For instance, the subfomula $\varphi_1 \shuffle \varphi_2$ of the EPIL
 formula $\varphi$ describing  the Publish/Subscribe architecture (cf. Example \ref{b-pu-su}) cannot be described by means of pomsets.

Alternative modelling frameworks for the communication patterns of software components include session types and behavioral contracts (cf. \cite{Hu:Fo}). In \cite{De:Pa}
 the authors introduced a 
type theory for multiparty sessions to globally specify parametric
communication protocols whose interactions carry data and their topology is specified by the Ring, Star, or Grid architectures.
The authors showed that their type-checking algorithm ensured type-safety and deadlock-freedom
of parameterised multiparty communication
protocols.  Moreover, in \cite{Ch:Pa} the authors introduced a novel session type
system and a language to describe multi-actor communication in parameterized protocols.
The proposed type theory included operators for describing exclusive, sequential, concurrent events as well as 
their arbitrary reorderings using a shuffling operator. 
Then the authors applied their type
system for static verification of asynchronous communication protocols. In contrast to our setting, 
 the parametric framework of \cite{De:Pa} did not address the sequential and interleaving interactions of processes, while the work of \cite{Ch:Pa}
did not consider the modelling problem of architectures in
parameterized systems.

Although there has been an abundance of work for parametric
systems in the qualitative setting this has not been the case for the quantitative one. 
Some work for quantitative parametric systems was considered in \cite{An:Pa, Be:Pa, Es:Po, Fo:Th, Ka:We,Pa:On, Pa:We}. 
Specifically, in \cite{Es:Po} the authors studied the model checking problem of population protocols against linear-time specifications. Population protocols form a specific class of parametric systems, 
in which a set of identical and anonymous
finite-state processes interact pairwise through rendezvous synchronization. The authors obtained a decidability result for the qualitative problem i.e., deciding
if a linear temporal logic formula holds with probability $1$, while undecidability was proved for the quantitative problem that is deciding if the property holds
with at least a given probability. Although the work of \cite{Es:Po} studied the quantitative model checking problem in the parameterized setting, it did not actually modelled the quantitative aspects of parametric communication that we address in our framework. Moreover, in contrast to our approach the systems' architecture was not considered 
in the design process of \cite{Es:Po}. As result, the order restrictions and recursion in parametric communication was not addressed in the modelling of the protocols' topologies.

In \cite{Be:Pa, Fo:Th} the authors considered broadcast communication and clique's topology for networks
of many identical probabilistic timed processes, where the number of processes was a parameter.
Then, they investigated the decidability results of several qualitative parameterized verification problems
for probabilistic timed protocols.
In the subsequent work of \cite{An:Pa} the authors extended broadcast protocols
and parametric timed automata and introduced a model of 
parametric timed broadcast
protocols with two different types of parameters,
namely the number of identical processes and the timing features. Parametric communication in \cite{An:Pa} was defined by clique's
semantics where every
message reached every process, and by reconfigurable semantics where
the set of receivers was chosen non-deterministically. Then the
decidability of reachability problems were studied for parametric timed broadcast
protocols in \cite{An:Pa}. The main difference of our contribution and the work of \cite{An:Pa,Be:Pa, Fo:Th}
is that we introduce a logic-based approach for the quantitative modelling of arbitrary parametric architectures. Also, the topologies of the protocols studied in 
the aforementioned work did not consider recursion. Moreover, we investigate the modelling problem of more complicated parametric
 architectures such as Publish/Subscribe or Request/Response. On the other hand, the authors of  \cite{An:Pa,Be:Pa, Fo:Th} obtained several nice verification results
for concrete parameterized protocols. Applying our framework for the study of parametric verification is left as a future work direction.

Finally, according to our best knowledge the only weighted logic-based approach for the modelling of architectures is found in the recent work of \cite{Ka:We,Pa:On,Pa:We}.
Specifically, in \cite{Pa:On,Pa:We} the authors studied PCL in the weighted setup for modelling the quantitative aspects of architectures. The formulas of the resulting logic, namely weighted PCL (wPCL for short), were interpreted as polynomials with values over a commutative semiring.
 Soundness was proved for that logic and the authors obtained the decidability for the equivalence problem of its formulas.
In \cite{Ka:We} the authors extended the work of \cite{Pa:On,Pa:We} and studied weighted PCL over a
product valuation monoid. That algebraic structure allowed to compute the average cost as well as the maximum
cost among all costs occurring most frequently for executing the interactions within architectures.
The authors applied that logic to model
several weighted software architectures and proved that the equivalence problem of its formulas is decidable. In contrast to our setting, the work of \cite{Ka:We,Pa:On, Pa:We} described architectures without addressing ordered or recursive interactions. Parametric weighted architectures were considered only in \cite{Pa:On} using the weighted first-order level of weighted PCL, namely weighted FOCL. Nevertheless, weighted FOCL considered no execution order of the interactions of parametric weighted architectures.

\section{Preliminaries}

\subsection{Notations}
For every natural number $n \geq 1$ we denote by $[n]$ the set $\{1, \ldots, n \}$. Hence, in the sequel, whenever we use the notation $[n]$ we always assume that $n \geq 1$. For every set $S$ we denote by $\mathcal{P}(S)$ the powerset of $S$. Let $A$ be an alphabet, i.e., a finite nonempty set. As usual we denote by $A^*$ the set of all finite words over $A$ and we let $A^+ = A^* \setminus \{\varepsilon\}$ where $\varepsilon$ denotes the empty word. For every word $w\in A^{\ast}$ we let $\vert w\vert$ denote the length of $w$, i.e., the number of letters comprising $w$. Given two words $w, u \in A^*$, the shuffle product $w \shuffle u$ of $w$ and $u$ is a language over $A$ defined by
$$w \shuffle u= \{w_1u_1 \ldots w_mu_m \mid w_1, \ldots, w_m, u_1, \ldots , u_m \in A^* \text{ and } w=w_1 \ldots w_m, u=u_1\ldots u_m \}.$$

 Let $L_1, L_2, L \subseteq A^*$ be languages over $A$. Then, the Cauchy product $L_1 * L_2$ of $L_1$ and $L_2$, the shuffle product $L_1 \shuffle L_2$ of $L_1$ and $L_2$, and the iteration $L^+$ of $L$ are defined respectively, by:
\begin{itemize}
\item[-] $L_1 * L_2 = \{w_1w_2 \mid w_1 \in L_1, w_2 \in L_2 \}$,
  
\item[-] $L_1 \shuffle L_2 = \bigcup_{w_1 \in L_1, w_2 \in L_2}w_1\shuffle w_2$,

\item[-] $L^+= \bigcup_{\nu\geq 1}L^{\nu}$, \ \ where $L^1=L$, $L^{\nu+1}=L^{\nu} * L$, and $\nu$ is a natural number.
\end{itemize}

\noindent Shuffle operation has been widely studied in formal languages and in concurrency theory
in order to model parallel composition semantics of concurrent processes \cite{Ac:Sa,Re:Sh}. 
According to the algebraic definition, the shuffle operator returns the
set of all possible interleavings of the corresponding objects. Similarly to the work of \cite{Ch:Pa,Hu:Fo}, the motivation to introduce a shuffle type of operator
and quantifier in our logics lies in encoding any possible order (interleaving) of sequences of interactions within architectures. 
The study of parametric systems with concurrent semantics is not considered here, though it could be investigated in future work.

\subsection{Semirings}
A \emph{semiring} $(K,+,\cdot,0,1)$\emph{ }consists of a set
$K,$ two binary operations $+$ and $\cdot$ and two constant elements $0$
and $1$ such that $(K,+,0)$ is a commutative monoid, $(K,\cdot,1)$ is a
monoid, multiplication $\cdot$ distributes over addition $+$, and $0 \cdot k=k\cdot 0=0$
for every $k\in K$. If the monoid $(K,\cdot,1)$ is commutative, then the
semiring is called commutative.\ The semiring is denoted simply by $K$ if the
operations and the constant elements are understood. The result of the empty
product as usual equals to $1$. If no confusion arises, we denote sometimes in the sequel the multiplication operation $\cdot$ just by juxtaposition. The following algebraic structures are well-known commutative semirings. \begin{itemize}
\item The semiring $(\mathbb{N},+,\cdot,0,1)$  of natural numbers,
\item the semiring  $(\mathbb{Q},+,\cdot,0,1)$ of rational numbers,
\item the semiring  $(\mathbb{R},+,\cdot,0,1)$ of real numbers,
\item the Boolean semiring $B=(\{0,1\},+,\cdot,0,1)$, 
\item the arctical or max-plus semiring $\mathbb{R}_{\max}=(\mathbb{R}_{+}\cup\{-\infty\},\max,+,-\infty,0)$ where $\mathbb{R}_{+}=\{r\in\mathbb{R}\mid r\geq0\}$,
\item the tropical or min-plus semiring $\mathbb{R}_{\min}=(\mathbb{R}_{+}\cup\{\infty\},\min,+,\infty,0)$,
\item the Viterbi semiring $(\left[  0,1\right]  ,\max,\cdot,0,1)$ used in probability theory, 
\item every bounded distributive lattice with the operations $\sup$ and $\inf$, in particular the fuzzy semiring $F=([0,1],\max,\min,0,1)$.
\end{itemize}

\noindent A semiring $(K,+,\cdot , 0,1)$ is called a \emph{skew field} if the monoids $(K,+,0)$ and $(K\setminus \{0\}, \cdot, 1)$ are groups. For instance, $(\mathbb{Q}, +, \cdot,0, 1)$ and  $(\mathbb{R}, +, \cdot, 0, 1)$ are skew fields, and more generally every field is a skew field.

 A  \emph{formal series} (or simply \emph{series}) \emph{over}
$A^*$ \emph{and} $K$ is a mapping $s:A^*\rightarrow K$. The \emph{support of} $s$ is the set $\mathrm{supp}(s)=\{w \in A^* \mid s(w) \neq 0  \}$. A series with finite support is called  a \emph{polynomial}. The \emph{constant series}
$\widetilde{k}$ ($k\in K$) is defined, for every $w\in A^*$,$\ $by
$\widetilde{k}(w)=k$. We denote by $K\left\langle \left\langle A^*\right\rangle
\right\rangle $ the class of all series over $A^*$ and $K$, and by $K\left\langle  A^*
\right\rangle $ the class of all polynomials over $A^*$ and $K$. 
Let $s,r\in K\left\langle \left\langle A^*\right\rangle \right\rangle $ and
$k\in K$. The \emph{sum} $s\oplus r$, the \emph{products with scalars} $ks$
and $sk$, and the\emph{ Hadamard product} $s \otimes r$\ are series in $K\left\langle \left\langle A^*\right\rangle \right\rangle $ and are defined
elementwise, respectively  by $ s \oplus r(w)=s(w) + r(w), (ks)(w)=k \cdot s(w), (sk)(w)=s(w)\cdot k, s\otimes r(w)=s(w)\cdot r(w)$  for every
$w\in A^*$. It is a folklore result that the structure $\left(  K\left\langle
\left\langle A^*\right\rangle \right\rangle ,\oplus,\otimes,\widetilde
{0},\widetilde{1}\right)$  is a semiring. Moreover, if $K$ is commutative,
then $\left(  K\left\langle \left\langle A^*\right\rangle \right\rangle
\oplus,\otimes,\widetilde{0},\widetilde{1}\right)  $ is also commutative. The \emph{Cauchy product }$s\odot
r\in K\left\langle \left\langle A^{\ast}\right\rangle \right\rangle $ is determined by $s\odot r(w)=\sum\nolimits_{w=w_1w_2}s(w_1)r(w_2)$ for every $w\in A^{\ast}$,  whereas the \emph{shuffle product} $s\varpi r \in K\left\langle \left\langle A^{\ast}\right\rangle \right\rangle $ is defined by $s\varpi r(w)=\sum\nolimits_{w\in w_1\shuffle w_2}s(w_1)r(w_2)$ for every $w\in A^{\ast}$. The $\nu$-th iteration $s^{\nu}\in  K\left\langle \left\langle A^*\right\rangle \right\rangle $ ($\nu\geq 1$) is defined inductively by $s^{1}=s$ and $s^{\nu+1}=s^{\nu} \odot s$, for every natural number $\nu\geq 1$. A series $s$ is called proper if $s(\varepsilon)=0$. Then, the iteration $s^{\oplus}$ of $s$ is defined by $s^{\oplus}=\sum\limits_{\nu\geq 1} s^{\nu}$.  If $s$ is proper, then for every $w\in A^{+}$ and $\nu> \vert w\vert$ we have $s^{\nu}(w)=0$. Hence, $s^{\oplus}(w)=\sum\limits_{1\leq \nu \leq \vert w\vert} s^{\nu}(w)$ for every $w\in A^{+}$.

\begin{quote}
\emph{Throughout the paper }$(K, +, \cdot, 0,1)$\emph{ will denote a commutative semiring.}
\end{quote}

\subsection{Component-based systems}
\label{comp_based_system}
In this subsection we deal with component-based systems comprised of a finite number of components of the same or different type. In our setup, components  are defined by labelled transition systems (LTS for short) like in several well-known component-based modelling frameworks including BIP \cite{Bl:Al}, REO \cite{Am:RE}, and B \cite{Al:Th}. Next  we use the terminology of BIP framework for the basic notions and definitions, though we focus only on the communication patterns of components building a component-based system. 
This is justified by the fact that the presented logic-based modelling framework for architectures does not require the knowledge of the actual transitions of the components. Indeed, communication among components is achieved through their corresponding interfaces. For an LTS, its interface is the associated set of labels, called ports. Then, communication of components is defined by interactions, i.e., nonempty sets of ports. Interactions can be represented by formulas of a propositional logic, namely 
 \emph{propositional interaction logic} (PIL for short), which has been used for developing first- and second-order logics applied for the modelling of parametric architectures \cite{Bo:St,Ko:Pa,Ma:Co}.

Let us firstly define the notion of atomic components. 

\begin{defi}\label{at_comp}
 An atomic component is an LTS $B=(Q,P,q_0,R)$ where $Q$ is a finite
set of states, $P$ is a finite set of ports, $q_0$ is the initial state and $R\subseteq Q \times P \times Q$ is the set of transitions.
\end{defi}

For simplicity, we assume in the above definition, that every port $p \in P$ occurs at most in one transition. Then, we use this condition in the weighted setup in Section \ref{weighted_part} in order to associate a unique weight with each port. 

In the sequel, we call an atomic component $B$ a \emph{component}, whenever we deal with several atomic components. For every set $\mathcal{B} = \{B(i) \mid   i \in [n] \} $ of components, with $B(i)=(Q(i),P(i),q_{0}(i),R(i))$, $i \in [n]$, we consider in the paper, we assume that  $(Q(i) \cup P(i))\cap (Q(i') \cup P(i')) = \emptyset $ 
for every $1 \leq i\neq i' \leq n$.

We now recall PIL. Let $P$ be a finite nonempty set of \emph{ports}. We let $I(P)=\mathcal{P}%
(P)\setminus\{\emptyset\}$ for the set of interactions over $P$ and  $\Gamma(P)=\mathcal{P}(I(P))\setminus \{\emptyset\}$. Then, the syntax of PIL formulas $\phi$ over $P$\ is given by the grammar
$$
\phi::=\mathrm{true}\mid p\mid\neg\phi\mid\phi\vee\phi 
$$
where $p\in P$. 

We set $\neg(\neg\phi)=\phi$ for every
PIL formula $\phi$ and $\mathrm{false}=\neg\mathrm{true}$. As usual the conjunction and the implication 
of two PIL formulas $\phi,\phi^{\prime}$ over $P$ are defined respectively, by $\phi\wedge
\phi^{\prime}:=\neg{\left(  \neg\phi\vee\neg\phi^{\prime}\right)  }$ and $\phi \rightarrow \phi': = \neg \phi \vee \phi'$.  
PIL formulas are interpreted over interactions in $I(P)$. More precisely, for every PIL formula  $\phi$ and $a \in I(P)$ we define the satisfaction relation 
$a\models_{\mathrm{PIL}}\phi$ by induction on the structure of $\phi$ as follows:\hfill
\begin{itemize}
\item[-]  $a\models_{\mathrm{PIL}} \mathrm{true}$,

\item[-]  $a\models_{\mathrm{PIL}} p$  \ \ iff \ \  $p \in a$,

\item[-]  $a\models_{\mathrm{PIL}} \neg\phi$ \ \ iff \ \ $a \not \models_{\mathrm{PIL}} \phi$,

\item[-]   $a\models_{\mathrm{PIL}} \phi_1 \vee \phi_2$ \ \ iff \ \ $a \models_{\mathrm{PIL}} \phi_1$ or $a \models_{\mathrm{PIL}} \phi_2$.
\end{itemize}

\noindent Note that PIL differs from propositional logic,   since it is interpreted over 
interactions, and thus the name ``interaction" is assigned to it. 

Two PIL formulas $\phi, \phi'$ are called equivalent, and we denote it by $\phi \equiv \phi'$, whenever $a \models_{\mathrm{PIL}} \phi$ iff $a \models_{\mathrm{PIL}} \phi'$ for every $a \in I(P)$. A PIL formula $\phi$ is called a \emph{monomial over} $P$ if it is of the form $p_1\wedge \ldots \wedge  p_l$, where $l\geq 1$ and $p_{\lambda}\in P$ or $p_{\lambda}=\neg p'_{\lambda}$ with $p'_{\lambda}\in P$, for every $\lambda\in [l]$. 
For every interaction $a=\lbrace p_1,\ldots,p_l\rbrace \in I(P)$ we consider the monomial $\phi_a =p_1\wedge \ldots \wedge  p_l$. Then, it trivially holds $a\models_{\mathrm{PIL}} \phi_a$, and  for every $a, a' \in I(P)$ we get $a=a'$ iff $\phi_{a} \equiv \phi_{a'} $. We can describe a set of interactions as a disjunction of PIL formulas. More precisely, let $\gamma=\lbrace a_1,\ldots, a_m\rbrace\in \Gamma(P)$, where $a_{\mu}=\left \{ p_1^{(\mu)},\ldots,p^{(\mu)}_{l_{\mu}}\right \} \in I(P)$ for every $\mu\in [m]$. Then, the PIL formula $\phi_{\gamma} $ of $\gamma$ is $\phi_{\gamma}=\phi_{a_1}\vee \ldots \vee \phi_{a_m}$, i.e., 
$\phi_{\gamma}=\bigvee\limits_{\mu\in [m]}\bigwedge\limits_{\lambda\in [l_{\mu}]}p_{\lambda}^{(\mu)}$.

We can now define component-based systems. For this, let $\mathcal{B} = \{B(i) \mid i \in [n]   \} $ be a set of components where $B(i)=(Q(i),P(i),q_{0}(i),R(i))$, for $i \in [n]$. We denote with $P_{\mathcal{B}} =\bigcup_{i \in [n]}  P(i)$ the set comprising all ports of the elements of $\mathcal{B}$. Then an \emph{interaction of} $\B$ is an  interaction $a \in I(P_{\mathcal{B}})$ such that  $\vert a \cap P(i)\vert \leq 1$, for every $i \in [n]$. We denote by $I_{\mathcal{B}}$ the set of all interactions of $\mathcal{B}$, i.e., 
$$I_{\mathcal{B}} = \left \{ a \in I(P_{\mathcal{B}}) \mid   \vert a \cap P(i)\vert \leq 1 \text{ for every } i \in [n] \right\},$$ 
and let $\Gamma_{\mathcal{B}}= \mathcal{P}(I_{\mathcal{B}})\setminus \{ \emptyset\} $.

\begin{defi} \label{BIP-def}A component-based system is a pair $(\B, \gamma)$ where $\B=\{B(i) \mid i \in [n]  \}$ is a set of components, with $B(i)=(Q(i),P(i),q_{0}(i),R(i))$ for every $i \in [n]$, and $\gamma$  is a set of interactions in $I_{\B}$.
\end{defi}

The set $\gamma$ of interactions of a component-based system $(\B,\gamma)$ specifies the topology with which the components are connected in the system, i.e., the architecture of the system. Due to discussion before Definition \ref{BIP-def} we can replace the set of interactions $\gamma$ by its corresponding PIL formula $\phi_{\gamma}$, i.e., in a logical directed notation. Expression of software architectures by logics has been considered in several work and gave nice results (cf. for instance \cite{Bo:St,Ko:Pa,Ma:Co}).

\section{Extended propositional interaction logic}
One of the most important characteristics of component-based systems is the architecture which specifies the coordination primitives of the connected components. Although PIL can describe nicely several architectures, it does not capture an important feature of more complicated ones: the specified order required for the execution of the interactions. Such architectures, with an increased interest in applications, are for instance the Request/Response and Publish/Subscribe. In \cite{Pi:Ar} we introduced a propositional logic that extends 
PIL by equipping it with two operators, namely  the concatenation $*$ and the shuffle operator $\shuffle$. 
With that logic we succeeded to represent architectures of component-based systems where the order of the interactions is involved.
Specifically, 
the concatenation operator $*$ modelled sequential interactions and the shuffle operator $\shuffle$ interactions executed with interleaving. 
Here, we further augment PIL with
an iteration operator $^{+}$ that serves to encode sequential (at least one) repetition of interactions. This in turn implies, that the resulting logic
allows to describe the ongoing implementation of an architecture during the system's operation. Moreover, our propositional logic
with its first-order level is proved to be a sufficient modelling language
for the symbolic representation of architectures of parametric component-based systems.

\begin{defi}\label{syn-epil}
Let $P$ be a finite set of ports. The syntax of extended propositional interaction logic
(EPIL for short) formulas $\varphi$ over $P$\ is given by the grammar
\begin{align*}
\zeta & ::=\phi \mid \zeta * \zeta \mid \zeta \shuffle \zeta\\
\varphi & ::=\zeta \mid \neg \zeta \mid \varphi\vee\varphi \mid \varphi \wedge \varphi \mid \varphi * \varphi
\mid \varphi \shuffle \varphi \mid \varphi^+
\end{align*}
where $\phi $ is a PIL formula over $P$, $*$ is the concatenation operator,  $\shuffle$ is the shuffle operator, and $^{+}$ is the iteration operator. 
\end{defi}

The binding strength, in decreasing order, of the operators in EPIL is the following: negation, iteration, shuffle, concatenation, conjunction, and disjunction.  The reader should notice that we consider a restricted  use of  negation in the syntax of EPIL formulas. Specifically, negation is permitted in PIL formulas and EPIL formulas of type $\zeta$. The latter will ensure exclusion of erroneous interactions in architectures.  The restricted use of negation has no impact to description of models characterized by EPIL formulas since most of  known architectures can be described by formulas in our EPIL. Furthermore, it contributes to a reasonable complexity of translation of first-order extended interaction logic formulas to finite automata. This in turn implies the decidability of equivalence, satisfiability, and validity of first-order extended interaction logic sentences (cf. Section \ref{sec_dec}).
  
For the satisfaction of EPIL formulas we consider finite words $w$ over $I(P)$. Intuitively, a word $w$ encodes each of the distinct interactions within a system as a letter. Moreover, the position of each letter in $w$ depicts the order in which the corresponding interaction is executed in the system, in case there is an order restriction.  For the semantics of EPIL formulas, we introduce the subsequent notation. For every EPIL formula $\varphi$ over $P$ and natural number $\nu \geq 1$ we define the EPIL formula $\varphi^{\nu}$ by induction on $\nu$ as follows:
\begin{itemize}
\item[-] $\varphi^1=\varphi$,
\item[-] $\varphi^{\nu+1}=\varphi^{\nu} * \varphi$. 
\end{itemize}

\begin{defi}\label{sem-epil}
Let $\varphi$ be an EPIL formula over $P$ and $w \in I(P)^+$. Then we define the satisfaction relation 
$w\models\varphi$ by induction on the structure of $\varphi$ as follows:
\begin{itemize}
\item[-] $w\models \phi$  \ \ iff \ \ $ w \models_{\mathrm{PIL}} \phi$,

\item[-] $w \models \zeta_1 * \zeta_2$  \ \ iff \ \ there exist $w_1,w_2 \in I(P)^+$ such that  $w=w_1w_2$ and $w_i \models\zeta_i$ for $i=1,2$,

\item[-] $w \models \zeta_1 \shuffle \zeta_2$  \ \ iff \ \ there exist $w_1,w_2 \in I(P)^+$ such that  $w\in w_1 \shuffle w_2$ and $w_i \models\zeta_i$ for $i=1,2$,

\item[-] $w \models \neg \zeta$ \ \ iff \ \ $w \not \models \zeta$,

\item[-] $w\models \varphi_1 \vee \varphi_2$ \ \ iff \ \ $w \models \varphi_1$ or $w \models \varphi_2$,

\item[-] $w\models \varphi_1 \wedge \varphi_2$ \ \ iff \ \ $w \models \varphi_1$ and $w \models \varphi_2$,

\item[-] $w \models \varphi_1 * \varphi_2$  \ \ iff \ \ there exist $w_1,w_2 \in I(P)^+$ such that  $w=w_1w_2$ and $w_i \models\varphi_i$ for $i=1,2$,

\item[-] $w \models \varphi_1 \shuffle \varphi_2$  \ \ iff \ \ there exist $w_1,w_2 \in I(P)^+$ such that  $w\in w_1\shuffle w_2$ and $w_i \models\varphi_i$ for $i=1,2$,

\item[-] $w \models \varphi^+ $  \ \ iff \ \ there exists $\nu \geq 1$ such that  $w \models\varphi^{\nu}$.

\end{itemize}
\end{defi}

\noindent The empty word is not included in the semantics of EPIL since in our framework we do not consider architectures with no interactions.

If $\varphi=\phi$ is a PIL formula, then $w \models \phi$ implies that $w$ is a letter in $I(P)$. Two EPIL formulas $\varphi, \varphi'$ are called \emph{equivalent}, and we denote it by $\varphi \equiv \varphi' $, whenever $w \models \varphi$ iff $w \models \varphi'$ for every $w \in I(P)^+$.

It can be easily seen that the concatenation operator  is not commutative.  The proof of the next proposition is straightforward.

\begin{prop}
\label{epil_properties}
Let $\varphi,\varphi_1, \varphi_2, \varphi_3$ be \emph{EPIL} formulas over $P$. Then,
\begin{enumerate}[(i)]
\item $\varphi_1 * (\varphi_2 * \varphi_3)\equiv(\varphi_1 * \varphi_2) * \varphi_3$, \label{epil_assoc}
\item $\varphi * (\varphi_1\vee \varphi_2)\equiv(\varphi * \varphi_1)\vee (\varphi * \varphi_2)$, 
\item $(\varphi_1\vee \varphi_2) * \varphi\equiv(\varphi_1 * \varphi)\vee (\varphi_2 * \varphi)$.\label{epil_distr_dis}
\end{enumerate}
\end{prop}

\noindent Now, we define an updated version of component-based systems where in comparison to the one in Definition \ref{BIP-def}, we replace the set of interactions $\gamma$ (equivalently its corresponding PIL formula $\phi_{\gamma}$ representing them) by an EPIL formula.

\begin{defi} \label{BIP_upd-def}A component-based system  is a pair $(\B, \varphi)$ where $\B=\{B(i) \mid i \in [n]  \}$ is a set of components and $\varphi$  is an EPIL formula over $P_{\B}$.
\end{defi}
We should note that the EPIL formula $\varphi$ in the above definition is defined over $P_{\B}$. Nevertheless, we will be interested in words $w$ of interactions in $I_{\B}$ satisfying $\varphi$.
  
Observe that in this work we develop no theory about the computation of the semantics of component-based systems, a problem that will be studied in future work. Though it should be clear that we specify how we model the components of  systems in order to formalize the corresponding architectures. 

\subsection{Examples of architectures described by EPIL formulas} 
\label{examples_EPIL}   
Next we present three examples of component-based models whose architectures have ordered interactions encoded by EPIL formulas. Clearly, there exist several variations of the following architectures and their order restrictions, that EPIL formulas could also model sufficiently by applying relevant modifications. We need to define the following macro EPIL formula. Let $P=\{p_1, \ldots , p_n\}$ be a set of ports. Then, for  $p_{i_1}, \ldots , p_{i_m} \in P$ with $m <n$ we let 
$$\#(p_{i_1} \wedge \ldots \wedge p_{i_m})::=p_{i_1}\wedge \ldots \wedge p_{i_m} \wedge \bigwedge_{p \in P \setminus \{p_{i_1}, \ldots, p_{i_m}\}}\neg p.$$

\begin{exa}\label{ex_b_blackboard}\textbf{(Blackboard)}
We consider a component-based system $(\B,\varphi)$ with the Blackboard architecture. Blackboard architecture is applied to multi-agent systems 
for solving problems with nondeterministic strategies that result from multiple partial solutions (cf. Chapter 2 in \cite{Bu:Pa},\cite{Co:Bl}).
Applications based on a Blackboard architecture, 
include 
planning and scheduling (cf. \cite{St:Cr}) as well as 
artificial intelligence \cite{Ni:Bl}. Blackboard architecture involves three component types, one blackboard component, one controller
component and the knowledge sources components \cite{Bu:Pa,Co:Bl,Ni:Bl}. Blackboard is a global data store that presents the state of the problem to be solved. 
Knowledge sources, simply called sources, are expertised agents that provide 
partial solutions to the given problem. Knowledge sources are independent and do not know about the existence of other sources. If
there is sufficient information for a source to provide its partial solution, then the
source is triggered i.e., is keen to write on the blackboard. Since multiple sources are triggered and compete to provide their solutions, a controller component is used
to resolve any conflicts. Controller accesses both the blackboard to inspect the available data and the sources to schedule them so that they execute their solutions on blackboard.

For our example we consider three knowledge sources components. 
Therefore, we have  $\B=\lbrace B(i)\mid i\in [5]\rbrace$ where $B(1), B(2),B(3),B(4), B(5)$ refer to  blackboard,  controller and the three sources
components, respectively. The set of ports of each component is $P(1)=\lbrace p_d,p_a\rbrace,
 P(2)=\lbrace p_r,p_l,p_e\rbrace,$ $P(3)=\lbrace p_{n_1},p_{t_1},p_{w_1} \rbrace$, $P(4)=\lbrace p_{n_2},p_{t_2},p_{w_2} \rbrace$, and $P(5)=\lbrace p_{n_3},p_{t_3},p_{w_3} \rbrace$. 
 Figure \ref{b_blackboard} depicts an instantiation of the permissible interactions 
 among the five components of our component-based system. Blackboard has two ports $p_d,p_a$ to declare the state of the problem and add the new data as obtained by a knowledge source, respectively.
Knowledge sources  have three ports $p_{n_{\kappa}},p_{t_{\kappa}},p_{w_{\kappa}}$, for $\kappa=1,2,3$, for being notified about the existing data
on blackboard, the trigger of the source, and for writing the partial solution on blackboard, respectively. 
Controller has three ports, namely $p_r$ used to record blackboard data, $p_l$ for the log process of
 triggered sources, and $p_e$ for their execution on blackboard. Here we assume that all knowledge sources are triggered, i.e., that all available sources participate in the architecture. The permissible interactions in the architecture range over $I_{\B}$ and are obtained as follows:
\begin{itemize}
\item The sets $\lbrace p_d,p_r \rbrace$ and $\lbrace  p_d,p_{n_{\kappa}} \rbrace$, for $\kappa=1,2,3$, capture the interactions of blackboard with controller and the three sources, respectively, for presenting the state of the problem. The corresponding connections are shown in Figure \ref{b_blackboard} by the four black lines. 
\item The sets $ \lbrace p_l, p_{t_{\kappa}} \rbrace$, for $\kappa=1,2,3$, refer to the log process of
 triggered sources in the controller. The orange lines in Figure \ref{b_blackboard} depict a possible scenario of such connections between the controller and each of the sources $B(3)$ and $B(5)$.
 \item Finally, the sets $ \lbrace p_e, p_{w_{\kappa}},p_a \rbrace$, for ${\kappa}=1,2,3$, refer to the connection of the sources with blackboard through its controller for adding the new data. These interactions are shown
 with the green lines in Figure \ref{b_blackboard}, for the case that sources $B(3)$ and $B(5)$ have been triggered. 
\end{itemize}

\noindent The EPIL formula $\varphi$ describing the ordered and recursive interactions of Blackboard architecture is

\begin{multline*}\varphi=\Bigg(\#(p_d\wedge p_r) *  \bigg(\#(p_d\wedge p_{n_1})\shuffle \#(p_d\wedge p_{n_2})\shuffle \#(p_d\wedge p_{n_3}) \bigg) * \\ \bigg(\varphi_1\vee \varphi_2\vee \varphi_3\vee(\varphi_1 \shuffle \varphi_2) \vee (\varphi_1\shuffle \varphi_3)\vee (\varphi_2\shuffle \varphi_3)\vee (\varphi_1\shuffle\varphi_2\shuffle\varphi_3)\bigg)^+ \Bigg)^+
\end{multline*}
where 
$$\varphi_i= \#(p_l\wedge p_{t_i}) * \#(p_e\wedge p_{w_i}\wedge p_a)$$
for $i=1,2,3$. The first PIL subformula encodes the connection between blackboard and controller. The EPIL subformula between the two concatenation operators represents the connections of the three knowledge sources with blackboard in order to be informed for existing data. The last part of $\varphi$ captures the connection of some of the three knowledge sources with controller and blackboard for the triggering and the writing process. The use of shuffle operator in $\varphi$ serves for capturing any possible order, among the sources, for implementing the corresponding connections with controller and blackboard. On the other hand, the two concatenation operators in the first line of $\varphi$ ensure respectively that the sources are informed after the controller for the data on blackboard and before the triggering and writing process. The use of the inner iteration operator allows the repetition of the triggering and writing process of some of the sources based on the blackboard's
data at one concrete point. Moreover, the outer iteration operator describes the recursive permissible interactions in the architecture at several points. This implies that as new data are added on blackboard, then the controller and the sources are informed anew in each communication `cycle' for these data, and then the sources decide whether they are able to write on blackboard or not. Therefore, combining the two iteration operators we achieve to describe the subsequent implementation of
Blackboard architecture in the given component-based system $\B$.

\definecolor{harlequin}{rgb}{0.25, 1.0, 0.0}
\definecolor{ao}{rgb}{0.0, 0.0, 1.0}

\definecolor{darkorange}{rgb}{1.0, 0.55, 0.0}
\definecolor{harlequin}{rgb}{0.25, 1.0, 0.0}
\definecolor{ao}{rgb}{0.0, 0.0, 1.0}

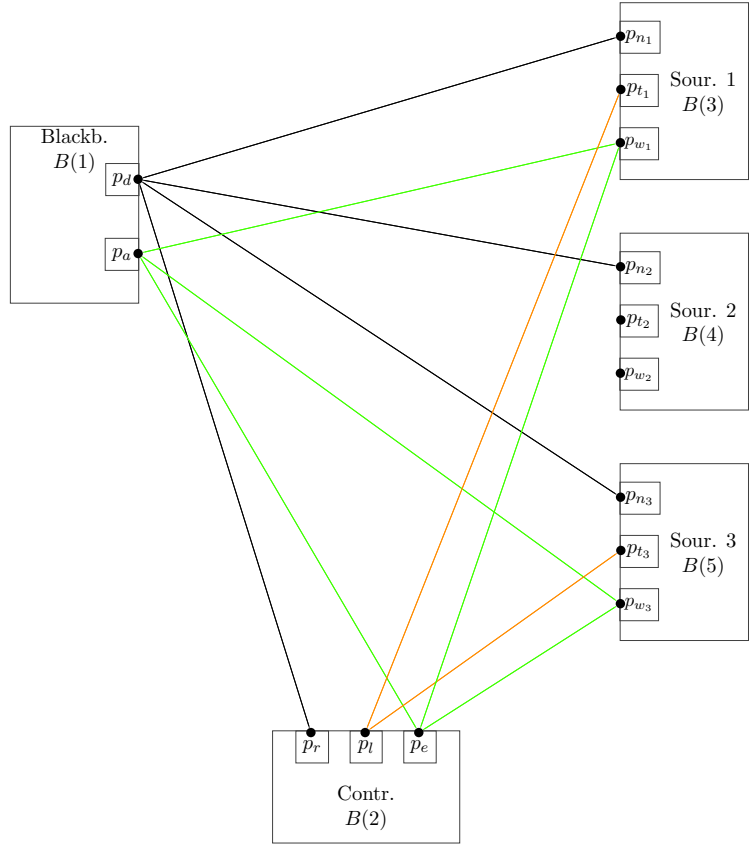
\begin{figure}[h]
\centering
\resizebox{0.65\linewidth}{!}{
\begin{tikzpicture}[>=stealth',shorten >=1pt,auto,node distance=1cm,baseline=(current bounding box.north)]
\tikzstyle{component}=[rectangle,ultra thin,draw=black!75,align=center,inner sep=9pt,minimum size=1.5cm,minimum height=3.31cm,minimum width=2.4cm]
\tikzstyle{port}=[rectangle,ultra thin,draw=black!75,minimum size=6mm]
\tikzstyle{bubble} = [fill,shape=circle,minimum size=5pt,inner sep=0pt]
\tikzstyle{type} = [draw=none,fill=none] 

\node [component,align=center] (a1)  {};
\node [port] (a2) [above right= -1.3cm and -0.62cm of a1]  {$p_d$};
\node[bubble] (a3) [above right=-1.06cm and -0.08cm of a1]   {};

\node [port] (a4) [above right= -2.7cm and -0.63cm of a1]  {$p_a$};
\node[bubble] (a5) [above right=-2.45cm and -0.08cm of a1]   {};

\node[type] (a6) [above=-0.45cm of a1]{Blackb.};
\node            [below=-0.1cm of a6]{$B(1)$};

\tikzstyle{control}=[rectangle,ultra thin,draw=black!75,align=center,inner sep=9pt,minimum size=1.5cm,minimum height=2.1cm,minimum width=3.5cm]

\node [control,align=center] (b1)  [below right=8cm and 2.5cm of a1]{};
\node [port] (b2) [above left=-0.605cm  and -1.05 of b1]  {$p_r$};
\node[bubble] (b3) [above left=-0.1cm and -0.35cm of b2]   {};

\node [port] (b4) [right=0.4 of b2]  {$p_l$};
\node[bubble] (b5) [above left=-0.1cm and -0.35cm of b4]   {};

\node [port] (b6) [right=0.4 of b4]  {$p_e$};
\node[bubble] (b7) [above left=-0.1cm and -0.35cm of b6]   {};

\node[type] (b8) [below=-1.2cm of b1]{Contr.};
\node[]      [below=-0.1cm of b8]{$B(2)$};

\tikzstyle{source}=[rectangle,ultra thin,draw=black!75,align=center,inner sep=9pt,minimum size=1.5cm,minimum height=3.31cm,minimum width=2.4cm]

\tikzstyle{port}=[rectangle,ultra thin,draw=black!75,minimum size=6mm,inner sep=3pt]

\node [source,align=center] (c1) [above right=-1cm and 9cm of a1] {};
\node [port] (c2) [above left= -0.95cm and -0.75cm of c1]  {$p_{n_1}$};
\node[bubble] (c3) [above left=-0.35cm and -0.08cm of c2]   {};

\tikzstyle{port}=[rectangle,ultra thin,draw=black!75,minimum size=6mm,inner sep=3.5pt]
\node [port] (c4) [above left= -1.6cm and -0.72cm of c2]  {$p_{t_1}$};
\node[bubble] (c5) [above left=-0.35cm and -0.08cm of c4]   {};

\tikzstyle{port}=[rectangle,ultra thin,draw=black!75,minimum size=6mm,inner sep=2.7pt]
\node [port] (c6) [above left= -1.6cm and -0.72cm of c4]  {{\small $p_{w_1}$}};
\node[bubble] (c7) [above left=-0.35cm and -0.08cm of c6]   {};

\node[type,align=center] (c8) [above right=-1.7cm and -1.6cm of c1]{Sour. $1$};
\node[] [below=-0.1cm of c8]{$B(3)$};

\tikzstyle{port}=[rectangle,ultra thin,draw=black!75,minimum size=6mm,inner sep=3pt]
\node [source,align=center] (d1) [below= 1cm of c1] {};
\node [port] (d2) [above left= -0.95cm and -0.75cm of d1]  {$p_{n_2}$};
\node[bubble] (d3) [above left=-0.35cm and -0.08cm of d2]   {};

\tikzstyle{port}=[rectangle,ultra thin,draw=black!75,minimum size=6mm,inner sep=3.5pt]
\node [port] (d4) [above left= -1.6cm and -0.72cm of d2]  {$p_{t_2}$};
\node[bubble] (d5) [above left=-0.35cm and -0.08cm of d4]   {};

\tikzstyle{port}=[rectangle,ultra thin,draw=black!75,minimum size=6mm,inner sep=2.7pt]
\node [port] (d6) [above left= -1.6cm and -0.72cm of d4]  {{\small $p_{w_2}$}};
\node[bubble] (d7) [above left=-0.35cm and -0.08cm of d6]   {};

\node[type] (d8) [above right=-1.7cm and -1.6cm of d1]{Sour. $2$};
\node[] [below=-0.1cm of d8]{$B(4)$};

\tikzstyle{port}=[rectangle,ultra thin,draw=black!75,minimum size=6mm,inner sep=3pt]
\node [source,align=center] (e1) [below= 1cm of d1] {};
\node [port] (e2) [above left= -0.95cm and -0.75cm of e1]  {$p_{n_3}$};
\node[bubble] (e3) [above left=-0.35cm and -0.08cm of e2]   {};

\tikzstyle{port}=[rectangle,ultra thin,draw=black!75,minimum size=6mm,inner sep=3.5pt]
\node [port] (e4) [above left= -1.6cm and -0.72cm of e2]  {$p_{t_3}$};
\node[bubble] (e5) [above left=-0.35cm and -0.08cm of e4]   {};

\tikzstyle{port}=[rectangle,ultra thin,draw=black!75,minimum size=6mm,inner sep=2.7pt]
\node [port] (e6) [above left= -1.6cm and -0.72cm of e4]  {{\small $p_{w_3}$}};
\node[bubble] (e7) [above left=-0.35cm and -0.08cm of e6]   {};

\node[type] (e8) [above right=-1.7cm and -1.6cm of e1]{Sour. $3$};
\node[] [below=-0.1cm of e8]{$B(5)$};
 
 
\path[-]          (a3)  edge                  node {} (c3);
 \path[-]          (a3)  edge                  node {} (d3);
 \path[-]          (a3)  edge                  node {} (e3);

 \path[-]          (a3)  edge                  node {} (b3);
  \path[-]          (b3)  edge                  node {} (a3);
 
 
 \path[-]          (c3)  edge                  node {} (a3);
 \path[-]          (d3)  edge                  node {} (a3);
 \path[-]          (e3)  edge                  node {} (a3);

 
 \path[-]          (a5)  edge  [harlequin]               node {} (c7);
 \path[-]          (a5)  edge  [harlequin]                node {} (e7);

 \path[-]          (a5)  edge  [harlequin]                node {} (b7);
 
 
 \path[-]          (c7)  edge  [harlequin]                node {} (a5);
 \path[-]          (e7)  edge   [harlequin]               node {} (a5);

 \path[-]          (b7)  edge   [harlequin]              node {} (a5);
 
 \path[-]          (b5)  edge  [darkorange]                node {} (c5);
   \path[-]          (b5)  edge [darkorange]                node {} (e5);
   
    \path[-]          (c5)  edge  [darkorange]                node {} (b5);
  \path[-]          (e5)  edge  [darkorange]               node {} (b5);

 \path[-]          (b7)  edge  [harlequin]               node {} (c7);
 \path[-]          (c7)  edge  [harlequin]               node {} (b7);

    \path[-]          (b7)  edge  [harlequin]               node {} (e7);
 \path[-]          (e7)  edge  [harlequin]               node {} (b7);
   
\end{tikzpicture}}
\caption{Blackboard architecture. A possible execution for the interactions.}
\label{b_blackboard}
\end{figure}
\end{exa}

Before our second example we show the expressive difference among EPIL and PCL formulas of \cite{Ma:Co}.
 
\begin{rem}
\label{remark}
Consider the Blackboard architecture presented in the previous example. Then the  corresponding  PCL formula $\rho$ (cf. \cite{Ma:Co}) describing that architecture is
\begin{multline*}
\rho= \#(p_d\wedge p_r) + \#(p_d\wedge p_{n_1}) + \#(p_d\wedge p_{n_2}) + \#(p_d\wedge p_{n_3}) + \\
 \big(\phi_1 \sqcup \phi_2 \sqcup \phi_3 \sqcup (\phi_1 + \phi_2) \sqcup (\phi_1 + \phi_3) \sqcup (\phi_2 + \phi_3) \sqcup  (\phi_1 + \phi_2 + \phi_3)\big) 
\end{multline*}
where $+$ denotes the coalescing operator, $\sqcup$ denotes the union operator, and 

$\phi_i = \#(p_l\wedge p_{t_i}) +   \#(p_e\wedge p_{w_i} \wedge p_a)$

\noindent for $i=1,2,3$.

\noindent Then, the PCL formula $\rho$ is interpreted over sets of interactions in $\mathcal{P}(I(P))\setminus \{\emptyset\}$ which trivially cannot express the required order for implementing the interactions. For instance the set of interactions 
$\big\{\{p_d, p_r\} , \{p_d, p_{n_1}\} , \{p_d , p_{n_2}\} , \{p_d ,  p_{n_3}\}, \{p_l, p_{t_2} \} ,   \{p_e , p_{w_2} , p_a \}  \big\}$ satisfies $\rho$ but represents no order of the interactions' execution.  
\end{rem}

Note that the first-order logics studied in \cite{Bo:Ch,Bo:St,Ko:Pa} were built over PIL. Although PIL served as the basis for developing logics that encode parametric architectures (still without addressing ordering constraints), it is far from describing complex architectures like Blackboard, Request/Response, or Publish/Subscribe.

\begin{exa} \textbf{(Request/Response)}
\label{b-r-r}
We consider a component-based system $(\B,\varphi)$ with the Request/Response architecture. These
architectures are classical interaction patterns and 
widely used for web services \cite{Da:Se}. A Request/Response architecture refers to clients and services. In order for a service to be made available
it should be enrolled in the service registry. The enrollment of a service in the registry allows service consumers, simply called clients,
to search the existing services. Once services are signed up, then clients search the corresponding registry and take the address of the services. Each 
client that is interested in a service sends a request to the service and waits until the service will respond. No other client can be connected
to the service until the response of the service to the client who sent the request will be completed. In \cite{Ma:Co} the authors described this process
by adding, for each service, a third component type called coordinator. Coordinator takes care that only one
client is connected to a service until the process among them is completed. 

The Request/Response architecture of our component-based system consists of four component types namely  service registry, service, 
client, and  coordinator (Figure \ref{b_r_r}). For our example we consider seven components, and specifically, the service registry, two services  with their associated coordinators, and two clients. Therefore, we have that $\B=\lbrace B(i)\mid i\in [7]\rbrace$ where $B(1),\ldots, B(7)$ refer to each of the aforementioned components,
respectively. The set of ports of each component is $P(1)=\lbrace p_e, p_u, p_t \rbrace$, $P(2)=\lbrace p_{r_{1}},p_{g_{1}}, p_{s_{1}}\rbrace$, $P(3)=\lbrace p_{r_{2}},p_{g_{2}}, p_{s_{2}}\rbrace$, $P(4)=\lbrace p_{m_{1}}, p_{a_{1}},p_{d_{1}}\rbrace$, $P(5)=\lbrace p_{m_{2}}, p_{a_{2}},p_{d_{2}}\rbrace$, $P(6)=\lbrace p_{l_{1}},p_{o_{1}},p_{n_{1}}, p_{q_{1}}, p_{c_1}\rbrace$, and $P(7)=\lbrace p_{l_{2}},p_{o_{2}},p_{n_{2}}, p_{q_{2}}, p_{c_2}\rbrace$. Figure \ref{b_r_r} depicts the permissible interactions 
 among the service registry, service $B(2)$ and the two clients  components of our component-based system. The corresponding interactions for service $B(3)$, 
 which have been omitted for simplicity, are derived similarly. Service registry has three ports denoted by $p_e$, $p_u$, and $p_t$ used for 
 connecting with the service for its enrollment, for authorizing the client to search for a service, and for transmitting the address (link) of the 
service to the client in order for the client to send then its request, respectively. Services have three ports $p_{r_{\kappa}},p_{g_{\kappa}}, p_{s_{\kappa}}$, for ${\kappa}=1,2$, which establish the connection to the service registry for the signing up of the service, and the connection to 
a client (via coordinator) for getting a request and responding, respectively.
Each client $\kappa$ has five ports denoted by $p_{l_{\kappa}},p_{o_{\kappa}},p_{n_{\kappa}}, p_{q_{\kappa}}$ and $p_{c_{\kappa}}$, for ${\kappa}=1,2$. The first two ports are used for connection with the 
service registry to look up the available services and for obtaining the address of the service. 
The latter three ports express the connection of the client to  coordinator, 
 to service (via coordinator) for sending the request, and to service (via coordinator) for collecting its response,
 respectively. Coordinators have three ports namely $p_{m_{\kappa}}, p_{a_{\kappa}},p_{d_{\kappa}}$, for ${\kappa}=1,2$. The first port controls that only
 one client is connected to a service. The second one is used for acknowledging that the connected client sends a request, and the
 third one disconnects the client when the service responds to the request. 
 
\noindent The allowed interactions in the architecture range over $I_{\B}$ and are obtained as follows:
\begin{itemize}
\item The sets $\lbrace p_e, p_{r_{\kappa}}\rbrace $ refer the enrollment of the two services in registry, while the sets $\lbrace p_{l_{\kappa}}, p_u \rbrace$ and $\lbrace p_{o_{\kappa}},p_t \rbrace$ express the interactions of the two clients with registry, for $\kappa=1,2$. The corresponding connections are shown in Figure \ref{b_r_r} with black and orange color for the services and clients, respectively. 
\item The sets $\lbrace p_{n_{\kappa}},p_{m_{\lambda}} \rbrace$, $\lbrace p_{q_{\kappa}}, p_{a_{\lambda}}, p_{g_{\lambda}}\rbrace$, $\lbrace p_{c_{\kappa}},p_{d_{\lambda}},p_{s_{\lambda}} \rbrace$ capture the interactions of each of the two clients with the two services through their coordinators, for $\kappa,\lambda=1,2$. The green lines in Figure \ref{b_r_r} depict a possible case for these connections, and specifically between the two client components and service $B(2)$ through its coordinator $B(4)$.
\end{itemize}
\noindent Then, the EPIL formula $\varphi$ describing the Request/Response architecture, with its permissible ordered and recursive interactions, is

\begin{multline*}
\varphi=\big(\#(p_e\wedge p_{r_1})\shuffle \#(p_e\wedge p_{r_2})\big) * \big(\xi_1 \shuffle \xi_2\big) * \\ \Bigg( \bigg(\varphi_{11}  \vee \varphi_{21} \vee (\varphi_{11} * \varphi_{21}) \vee (\varphi_{21} * \varphi_{11})\bigg)^+ \bigvee  \bigg(\varphi_{12}  \vee \varphi_{22} \vee (\varphi_{12} * \varphi_{22}) \vee (\varphi_{22} * \varphi_{12}) \bigg)^+ \bigvee \\ \bigg(\big(\varphi_{11}  \vee \varphi_{21} \vee (\varphi_{11} * \varphi_{21}) \vee (\varphi_{21} * \varphi_{11})\big)^+ \shuffle 
 \\  \big(\varphi_{12}  \vee \varphi_{22} \vee (\varphi_{12} * \varphi_{22}) \vee (\varphi_{22} * \varphi_{12}) \big)^+\bigg)^+\Bigg)^+
\end{multline*} 
where
\begin{itemize}
 \item[-]$\xi_1=\#(p_{l_1}\wedge p_u) * \#(p_{o_1}\wedge p_t)$
 \item[-]$\xi_2=\#(p_{l_2}\wedge p_u) * \#(p_{o_2}\wedge p_t)$
\end{itemize}
describe the interactions of the two clients with the service registry and
\begin{itemize} 
\item[-]$\varphi_{11}=\#(p_{n_1}\wedge p_{m_1}) * \#(p_{q_1}\wedge p_{a_1}\wedge p_{g_1})* \#(p_{c_1}\wedge p_{d_1}\wedge p_{s_1})$ 
\item[-] $\varphi_{12}=\#(p_{n_1}\wedge p_{m_2}) * \#(p_{q_1}\wedge p_{a_2}\wedge p_{g_2})* \#(p_{c_1}\wedge p_{d_2}\wedge p_{s_2})$ 
\item[-]$\varphi_{21}=\#(p_{n_2}\wedge p_{m_1}) * \#(p_{q_2}\wedge p_{a_1}\wedge p_{g_1})* \#(p_{c_2}\wedge p_{d_1}\wedge p_{s_1})$
\item[-] $\varphi_{22}=\#(p_{n_2}\wedge p_{m_2}) * \#(p_{q_2}\wedge p_{a_2}\wedge p_{g_2})* \#(p_{c_2}\wedge p_{d_2}\wedge p_{s_2})$
\end{itemize}

\noindent encode the connections of each of the two clients with the two services through their coordinators. 

The EPIL formula $\varphi$ is interpreted as follows. The two subformulas at the left of the first two concatenation operators encode the interleaving connections of the two services and the two clients with registry, respectively. Then, each of the three subformulas connected with the big disjunctions express that either one of the two clients or both of them (one at each time) are connected with the first service only, the second service only, or both of the services, respectively. Observe that in these subformulas we make use of the concatenation operator for the connection of the two clients with the same service.
 This results from the fact that the architecture requires that 
a unique client should be connected with each service through its coordinator. 
On the other hand, there is no restriction for the connection of a client with multiple services, and 
hence we use the shuffle operator to connect the corresponding subformulas at the third and fourth line of the EPIL formula $\varphi$. Then, the  first iteration operator in EPIL formula $\varphi$ allows the repetition of the interactions of only first, only  second, or both of
clients with the first service through its coordinator. Analogously, 
the next four iteration operators  encode the recursive interactions for some of the clients with  the second service or with both of the services (applied consecutively) through their coordinators. Finally, the  use of last iteration operator 
describes the repetition of the interactions for any of the aforementioned cases.

\definecolor{darkorange}{rgb}{1.0, 0.55, 0.0}
\definecolor{harlequin}{rgb}{0.25, 1.0, 0.0}
\definecolor{ao}{rgb}{0.0, 0.0, 1.0}

\begin{figure}[t]
\centering
\resizebox{1.0\linewidth}{!}{
\begin{tikzpicture}[>=stealth',shorten >=1pt,auto,node distance=3cm,baseline=(current bounding box.north)]
\tikzstyle{component}=[rectangle,ultra thin,draw=black!75,align=center,inner sep=9pt,minimum size=2.5cm,minimum height=3cm, minimum width=3.9cm]
\tikzstyle{port}=[rectangle,ultra thin,draw=black!75,minimum size=6mm]
\tikzstyle{bubble} = [fill,shape=circle,minimum size=5pt,inner sep=0pt]
\tikzstyle{type} = [draw=none,fill=none]

\node [component] (a1) {};

\tikzstyle{port}=[rectangle,ultra thin,draw=black!75,minimum size=6mm,inner sep=2.5pt]
\node[bubble] (a2)  [above left=-0.105cm and -0.65cm of a1]   {};   
\node [port]  (a3)  [above left=-0.605cm and -1.0cm of a1]  {$p_{m_1}$};  

\tikzstyle{port}=[rectangle,ultra thin,draw=black!75,minimum size=6mm,inner sep=4.5pt]
\node[bubble] (a4)  [above=-0.105cm of a1]   {};   
\node [port]  (a5)  [above=-0.605cm and -2.0cm of a1]  {$p_{a_1}$};  

\node[bubble] (a6)  [above right=-0.105cm and -0.65cm of a1]   {};   
\node [port]  (a7)  [above right=-0.605cm and -1.0cm of a1]  {$p_{d_1}$};  

\node[type]   (a8)      [below=-1.5cm of a1]{Coord. $1$};
\node                   [below=-0.1cm of a8]{$B(4)$};

\node [component] (h1) [left=1.5 of a1]{};

\tikzstyle{port}=[rectangle,ultra thin,draw=black!75,minimum size=6mm,inner sep=2.5pt]
\node[bubble] (h2)  [above left=-0.105cm and -0.65cm of h1]   {};   
\node [port]  (h3)  [above left=-0.605cm and -1.0cm of h1]  {$p_{m_2}$};  

\tikzstyle{port}=[rectangle,ultra thin,draw=black!75,minimum size=6mm,inner sep=4.5pt]
\node[bubble] (h4)  [above=-0.105cm of h1]   {};   
\node [port]  (h5)  [above=-0.605cm and -2.0cm of h1]  {$p_{a_2}$};  

\node[bubble] (h6)  [above right=-0.105cm and -0.65cm of h1]   {};   
\node [port]  (h7)  [above right=-0.605cm and -1.0cm of h1]  {$p_{d_2}$};  

\node[type]   (h8)      [below=-1.5cm of h1]{Coord. $2$};
\node                   [below=-0.1cm of h8]{$B(5)$};

\tikzstyle{port}=[rectangle,ultra thin,draw=black!75,minimum size=6mm,inner sep=4.5pt,minimum height=0.4pt]
\node [component] (b1) [above left=8.0cm and 10.0cm of a1]{};

\node[bubble]     (b2) [below left=-0.105cm and -0.7cm of b1]   {};   
\node [port]      (b3) [below left=-0.58cm and -1.1cm of b1]  {$p_{n_2}$};

\tikzstyle{port}=[rectangle,ultra thin,draw=black!75,minimum size=6mm,inner sep=4.5pt, minimum height=2pt, minimum width= 0.4cm]
\node[bubble]     (b4) [below=-0.105cm of b1]   {};  
\node [port]      (b5) [below=-0.605cm and -2.0cm of b1]  {$p_{q_2}$};

\tikzstyle{port}=[rectangle,ultra thin,draw=black!75,minimum size=6mm,inner sep=4.5pt, minimum height=1pt, minimum width= 0.4cm]
\node[bubble]     (b6) [below right=-0.105cm and -0.75cm of b1]   {};   
\node [port]      (b7) [below right=-0.58cm and -1.05cm of b1]  {$p_{c_2}$};  
 
\tikzstyle{port}=[rectangle,ultra thin,draw=black!75,minimum size=6mm,inner sep=4.5pt, minimum height=2.4pt, minimum width=0.8cm]
\node[bubble] (b8)  [above right=-0.75 and -0.1cm of b1]   {};   
\node [port]  (b9)  [above right=-1.0cm and -.81cm of b1]  {$p_{l_2}$};  

\node[bubble] (b10)  [below right=-1.2cm and -0.1cm of b1]   {};   
\node [port]  (b11)  [below right=-1.4 cm and -.81cm of b1]  {$p_{o_2}$};

\node[type]    (b12)        [above=-0.8cm of b1]{Client $2$};
\node                  [below=-0.1cm of b12]{$B(7)$};

\node [component] (g1) [right=2.0 of b1]{};

\tikzstyle{port}=[rectangle,ultra thin,draw=black!75,minimum size=6mm,inner sep=4.5pt, minimum height=2pt, minimum width=0.8cm]
\node[bubble] (g2)  [right=-0.1cm of g1]   {};   
\node [port]  (g3)  [right=-.81cm of g1]  {$p_{e}$};

\node[bubble] (g4)  [above left=-0.75 and -0.1cm of g1]   {};   
\node [port]  (g5)  [above left=-1.0cm and -.81cm of g1]  {$p_{u}$};  

\node[bubble] (g6)  [below left=-1.2cm and -0.1cm of g1]   {};   
\node [port]  (g7)  [below left=-1.4 cm and -.81cm of g1]  {$p_{t}$};  

\node[type]   (g8)      [above=-0.85cm of g1]{S. Registry};
\node                   [below=-0.1cm of g8]{$B(1)$};

\tikzstyle{port}=[rectangle,ultra thin,draw=black!75,minimum size=6mm,inner sep=4.5pt, minimum height=1pt, minimum width= 0.4cm]
\node [component] (c1) [below=2cm of b1]{};

\node[bubble]     (c2) [below left=-0.105cm and -0.7cm of c1]   {};   
\node [port]      (c3) [below left=-0.58cm and -1.1cm of c1]  {$p_{n_1}$};

\tikzstyle{port}=[rectangle,ultra thin,draw=black!75,minimum size=6mm,inner sep=4.5pt, minimum height=2pt, minimum width= 0.4cm]
\node[bubble]     (c4) [below=-0.105cm of c1]   {};  
\node [port]      (c5) [below=-0.605cm and -2.0cm of c1]  {$p_{q_1}$};

\tikzstyle{port}=[rectangle,ultra thin,draw=black!75,minimum size=6mm,inner sep=4.5pt, minimum height=1pt, minimum width= 0.4cm]
\node[bubble]     (c6) [below right=-0.105cm and -0.75cm of c1]   {};   
\node [port]      (c7) [below right=-0.58cm and -1.05cm of c1]  {$p_{c_1}$};  

\tikzstyle{port}=[rectangle,ultra thin,draw=black!75,minimum size=6mm,inner sep=4.5pt, minimum height=2.4pt, minimum width=0.8cm]
\node[bubble] (c8)  [above right=-0.75 and -0.1cm of c1]   {};   
\node [port]  (c9)  [above right=-1.0cm and -.81cm of c1]  {$p_{l_1}$};  

\node[bubble] (c10)  [below right=-1.2cm and -0.1cm of c1]   {};   
\node [port]  (c11)  [below right=-1.4 cm and -.81cm of c1]  {$p_{o_1}$};

\node[type]    (c12)        [above=-0.8cm of c1]{Client $1$};
\node                       [below=-0.1cm of c12]{$B(6)$};

\tikzstyle{port}=[rectangle,ultra thin,draw=black!75,minimum size=6mm,inner sep=4.5pt, minimum height=1pt, minimum width= 0.4cm]
\node [component] (d1) [right=9cm of b1]{};
\node[bubble] (d2) [below left=-0.105cm and -1.2cm of d1]   {};   
\node [port] (d3) [below left=-0.605cm and -1.5cm of d1]  {$p_{g_2}$};

\node[bubble] (d4) [below right=-0.105cm and -1.2cm of d1]   {};   
\node [port] (d5) [below right=-0.58cm and -1.5cm of d1]  {$p_{s_2}$};  

\node[bubble] (d6) [left=-0.1cm of d1]   {};   
\node [port] (d7) [left=-.81cm of d1]  {$p_{r_2}$}; 

\node[type] (d8) [above=-0.8cm of d1]{Service $2$};
\node            [below=-0.1cm of d8]{$B(3)$};

\node [component] (e1) [below right=2cm and 2cm of d1]{};
\node[bubble] (e2) [below left=-0.105cm and -1.2cm of e1]   {};   
\node [port] (e3) [below left=-0.605cm and -1.5cm of e1]  {$p_{g_1}$};  

\node[bubble] (e4) [below right=-0.105cm and -1.2cm of e1]   {};   
\node [port] (e5) [below right=-0.58cm and -1.5cm of e1]  {$p_{s_1}$}; 

\node[bubble] (e6) [left=-0.1cm of e1]   {};   
\node [port] (e7) [left=-.81cm of e1]  {$p_{r_1}$};

\node[type] (e8) [above=-0.8cm of e1]{Service $1$};
\node        [below=-0.1cm of e8]{$B(2)$};

\path[-]          (g2)  edge                  node {} (d6);
\path[-]          (d6)  edge                  node {} (g2);

\path[-]          (g2)  edge                 node {} (e6);
\path[-]          (e6)  edge                  node {} (g2);
\path[-]          (g4)  edge   [darkorange]               node {} (c8);
\path[-]          (c8)  edge   [darkorange]               node {} (g4);

\path[-]          (g6)  edge  [darkorange]                node {} (c10);
\path[-]          (c10)  edge  [darkorange]                node {} (g6);

\path[-]          (c2)  edge                  node {} (a2);
\path[-]          (c4)  edge                  node {} (a4);
\path[-]          (c6)  edge                  node {} (a6);

\path[-]          (a2)  edge [harlequin]                 node {} (c2);
\path[-]          (a4)  edge [harlequin]                node {} (c4);
\path[-]          (a6)  edge [harlequin]                 node {} (c6);


\path[-]          (g4)  edge  [darkorange]                node {} (b8);
\path[-]          (b8)  edge  [darkorange]                node {} (g4);

\path[-]          (g6)  edge   [darkorange]               node {} (b10);
\path[-]          (b10)  edge  [darkorange]                node {} (g6);

\path[-]          (b2)  edge    [harlequin]              node {} (a2);
\path[-]          (a2)  edge    [harlequin]             node {} (b2);

\path[-]          (b4)  edge   [harlequin]               node {} (a4);
\path[-]          (a4)  edge   [harlequin]              node {} (b4);

\path[-]          (b6)  edge  [harlequin]                node {} (a6);
\path[-]          (a6)  edge  [harlequin]                node {} (b6);

\path[-]          (e2)  edge  [harlequin]                node {} (a4);
\path[-]          (a4)  edge  [harlequin]                node {} (e2);

\path[-]          (e4)  edge [harlequin]                 node {} (a6);
\path[-]          (a6)  edge  [harlequin]               node {} (e4);

\end{tikzpicture}}
\caption{Request/Response architecture. The omitted interactions are derived similarly.}
\label{b_r_r}
\end{figure}
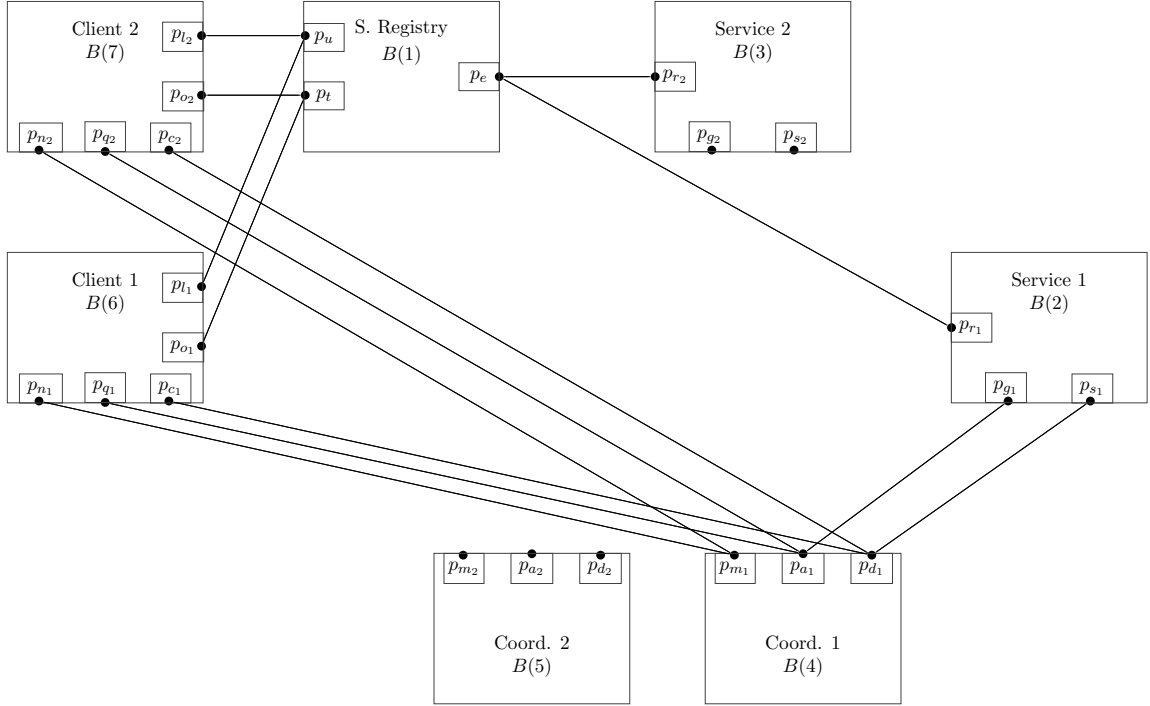
\end{exa}

\begin{exa} \textbf{(Publish/Subscribe)}
\label{b-pu-su}
We consider a component-based system $(\B, \varphi) $ with the Publish/Subscribe architecture. The latter is widely used in IoT 
 applications (cf. for instance \cite{Ol:AP,Pa:Pu}), and recently in cloud systems \cite{Ya:Pr} and robotics \cite{Ma:Ho}. Publish/Subscribe architecture involves three types of components, namely publishers, subscribers, and topics.
 Publishers advertise and transmit to
topics the type of messages they are able to produce. Then, subscribers are connected with the topics
they are interested in,  and topics in turn transfer the messages from publishers to corresponding subscribers. 
Once a subscriber receives
the message it has requested, then it is disconnected from the relevant topic. Publishers cannot check the existence of subscribers and vice-versa \cite{Eg:Pu}. 

For our example we consider two publisher components, two topic components and three subscriber components. 
Hence, we have that $\B=\lbrace B(i)\mid i\in [7]\rbrace$ where $B(1),\ldots, B(7)$ refer to the aforementioned components, respectively. The set of ports of each component is $P(1)=\lbrace p_{a_1},p_{t_1}\rbrace,P(2)=\lbrace p_{a_2},p_{t_2}\rbrace, P(3)=\lbrace p_{n_1},p_{r_1},p_{c_1},p_{s_1},p_{f_1}\rbrace, P(4)=\lbrace p_{n_2},p_{r_2},p_{c_2},p_{s_2},p_{f_2}\rbrace,$ $P(5)=\lbrace p_{e_1},p_{g_1},p_{d_1} \rbrace$, $P(6)=\lbrace p_{e_2},p_{g_2},p_{d_2} \rbrace$, and $P(7)=\lbrace p_{e_3},p_{g_3},p_{d_3} \rbrace$. 
 Figure \ref{b-p-s} depicts one of the possible instantiations for the interactions  
 among the components of our system. The ports $p_{a_{\kappa}}$ and $p_{t_{\kappa}}$, for $\kappa=1,2$, are used from the publishers for advertising and transferring their
 messages to topic components, respectively. Each of the two topics is notified from the publishers and receives their messages through ports $p_{n_{\kappa}}$ and $p_{r_{\kappa}}$, for $\kappa=1,2$,
 respectively. Ports $p_{c_{\kappa}},p_{s_{\kappa}}$ and $p_{f_{\kappa}}$, for $\kappa=1,2$, are used from  topic components for the connection with a subscriber, the sending of 
 a message to a subscriber and for finalizing their connection (disconnection), respectively. Subscribers use the ports $p_{e_{\mu}},p_{g_{\mu}},p_{d_{\mu}} $, for $\mu=1,2,3$, for connecting with the topic (express interest),
getting a message from the topic, and disconnecting from the topic, respectively. The permissible interactions in the architecture range over $I_{\B}$ which includes the following sets:
\begin{itemize}
\item The sets $\lbrace p_{a_{\kappa}},p_{n_{\lambda}}\rbrace$ and $\lbrace p_{t_{\kappa}}, p_{r_{\lambda}}\rbrace$ that refer to the connections of the two publishers with each of the two topics for advertising and transferring their messages, for $\kappa, \lambda=1,2$. The black lines in Figure \ref{b-p-s} depict a possible case for these interactions, and specifically between publisher $B(1)$ with topic $B(3)$ and  publisher $B(2)$ with both topics.
\item The sets $\lbrace p_{c_{\kappa}},p_{e_{\mu}}\rbrace$, $\lbrace p_{s_{\kappa}}, p_{g_{\mu}}\rbrace$ and $\lbrace p_{f_{\kappa}}, p_{d_{\mu}}\rbrace$ that capture the interactions between topics and subscribers, for $\kappa=1,2$ and $\mu=1,2,3$. Figure \ref{b-p-s} represents with orange lines a possible instantiation of these connections, and in particular for topic $B(3)$ with subscribers $B(5)$ and $B(6)$ as well as for topic $B(4)$ with subscriber $B(7)$.
\end{itemize}

\noindent Then, the EPIL formula $\varphi$ for the Publish/Subscribe architecture is $\varphi=\big(\varphi_1\vee \varphi_2\vee (\varphi_1\shuffle \varphi_2)\big)^+$ with
\begin{multline*}
\varphi_1=\bigg(\big( \xi_1 * \varphi_{11}\big)\vee \big( \xi_1 * \varphi_{12}\big)\vee \big(  \xi_1 * \varphi_{13}\big)\vee \\
\big(  \xi_1 * (\varphi_{11}\shuffle \varphi_{12})\big)\vee \big( \xi_1 * (\varphi_{11}\shuffle \varphi_{13})\big)\vee \\ \big( \xi_1 * (\varphi_{12}\shuffle \varphi_{13})\big)\vee \big( \xi_1 * (\varphi_{11}\shuffle\varphi_{12}\shuffle \varphi_{13})\big)\bigg)
\end{multline*}
and 
\begin{multline*}
\varphi_2=\bigg(\big(  \xi_2 * \varphi_{21}\big)\vee \big( \xi_2 * \varphi_{22}\big)\vee \big(  \xi_2* \varphi_{23}\big)\vee \\
\big(  \xi_2 * (\varphi_{21}\shuffle \varphi_{22})\big)\vee \big( \xi_2 * (\varphi_{21}\shuffle \varphi_{23})\big)\vee \\ \big( \xi_2 * (\varphi_{22}\shuffle \varphi_{23})\big)\vee \big(  \xi_2 * (\varphi_{21}\shuffle\varphi_{22}\shuffle \varphi_{23})\big)\bigg)
\end{multline*}
\noindent where we make use of the following auxiliary subformulas: 
\begin{itemize}
    \item[-] $\xi_1=\xi_{11}\vee \xi_{12}\vee (\xi_{11}\shuffle \xi_{12})$
    \item[-] $\xi_2=\xi_{21}\vee \xi_{22}\vee (\xi_{21}\shuffle \xi_{22})$
\end{itemize}
\noindent encode that each of the two topics connects only with the first publisher $B(1)$, or with
the second publisher $B(2)$ or with both of them, and
\begin{itemize}
\item[-] $\xi_{11}=\#(p_{n_1}\wedge p_{a_1}) * \#(p_{r_1}\wedge p_{t_1})$
\item[-] $\xi_{12}=\#(p_{n_1}\wedge p_{a_2}) * \#(p_{r_1}\wedge p_{t_2})$
\item[-] $\xi_{21}=\#(p_{n_2}\wedge p_{a_1}) * \#(p_{r_2}\wedge p_{t_1})$
\item[-] $\xi_{22}=\#(p_{n_2}\wedge p_{a_2}) * \#(p_{r_2}\wedge p_{t_2})$
\end{itemize}
\noindent describe the interactions of the two topics with each of the two publishers,  and 
\begin{itemize}
\item[-] $\varphi_{11}= \#(p_{c_1}\wedge p_{e_1})* \#(p_{s_1}\wedge p_{g_1}) * \#(p_{f_1}\wedge p_{d_1})$
\item[-] $\varphi_{12}= \#(p_{c_1}\wedge p_{e_2})* \#(p_{s_1}\wedge p_{g_2}) * \#(p_{f_1}\wedge p_{d_2})$
\item[-] $\varphi_{13}= \#(p_{c_1}\wedge p_{e_3})*\#(p_{s_1}\wedge p_{g_3}) * \#(p_{f_1}\wedge p_{d_3})$
\item[-] $\varphi_{21}= \#(p_{c_2}\wedge p_{e_1})*\#(p_{s_2}\wedge p_{g_1}) * \#(p_{f_2}\wedge p_{d_1})$
\item[-] $\varphi_{22}= \#(p_{c_2}\wedge p_{e_2})*\#(p_{s_2}\wedge p_{g_2}) * \#(p_{f_2}\wedge p_{d_2})$
\item[-] $\varphi_{23}= \#(p_{c_2}\wedge p_{e_3})*\#(p_{s_2}\wedge p_{g_3}) * \#(p_{f_2}\wedge p_{d_3})$.
\end{itemize}
\noindent describe the connections of the two topics with each of the three subscribers.  Each of $\varphi_1$ and $\varphi_2$ encode the interactions of first and second topic, respectively, with some of the publisher and subscriber components.  The use of the shuffle operator in each of $\varphi_1$ and $\varphi_2$ expresses that the interactions among distinct subscribers may be executed with any order. Then, the EPIL formula $\varphi$ captures the participation of only the first, or only the second, or both topics in the implementation of the architecture. Since there are no order restrictions for the interactions among multiple topics, the EPIL subformulas $\varphi_1$ and $\varphi_2$ are connected in $\varphi$ with the shuffle operator. Finally, the use of the iteration operator in EPIL formula $\varphi$ permits the repetition of the interactions of one or both of the topics with some of the available publishers and subscribers.

\definecolor{darkorange}{rgb}{1.0, 0.55, 0.0}
\definecolor{harlequin}{rgb}{0.25, 1.0, 0.0}
\definecolor{ao}{rgb}{0.0, 0.0, 1.0}

\begin{figure}[h]
\centering
\resizebox{0.7\linewidth}{!}{
\begin{tikzpicture}[>=stealth',shorten >=1pt,auto,node distance=1cm,baseline=(current bounding box.north)]
\tikzstyle{component}=[rectangle,ultra thin,draw=black!75,align=center,inner sep=9pt,minimum size=1.5cm,minimum height=3.5cm,minimum width=3cm]
\tikzstyle{port}=[rectangle,ultra thin,draw=black!75,minimum size=7mm]
\tikzstyle{bubble} = [fill,shape=circle,minimum size=5pt,inner sep=0pt]
\tikzstyle{type} = [draw=none,fill=none]

\node [component,align=center] (a1)  {};
\node [port] (a2) [above right=-1.30 and -0.78cm of a1]  {$p_{a_1}$};
\node[bubble] (a3) [below right=-0.40 and -0.10cm of a2]   {};

\node [port] (a4) [below =0.90cm of a2]  {$p_{t_1}$};
\node[bubble] (a5) [below right=-0.35 and -0.10cm of a4]   {};

\node[type] (a6) [above=-0.45cm of a1]{Publ. $1$};
\node[type]  [below=-0.05cm of a6]{$B(1)$};

\node [component] (b1) [right=3cm of a1] {};
\node [port] (b2) [below right=-1.07cm and -0.78cm of b1]  {$p_{f_1}$}; 
\node[bubble] (b3) [below right=-0.40cm and -0.10cm of b2]   {};   

\node [port] (b4) [above=0.25cm of b2]  {$p_{s_1}$}; 
\node[bubble] (b5) [above right=-0.40cm and -0.10cm of b4]   {};   

\node [port] (b6) [above =0.3cm of b4]  {$p_{c_1}$}; 
\node[bubble] (b7) [above right=-0.40cm and -0.10cm of b6]   {};

\node [port] (b8) [above left= -1.4cm and -0.78cm of b1]  {$p_{n_1}$};
\node[bubble] (b9) [above left=-0.35cm and -0.10cm of b8]   {};

\node [port] (b10) [below= 0.85cm of b8]  {$p_{r_1}$};
\node[bubble] (b11) [below left=-0.40 and -0.10cm of b10]   {};
 
\node[type] (b12) [above=-0.55cm of b1]{Topic $1$};
\node[type]  [below=-0.1cm of b12]{$B(3)$};

\node [component] (c1) [above right= -1cm and 3cm of b1] {};
\node [port] (c2) [below left=-1.07cm and -0.78cm of c1]  {$p_{d_1}$}; 
\node[bubble] (c3) [below left=-0.40cm and -0.10cm of c2]   {};  

\node [port] (c4) [above=0.25cm of c2]  {$p_{g_1}$}; 
\node[bubble] (c5) [above left=-0.40cm and -0.10cm of c4]   {};   

\node [port] (c6) [above=0.3cm of c4]  {$p_{e_1}$}; 
\node[bubble] (c7) [above left=-0.40cm and -0.10cm of c6]   {};   
\node[type] (c8) [above=-0.55cm of c1]{Subs. $1$};
\node[type]  [below=-0.1cm of c8]{$B(5)$};

\node [component] (d1) [below right=-1cm and 3cm of b1] {};
\node [port] (d2) [below left=-1.07cm and -0.78cm of d1]  {$p_{d_2}$}; 
\node[bubble] (d3) [below left=-0.40cm and -0.10cm of d2]   {};  

\node [port] (d4) [above=0.25cm of d2]  {$p_{g_2}$}; 
\node[bubble] (d5) [above left=-0.40cm and -0.10cm of d4]   {};   

\node [port] (d6) [above=0.3cm of d4]  {$p_{e_2}$}; 
\node[bubble] (d7) [above left=-0.40cm and -0.10cm of d6]   {};   
\node[type] (d8) [above=-0.55cm of d1]{Subs. $2$};
\node[type]  [below=-0.1cm of d8]{$B(6)$};


\node [component,align=center] (e1) [below=4.5cm of a1] {};
\node [port] (e2) [above right=-1.30 and -0.78cm of e1]  {$p_{a_2}$};
\node[bubble] (e3) [below right=-0.40 and -0.10cm of e2]   {};

\node [port] (e4) [below =0.90cm of e2]  {$p_{t_2}$};
\node[bubble] (e5) [below right=-0.35 and -0.10cm of e4]   {};

\node[type] (e6) [above=-0.45cm of e1]{Publ. $2$};
\node[type]  [below=-0.05cm of e6]{$B(2)$};

\node [component] (f1) [right=3cm of e1] {};
\node [port] (f2) [below right=-1.07cm and -0.78cm of f1]  {$p_{f_2}$}; 
\node[bubble] (f3) [below right=-0.40cm and -0.10cm of f2]   {};   

\node [port] (f4) [above=0.25cm of f2]  {$p_{s_2}$}; 
\node[bubble] (f5) [above right=-0.40cm and -0.10cm of f4]   {};   

\node [port] (f6) [above =0.3cm of f4]  {$p_{c_2}$}; 
\node[bubble] (f7) [above right=-0.40cm and -0.10cm of f6]   {};

\node [port] (f8) [above left= -1.4cm and -0.78cm of f1]  {$p_{n_2}$};
\node[bubble] (f9) [above left=-0.35cm and -0.10cm of f8]   {};

\node [port] (f10) [below= 0.85cm of f8]  {$p_{r_2}$};
\node[bubble] (f11) [below left=-0.40 and -0.10cm of f10]   {};
 
\node[type] (f12) [above=-0.55cm of f1]{Topic $2$};
\node[type]  [below=-0.1cm of f12]{$B(4)$};

\node [component] (g1) [right= 3cm of f1] {};
\node [port] (g2) [below left=-1.07cm and -0.78cm of g1]  {$p_{d_3}$}; 
\node[bubble] (g3) [below left=-0.40cm and -0.10cm of g2]   {};  

\node [port] (g4) [above=0.25cm of g2]  {$p_{g_3}$}; 
\node[bubble] (g5) [above left=-0.40cm and -0.10cm of g4]   {};   

\node [port] (g6) [above=0.3cm of g4]  {$p_{e_3}$}; 
\node[bubble] (g7) [above left=-0.40cm and -0.10cm of g6]   {};   
\node[type] (g8) [above=-0.55cm of g1]{Subs. $3$};
\node[type]  [below=-0.1cm of g8]{$B(7)$};

\path[-]          (a3)  edge                  node {} (b9);
\path[-]          (b9)  edge                  node {} (a3);

\path[-]          (a5)  edge                  node {} (b11);
\path[-]          (b11)  edge                  node {} (a5);

\path[-]          (b7)  edge  [darkorange]                node {} (c7);
\path[-]          (c7)  edge  [darkorange]                node {} (b7);

\path[-]          (b7)  edge  [darkorange]                node {} (d7);
\path[-]          (d7)  edge  [darkorange]                node {} (b7);

\path[-]          (b5)  edge   [darkorange]               node {} (c5);
\path[-]          (c5)  edge   [darkorange]               node {} (b5);

\path[-]          (b5)  edge  [darkorange]                node {} (d5);
\path[-]          (d5)  edge   [darkorange]               node {} (b5);

\path[-]          (b3)  edge   [darkorange]               node {} (c3);
\path[-]          (c3)  edge  [darkorange]                node {} (b3);

\path[-]          (b3)  edge   [darkorange]               node {} (d3);
\path[-]          (d3)  edge  [darkorange]                node {} (b3);


\path[-]          (e3)  edge                  node {} (f9);
\path[-]          (f9)  edge                  node {} (e3);

\path[-]          (e5)  edge                  node {} (f11);
\path[-]          (f11)  edge                  node {} (e5);

\path[-]          (e3)  edge                  node {} (b9);
\path[-]          (b9)  edge                  node {} (e3);

\path[-]          (e5)  edge                  node {} (b11);
\path[-]          (b11)  edge                  node {} (e5);


\path[-]          (f7)  edge  [darkorange]                node {} (g7);
\path[-]          (g7)  edge  [darkorange]                node {} (f7);

\path[-]          (f5)  edge  [darkorange]                node {} (g5);
\path[-]          (g5)  edge  [darkorange]                node {} (f5);

\path[-]          (f3)  edge  [darkorange]                node {} (g3);
\path[-]          (g3)  edge  [darkorange]                node {} (f3);

\end{tikzpicture}}
\caption{Publish/Subscribe architecture. A possible execution for the interactions.}
\label{b-p-s}
\end{figure}
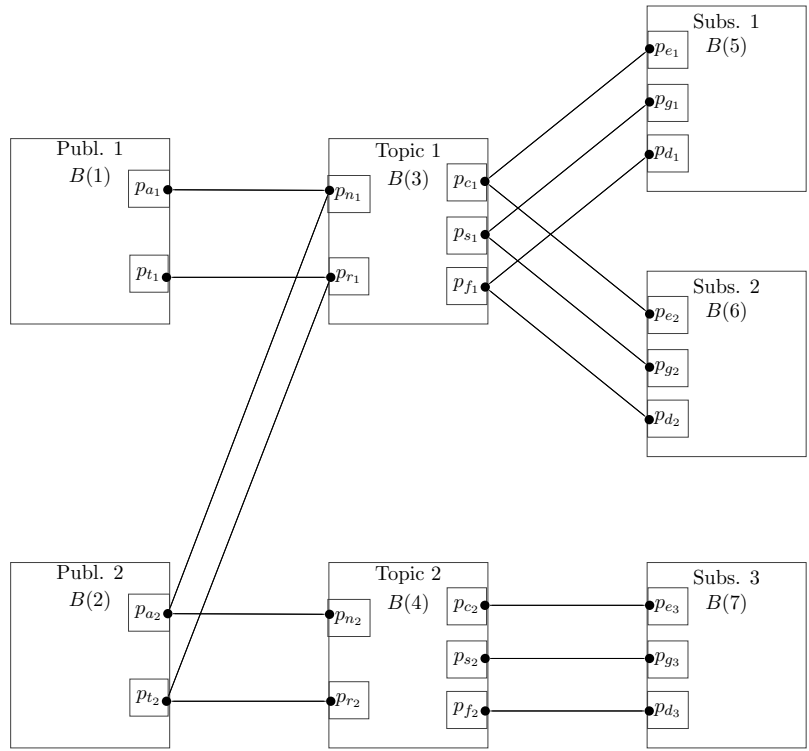

\end{exa}
 
The presented examples illustrate that EPIL formulas are expressive enough to 
encode the execution order of the permissible interactions as well as to specify
the subsequent implementation of a component-based system's architecture through recursive interactions. We note that in \cite{Pi:Ar} we described Blackboard, Publish/Subscribe, and Request/Response architectures, without applying any iteration operator, and hence without allowing recursion in the corresponding interactions.

\section{Parametric component-based systems}

In this section we deal with component-based systems in the parametric setting. 
Component-based systems considered in Subsection \ref{comp_based_system} are comprised of a finite number of components 
which are of the same or distinct type. On the other hand, in the parametric setting a
component-based model is comprised of a finite number of distinct \emph{component types} where the cardinality 
of the \emph{instances} of each type is a parameter for the system. It should be clear, that in real world applications we do not need an unbounded number of components. Nevertheless, the number of instances of every component type is unknown or it can be modified during a process. Therefore, we consider parametric component-based systems, i.e., component-based systems with
infinitely many instances of every component type.

Let $\B=\{B(i) \mid i \in [n]  \}$ be a set of component types. For every $i \in [n] $ and $j \geq 1$ we consider a copy $B(i,j)=(Q(i,j),P(i,j),q_{0}(i,j),R(i,j))$ of $B(i)$, namely the \emph{$j$-th instance of} $B(i)$.  Hence, for every $i \in [n]$ and $j \geq 1$, the instance $B(i,j)$ is also a component and we call it a \emph{parametric component} or a \emph{component instance}. We assume that $(Q(i,j) \cup P(i,j)) \cap (Q(i', j') \cup  P(i',j')) = \emptyset$ whenever $i \neq i' $ or $j \neq j'$ for every $ i, i' \in [n]$ and $j, j' \geq 1$.  This restriction is needed in order to identify the distinct parametric components. It also permits us to use, without any confusion, the notation  $P(i,j)=\{p(j) \mid p\in P(i)\}$ for every $i \in [n]$ and $j \geq 1$.   We set $p\mathcal{B}= \{B(i,j) \mid i \in [n],  j \geq 1 \} $ and call it a set of \emph{parametric components}. The set of ports of $p\B$ is given by $P_{p\B} =\bigcup_{i \in [n], j\geq 1}  P(i,j)$.

\begin{rem}
Observe that in the parametric setting we use an index ``$i$" in order to specify the type of a component. Specifically, 
we let $B(i,j)$ denote the $j$-th instance of a component of type $i$, for $i \in [n]$ and $j \geq 1$. On the other hand, for simplicity we avoided such a notation for component-based systems in the non-parametric setting. In particular we presented  a sequential enumeration of
the several, finite in number, components of the same or different type by $B(i)$, for $i\in [n]$. This explains why in the parametric setting we use the term ``component type". For instance in the Request/Response architecture of Example \ref{b-r-r}, we denote by $B(2),B(3)$ the two service components. On the other hand, in a parametric Request/Response architecture we would let $B(2,1)$ and $B(2,2)$ refer to two component instances of the service component type. For this, we arbitrarily choose $i=2$ to denote the type of service, and in turn we let $j=1,2$ refer to the two instances of that type.
\end{rem}

As it is already mentioned, in practical applications we do not know how many instances of each component type are connected at a concrete time. This means that we cannot define interactions of $p\mathcal{B}$ in the same way as we did for finite sets of components. Hence, we need a symbolic representation to describe interactions, and in turn architectures, of parametric systems. For this, we introduce the first-order extended interaction
 logic which is proved sufficient to describe a wide class of architectures of parametric component-based systems. Similarly to EPIL the semantics of our logic encodes the order and iteration of the interactions implemented within a parametric architecture. Formalization of such aspects is important for constructing well-defined architectures, and hence parametric systems which are 
less error prone and satisfy most of their requirements \cite{Am:Pa,Bl:De,De:Pa}.

\subsection{First-order extended interaction logic}
\label{sub_FOEIL}
In this subsection we introduce the  first-order extended interaction logic as a modelling language for describing the interactions of parametric component-based systems. For this, we need to equip EPIL formulas with variables. Due to the
nature of parametric systems we need to distinguish variables referring to different component types. Let $p\mathcal{B}= \{B(i,j) \mid i \in [n], j \geq 1 \} $ be a set of parametric components. We consider pairwise disjoint countable sets of first-order variables  $\mathcal{X}^{(1)}, \ldots , \mathcal{X}^{(n)}$ referring to  instances of component types $B(1), \ldots, B(n)$, respectively. First-order  variables in $\mathcal{X}^{(i)}$, for every $i \in [n]$, will be denoted by small  letters with the corresponding superscript. Hence by $x^{(i)} $  we understand that $x^{(i)} \in \mathcal{X}^{(i)}$, $i\in [n]$, is a first-order variable referring to an  instance of component type $B(i)$. We let $\mathcal{X}=\mathcal{X}^{(1)}\cup \ldots \cup \mathcal{X}^{(n)}$ and  set  $P_{p\B(\mathcal{X})} =\left\{p\left(x^{(i)}\right) \mid i \in [n], x^{(i)} \in \mathcal{X}^{(i)},  \text{ and } p\in P(i)\right \}$. 
  
\begin{defi} \label{def_fOEIL}
Let $p\mathcal{B}= \{B(i,j) \mid i \in [n], j \geq 1 \} $ be a set of parametric components. Then the syntax of first-order extended interaction logic (FOEIL for short) formulas $\psi$ over $p\mathcal{B}$\footnote{According to our terminology for EPIL formulas, a FOEIL formula should be defined over the set of ports of $p\B$. Nevertheless, we prefer for simplicity to refer to the set $p\B$ of parametric components.} is given by the grammar
\begin{multline*}
\psi  ::= \varphi  \mid x^{(i)}=y^{(i)} \mid  \neg(x^{(i)}=y^{(i)}) \mid  \psi\vee\psi \mid \psi \wedge \psi \mid   \psi * \psi   \mid \psi\shuffle \psi \mid \psi^+ \mid \\ \exists x^{(i)}. \psi \mid  \forall x^{(i)} . \psi \mid \exists^* x^{(i)}. \psi \mid \forall^*x^{(i)}.\psi \mid \exists^{\shuffle} x^{(i)}.\psi \mid \forall^{\shuffle}x^{(i)}.\psi  
\end{multline*}
where $\varphi$ is an EPIL formula over $P_{p\B(\mathcal{X})}$, $i \in [n]$, $x^{(i)}, y^{(i)} $ are  first-order variables in $\mathcal{X}^{(i)}$, $\exists^*$ denotes the existential concatenation  quantifier, $\forall^*$ the universal concatenation quantifier, $\exists^{\shuffle}$ is the existential shuffle quantifier, and $\forall^{\shuffle}$ is the universal shuffle quantifier. Furthermore, we assume that whenever $\psi$ contains a subformula of the form $\exists ^*x^{(i)}.\psi'$ or $\exists ^{\shuffle} x^{(i)}.\psi'$, then the application of negation in $\psi'$ is permitted only  in PIL formulas and formulas of the form $x^{(j)}=y^{(j)}$. 
\end{defi}

Let $\psi$ be a FOEIL formula over $p\B$. As usual, we denote by $\mathrm{free}(\psi)$ the set of free variables of $\psi$. If $\psi$ has no free variables, then it is a \emph{sentence}. 
We consider a mapping $r:[n] \rightarrow \mathbb{N}$. The value $r(i)$, for every $i \in [n]$, intends to represent the finite number of instances of the component type $B(i)$ in the parametric system, affecting in turn, the corresponding interactions. Hence, for different mappings we obtain a different parametric system. 
We let $p\B(r)= \{B(i,j) \mid i \in [n], j \in [r(i)] \} $ and call it  the \emph{instantiation of} $p\B$ w.r.t. $r$. We denote by $P_{p\B(r)}$ the set of all ports of components' instances in $p\B(r)$, i.e., $P_{p\B(r)}= \bigcup_{i\in [n], j \in [r(i)]}P(i,j)$. The set $I_{p\B(r)}$ of interactions of $p\B(r)$ is given by $I_{p\B(r)}=\{a \in I(P_{p\B(r)}) \mid  \vert a \cap P(i,j)\vert \leq 1 \text{ for every } i\in [n] \text{ and } j \in [r(i)]\}$.

Let $\mathcal{V} \subseteq \mathcal{X} $ be a finite set of first-order variables.  We let $P_{p\B(\mathcal{V})}= \{ p(x^{(i)}) \in P_{p\B(\mathcal{X})} \mid x^{(i)} \in \mathcal{V} \}$.
To interpret FOEIL formulas over $p\B$ we use the notion of an \emph{assignment} defined with respect to the set of variables $\mathcal{V}$ and the mapping $r$.  Formally, a $(\mathcal{V},r)$-\emph{assignment} is a mapping $\sigma : \mathcal{V} \rightarrow \mathbb{N}  $ such that $\sigma(\mathcal{V} \cap \mathcal{X}^{(i)} ) \subseteq [r(i)]$ for every $i \in [n]$. If $\sigma$ is a $(\mathcal{V},r)$-assignment, then $\sigma[x^{(i)} \rightarrow j]$  is the $(\mathcal{V} \cup \{x^{(i)}\},r)$-assignment which acts as $\sigma$ on $\mathcal{V}\setminus \{x^{(i)}\}$  and assigns $j$ to $x^{(i)}$. If  $\varphi$ is an EPIL formula over $P_{p\B(\mathcal{V})}$, then  $\sigma(\varphi)$ is an EPIL formula over $P_{p\B(r)}$ which is obtained by $\varphi$ by replacing every port $p(x^{(i)})$  in $\varphi$ by $p(\sigma(x^{(i)}))$. Intuitively, a $(\mathcal{V},r)$-assignment $\sigma$ assigns unique identifiers to each instance in a parametric system, w.r.t. the mapping $r$. 

We interpret FOEIL formulas over triples consisting of a mapping $r:[n] \rightarrow \mathbb{N}$, a $(\mathcal{V},r)$-assignment $\sigma$, and a word $w \in I_{p\B(r)}^+$. 
As for EPIL formulas, we define for every FOEIL formula $\psi$ over $p\B$ and natural number $\nu\geq 1$, the FOEIL formula $\psi^{\nu}$ over $p\B$ by induction on $\nu$: $\psi^1=\psi$ and $\psi^{\nu+1}=\psi^{\nu} * \psi$. The semantics of formulas of the form $\exists^* x^{(i)} . \psi$ and $\forall^* x^{(i)} . \psi$
(resp. $\exists^{\shuffle} x^{(i)} . \psi$ and $\forall^{\shuffle} x^{(i)} . \psi$) refer to satisfaction of $\psi$ by subwords of  
$w$. The subwords correspond to component instances which are determined by the application of the assignment $\sigma$ to $ x^{(i)}$, and  $w$ results by the $*$ (resp. $\shuffle$) operator among the subwords.

\begin{defi}
Let $\psi$ be a FOEIL formula over a set $p\mathcal{B}= \{B(i,j) \mid i \in [n], j \geq 1 \} $ of parametric components and $\mathcal{V} \subseteq \mathcal{X}$ a finite set containing $\mathrm{free}(\psi)$. Then for every $r:[n] \rightarrow \mathbb{N}$, $(\mathcal{V},r)$-assignment $\sigma$, and $w  \in I_{p\B(r)}^+$ we define the satisfaction relation $(r,\sigma,w) \models\psi$, inductively on the structure of $\psi$ as follows:

\

\begin{itemize}

     \item[-] $(r,\sigma,w) \models  \varphi$ iff $w \models \sigma(\varphi)$,
     
     \item[-] $(r,\sigma,w) \models  x^{(i)}=y^{(i)}$ iff $\sigma(x^{(i)})=\sigma(y^{(i)})$,
     
     \item[-] $(r,\sigma,w) \models  \neg(x^{(i)}=y^{(i)})$ iff $(r,\sigma,w) \not\models x^{(i)}=y^{(i)}$,

     \item[-]  $(r,\sigma, w)  \models \psi_1 \vee \psi_2$ iff $(r,\sigma, w)\models  \psi_1$ or $(r,\sigma, w)\models \psi_2$, 
 
\item[-]  $(r,\sigma, w)  \models \psi_1 \wedge \psi_2$ iff $(r,\sigma, w)\models  \psi_1$ and $(r,\sigma, w)\models \psi_2$,  
     
     \item[-] $(r,\sigma, w) \models \psi_1 * \psi_2$ iff there exist $w_1,w_2 \in I_{p\B(r)}^+$ such that $w=w_1w_2$ and   $(r,\sigma, w_i) \models  \psi_i$ for $i=1,2$, 
     \item[-] $(r, \sigma, w) \models \psi_1 \shuffle \psi_2$ iff there exist $w_1, w_2 \in I_{p\B(r)}^+$ such that $w \in w_1\shuffle w_2$ and $(r, \sigma, w_i) \models \psi_i$ for $i=1,2$, 
     
\item[-] $(r, \sigma, w) \models \psi^+ $ iff there exists $\nu \geq 1$ such that $(r, \sigma, w) \models \psi^{\nu}$,     
     
     \item[-]   $(r,\sigma, w) \models \exists x^{(i)} . \psi$ iff there exists $j \in [r(i)] $ such that  $(r, \sigma[x^{(i)} \rightarrow  j ], w) \models \psi$,

      \item[-]  $(r,\sigma, w) \models \forall x^{(i)} . \psi$ iff  $(r, \sigma[x^{(i)} \rightarrow j ], w) \models \psi$ for every $j \in [r(i)]$,
      
\item[-]  $(r,\sigma, w) \models \exists^ * x^{(i)} . \psi$ iff there exists $1\leq k\leq r(i)$ and $w_{l_1}, \ldots, w_{l_k} \in I_{p\B(r)}^+$ with $1 \leq l_1 <  \ldots <  l_k \leq r(i)$ such that $w = w_{l_1}  \ldots  w_{l_k}$ and $(r, \sigma[x^{(i)} \rightarrow j ], w_j) \models \psi$ for every $j = l_1, \ldots, l_k$,

      \item[-]  $(r,\sigma, w) \models \forall^* x^{(i)} . \psi$ iff there exist $w_1, \ldots, w_{r(i)} \in I_{p\B(r)}^+$ such that $w=w_1 \ldots w_{r(i)}$ and $(r, \sigma[x^{(i)} \rightarrow j ], w_j) \models \psi$ for every $j \in [r(i)]$,
      
\item[-]  $(r,\sigma, w) \models \exists^{\shuffle} x^{(i)} . \psi$ iff there exists $1\leq k\leq r(i)$ and $w_{l_1}, \ldots, w_{l_k} \in I_{p\B(r)}^+$ with $1 \leq l_1 <  \ldots <  l_k \leq r(i)$  such that $w \in w_{l_1} \shuffle \ldots \shuffle w_{l_k}$ and $(r, \sigma[x^{(i)} \rightarrow j ], w_j) \models \psi$ for every $j = l_1, \ldots, l_k$,

\item[-]  $(r,\sigma, w) \models \forall^{\shuffle} x^{(i)} . \psi$ iff there exist $w_1, \ldots, w_{r(i)} \in I_{p\B(r)}^+$ such that $w \in w_1 \shuffle \ldots \shuffle w_{r(i)}$ and $(r, \sigma[x^{(i)} \rightarrow j ], w_j) \models \psi$ for every $j \in [r(i)]$.
\end{itemize}
\end{defi}

\noindent By definition of parametric systems, all instances of each component type are identical, hence the order specified above in the semantics of $\exists^*, \forall^*, \exists^{\shuffle}, \forall^{\shuffle}$ quantifiers causes no restriction in the derived architecture. Moreover, it is important for the complexity of the translation algorithm of FOEIL formulas to finite automata (cf. proof of Proposition \ref{formula-aut}). 

If $\psi$ is a FOEIL sentence over $p\B$, then we simply write $(r, w) \models \psi$. Let also $\psi'$ be a FOEIL sentence over $p\mathcal{B}$. Then,  $\psi$ and $\psi'$ are called \emph{equivalent w.r.t.} $r$ whenever $(r,w) \models \psi$ iff $(r,w) \models \psi'$, for every $w  \in I_{p\B(r)}^+$.

In the sequel, we shall write also $x^{(i)} \neq y^{(i)}$  for $\neg(x^{(i)}=y^{(i)})$.

Let $\beta$ be a boolean combination of atomic formulas of the form  $x^{(i)}=y^{(i)}$ and $\psi$ a FOEIL formula over $p\B$. Then, we define 
$\beta \rightarrow \psi ::= \neg \beta \vee \psi$.

For simplicity sometimes we denote boolean combinations of  formulas of the form $x^{(i)} = y^{(i)}$ as constraints. For instance we write $\exists x^{(i)}\forall y^{(i)}\exists x^{(j)}\forall y^{(j)} ((x^{(i)}  \neq y^{(i)}) \wedge (x^{(j)} \neq y^{(j)})) .\psi$ for  $\exists x^{(i)}\forall y^{(i)}\exists x^{(j)}\forall y^{(j)}. (((x^{(i)}\neq y^{(i)}) \wedge (x^{(j)} \neq y^{(j)})) \rightarrow \psi)$. 

Note that in \cite{Ma:Co} the authors considered a universe of component types and hence, excluded in their logic formulas the erroneous types
for each architecture. Such a restriction is not needed in our setting since we consider a well-defined set $[n]$ of component types for
each architecture. 
Now we are ready to formally define the concept of a parametric component-based system.

\begin{defi}
A parametric component-based system is a pair $(p\B, \psi)$ where $p\B=\{B(i,j) \mid i\in [n], j \geq 1 \}$ is a set of parametric components and $\psi$ is a FOEIL sentence over $p\B$.  
\end{defi}

In the sequel, for simplicity we refer to parametric component-based systems as parametric systems. We remind that in this work we focus on the architectures of parametric systems. The study of parametric systems' semantics
is left for investigation in future work as a part of parametric verification.

For our examples in the next subsection, we shall need the following macro FOEIL formula. Let $p\B=\{B(i,j) \mid i\in [n], j \geq 1\}$ and $1 \leq i_1, \ldots, i_m \leq n$ be pairwise different indices. Then we set
\begin{multline*}
\#\big(p_{i_1}(x^{(i_1)}) \wedge \ldots \wedge p_{i_m}(x^{(i_m)})\big)::= \big (p_{i_1}(x^{(i_1)}) \wedge \ldots \wedge p_{i_m}(x^{(i_m)})\big) \wedge \\ 
\bigg(\bigwedge_{j =i_1, \ldots, i_m } \bigwedge_{p \in P(j)\setminus \{ p_j\}} \neg p(x^{(j)}) \bigg) \wedge 
\bigg(\bigwedge_{j=i_1, \ldots, i_m} \forall y^{(j)}(y^{(j)} \neq x^{(j)}).\bigwedge_{p \in P(j)} \neg p(y^{(j)}) \bigg ) \wedge \\
 \bigg( \bigwedge_{k \in [n]\setminus \{i_1, \ldots, i_m\} }\bigwedge_{p \in P(k)}\forall x^{(k)} .   \neg p(x^{(k)}) \bigg ) .
\end{multline*} 
The first $m-1$ conjunctions, in the above formula,   express that the ports appearing in the argument of $\#$ participate in the interaction. In the second line, the   
  double indexed conjunctions in the first pair of big parentheses disable all the other ports of the participating instances of component type $i_1, \ldots, i_m$ described by variables $x^{(i_1)}, \ldots, x^{(i_m)}$, respectively; conjunctions in the second pair of parentheses disable all ports of remaining instances of component types $i_1, \ldots, i_m$. Finally, the last conjunct in the third line ensures that no ports in instances of remaining component types
participate in the interaction.

\subsection{Examples of FOEIL sentences for parametric architectures}

In this subsection we present several examples of FOEIL sentences describing concrete parametric architectures. In what follows we often refer to a component instance simply by instance. Moreover, whenever is defined a unique instance for a component type we may also consider the corresponding set of variables as a singleton.

\begin{exa}
\label{ma-sl}
\textbf{(Parametric Master/Slave)} We present a FOEIL sentence for the parametric Master/Slave architecture. Master/Slave architecture concerns two types of components, namely masters and slaves \cite{Ma:Co}. We denote by $p_m$ the unique port of master component and by $p_s$ the unique port of slave component. Every slave must be connected with exactly one master. Interactions among masters (resp. slaves) are not permitted. An instantiation of the architecture for two masters and two slaves with all the possible cases of the allowed interactions is shown in Figure \ref{m-s}. 
   
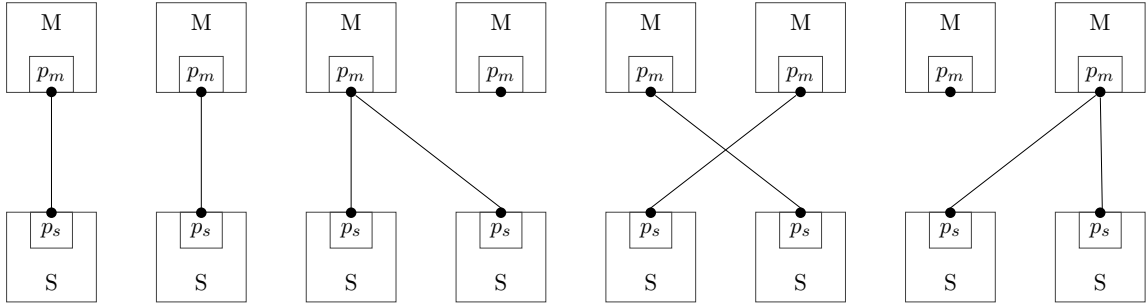
\begin{figure}[h]
\centering
\resizebox{1.0\linewidth}{!}{
\begin{tikzpicture}[>=stealth',shorten >=1pt,auto,node distance=1cm,baseline=(current bounding box.north)]
\tikzstyle{component}=[rectangle,ultra thin,draw=black!75,align=center,inner sep=9pt,minimum size=1.5cm]
\tikzstyle{port}=[rectangle,ultra thin,draw=black!75,minimum size=6mm]
\tikzstyle{bubble} = [fill,shape=circle,minimum size=5pt,inner sep=0pt]
\tikzstyle{type} = [draw=none,fill=none]

 \node [component,align=center] (a1)  {};
 \node [port] (a2) [below=-0.605cm of a1]  {$p_m$};
 \node[bubble] (a3) [below=-0.105cm of a1]   {};
\node[type]  [above=-0.6cm of a1]{M};

\node [component] (a4) [below=2cm of a1]  {};
\node [port,align=center,inner sep=5pt] (a5) [above=-0.6035cm of a4]  {$p_s$};
\node[bubble] (a6) [above=-0.105cm of a4]   {};
\node[type]  [below=-0.6cm of a4]{S};

\path[-]          (a1)  edge                  node {} (a4);

 \node [component] (b1) [right=1cm of a1] {};
 \node [port] (b2) [below=-0.605cm of b1]  {$p_m$};
 \node[bubble] (b3) [below=-0.105cm of b1]   {};
 \node[type]  [above=-0.6cm of b1]{M};

\node [component] (b4) [below=2cm of b1]  {};
\node [port,align=center,inner sep=5pt] (b5) [above=-0.6035cm of b4]  {$p_s$};
\node[bubble] (b6) [above=-0.105cm of b4]   {};
\node[type]  [below=-0.6cm of b4]{S};

\path[-]          (b1)  edge                  node {} (b4);

 \node [component] (c1)[right=1cm of b1] {};
 \node [port] (c2) [below=-0.605cm of c1]  {$p_m$};
 \node[bubble] (c3) [below=-0.105cm of c1]   {};
  \node[type]  [above=-0.6cm of c1]{M};

\node [component] (c4) [below=2cm of c1]  {};
\node [port,align=center,inner sep=5pt] (c5) [above=-0.6035cm of c4]  {$p_s$};
\node[bubble] (c6) [above=-0.105cm of c4]   {};
\node[type]  [below=-0.6cm of c4]{S};

\path[-]          (c1)  edge                  node {} (c4);

 \node [component] (d1)[right=1cm of c1] {};
 \node [port] (d2) [below=-0.605cm of d1]  {$p_m$};
 \node[bubble] (d3) [below=-0.105cm of d1]   {};
  \node[type]  [above=-0.6cm of d1]{M};

\node [component] (e4) [below=2cm of d1]  {};
\node [port,align=center,inner sep=5pt] (e5) [above=-0.6035cm of e4]  {$p_s$};
\node[] (i1) [above right=-0.25 cm and -0.25cm of e5]   {};
\node[bubble] (e6) [above=-0.105cm of e4]   {};
\node[type]  [below=-0.6cm of e4]{S};

\path[-]          (c3)  edge  node {}           (i1);

 \node [component] (f1)[right=1cm of d1] {};
 \node [port] (f2) [below=-0.605cm of f1]  {$p_m$};
 \node[bubble] (f3) [below=-0.105cm of f1]   {};
  \node[type]  [above=-0.6cm of f1]{M};

\node [component] (g4) [below=2cm of f1]  {};
\node [port,align=center,inner sep=5pt] (g5) [above=-0.6035cm of g4]  {$p_s$};
\node[] (i2) [above left=-0.25 cm and -0.25cm of g5]   {};
\node[bubble] (g6) [above=-0.105cm of g4]   {};
\node[type]  [below=-0.6cm of g4]{S};

 \node [component] (h1)[right=1cm of f1] {};
 \node [port] (h2) [below=-0.605cm of h1]  {$p_m$};
 \node[bubble] (h3) [below=-0.105cm of h1]   {};
  \node[type]  [above=-0.6cm of h1]{M};

\node [component] (j4) [below=2cm of h1]  {};
\node [port,align=center,inner sep=5pt] (j5) [above=-0.6035cm of j4]  {$p_s$};
\node[] (i3) [above right=-0.25 cm and -0.25cm of j5]   {};
\node[bubble] (j6) [above=-0.105cm of j4]   {};
\node[type]  [below=-0.6cm of j4]{S};

\path[-]          (h3)  edge                  node {} (i2);

\path[-]          (f3)  edge                  node {} (i3);

 \node [component] (k1)[right=1cm of h1] {};
 \node [port] (k2) [below=-0.605cm of k1]  {$p_m$};
 \node[bubble] (k3) [below=-0.105cm of k1]   {};
  \node[type]  [above=-0.6cm of k1]{M};

\node [component] (k4) [below=2cm of k1]  {};
\node [port,align=center,inner sep=5pt] (k5) [above=-0.6035cm of k4]  {$p_s$};
\node[] (i4) [above left=-0.25 cm and -0.25cm of k5]   {};
\node[bubble] (k6) [above=-0.105cm of k4]   {};
\node[type]  [below=-0.6cm of k4]{S};

 \node [component] (l1)[right=1cm of k1] {};
 \node [port] (l2) [below=-0.605cm of l1]  {$p_m$};
 \node[bubble] (l3) [below=-0.105cm of l1]   {};
  \node[type]  [above=-0.6cm of l1]{M};

\node [component] (m4) [below=2cm of l1]  {};
\node [port,align=center,inner sep=5pt] (m5) [above=-0.6035cm of m4]  {$p_s$};
\node[] (i5) [above right=-0.25 cm and -0.45cm of m5]   {};
\node[bubble] (m6) [above=-0.105cm of m4]   {};
\node[type]  [below=-0.6cm of m4]{S};

\path[-]          (l3)  edge                  node {} (i4);

\path[-]          (l3)  edge                  node {} (i5);

\end{tikzpicture}}
\caption{Master/Slave architecture.}
\label{m-s}
\end{figure}

\noindent We let $\mathcal{X}^{(1)}, \mathcal{X}^{(2)}$ denote the sets of variables of master and slave component instances, respectively. Then, the FOEIL sentence $\psi$ representing parametric Master/Slave architecture is
$$\psi=\forall^* x^{(2)} \exists x^{(1)}. \#( p_m(x^{(1)}) \wedge p_s(x^{(2)})). $$

\noindent In the above sentence, the universal concatenation quantifier accompanied with the existential one encodes that every slave
 instance should be connected with a master instance through their corresponding ports, and 
the distinct slave instances may apply these interactions consecutively.

\end{exa}

\begin{exa}
\label{str}
\textbf{(Parametric Star)} Star architecture has only one component type with a unique port namely $p$. One instance is considered as the center in the sense that every other instance has to be connected with it. No other interaction is permitted. Figure \ref{star} represents the Star architecture for five instances. 

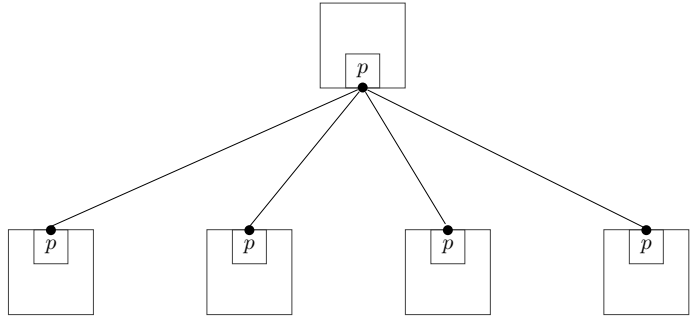
\begin{figure}[h]
\centering
\resizebox{0.6\linewidth}{!}{
\begin{tikzpicture}[>=stealth',shorten >=1pt,auto,node distance=1cm,baseline=(current bounding box.north)]
\tikzstyle{component}=[rectangle,ultra thin,draw=black!75,align=center,inner sep=9pt,minimum size=1.5cm]
\tikzstyle{port}=[rectangle,ultra thin,draw=black!75,minimum size=6mm]
\tikzstyle{bubble} = [fill,shape=circle,minimum size=5pt,inner sep=0pt]
\tikzstyle{type} = [draw=none,fill=none] 

\node [component] (a1) {};
\node [port] (a2) [below=-0.605cm of a1]  {$p$};
\node[bubble] (a3) [below=-0.105cm of a1]   {};

\node [component] (a4) [below left =2.5cm and 4cm of a1]  {};
\node [port] (a5) [above=-0.605cm of a4]  {$p$};
\node[] (i1) [above left=-0.15 cm and -0.73cm of a4]   {};
\node[bubble] (a6) [above=-0.105cm of a4]   {};

\path[-]          (a3)  edge                  node {} (i1);

\node [component] (b4) [right =2 cm of a4]  {};
\node [port] (b5) [above=-0.605cm of b4]  {$p$};
\node[] (i2) [above left=-0.25 cm and -0.30cm of b5]   {};
\node[bubble] (b6) [above=-0.105cm of b4]   {};

\path[-]          (a3)  edge                  node {} (i2);

\node [component] (c4) [right =2 cm of b4]  {};
\node [port] (c5) [above=-0.605cm of c4]  {$p$};
\node[] (i3) [above right=-0.20 cm and -0.38cm of c5]   {};
\node[bubble] (c6) [above=-0.105cm of c4]   {};

\path[-]          (a3)  edge                  node {} (i3);

\node [component] (d4) [right =2 cm of c4]  {};
\node [port] (d5) [above=-0.605cm of d4]  {$p$};
\node[] (i4) [above right=-0.20 cm and -0.27cm of d5]   {};
\node[bubble] (d6) [above=-0.105cm of d4]   {};

\path[-]          (a3)  edge                  node {} (i4);

\end{tikzpicture}}
\caption{Star architecture.}
\label{star}
\end{figure}

\noindent The FOEIL sentence $\psi$ for parametric Star architecture is as follows:
$$ \psi= \exists x^{(1)}\forall^* y^{(1)} (x^{(1)}\neq y^{(1)}). \#(p(x^{(1)}) \wedge p(y^{(1)})). $$

\noindent The 
universal concatenation quantifier preceded by the existential quantifier in $\psi$
 encodes that each of the instances in the architecture should be connected with the center instance
consecutively.

\end{exa}

\begin{exa}
\label{pi-fi}
\textbf{(Parametric Pipes/Filters)} Pipes/Filters architecture involves two types of components, namely pipes and filters \cite{Ga:An}. Pipe (resp. filter) component has an entry  port  $p_e$ and an output port $p_o$ (resp. $f_e,f_o$). Every filter $F$ is connected to two separate pipes $P$ and $P'$ via interactions  $\{f_e, p_o\}$ and $\{f_o,p'_e \}$, respectively. Every pipe $P$ can be connected to at most one filter $F$ via an interaction $\{p_o, f_e\}$. Any other interaction is not permitted. An instantiation of the architecture for four pipe and three filter components is shown in Figure \ref{p-f}. We denote by $\mathcal{X}^{(1)}$ and $\mathcal{X}^{(2)}$ the sets of variables  corresponding to pipe and filter component instances, respectively. The subsequent FOEIL sentence $\psi$ describes the parametric Pipes/Filters architecture.
\begin{multline*}
\psi= \forall^* x^{(2)}\exists x^{(1)} \exists y^{(1)}(x^{(1)} \neq y^{(1)}).\bigg(\#(p_o(x^{(1)}) \wedge f_e(x^{(2)}))\ast \#(p_e(y^{(1)}) \wedge f_o(x^{(2)})) \bigg) \bigwedge     \\
\qquad \quad \bigg ( \forall z^{(1)}\forall y^{(2)}. \bigg( \big(\forall z^{(2)} (y^{(2)} \neq z^{(2)}).  
\theta_1\big) \vee  \theta_2\bigg)\bigg)
\end{multline*}

\noindent where the EPIL formulas $\theta_1$ and $\theta_2$ are given respectively, by:

\

$\theta_1= \big( \big((p_o(z^{(1)}) \wedge f_e(y^{(2)}))\shuffle\mathrm{true}\big) \wedge \big(\neg \big(  (p_o(z^{(1)}) \wedge f_e(z^{(2)}))\shuffle
\mathrm{true}\big)\big)\big)$

\

\noindent and

\

$ \theta_2= \big(\neg\big( (p_o(z^{(1)}) \wedge f_e(y^{(2)}))  \shuffle 
\mathrm{true}\big)\big)$.

\

\noindent In the above sentence, the universal concatenation quantifier ($\forall^* x^{(2)}$) in conjunction with the two existential ones ($\exists x^{(1)}, \exists y^{(1)} $)
 describe that every filter instance is connected with two distinct pipe instances, and these interactions are implemented consecutively. The arguments of $\#$ express the connection of a filter entry (resp. output) port with a pipe output (resp. entry) port excluding by definition erroneous port connections. Then, the subformula after the big conjunction ensures that no more than one filter entry port will be connected to the same pipe output port.

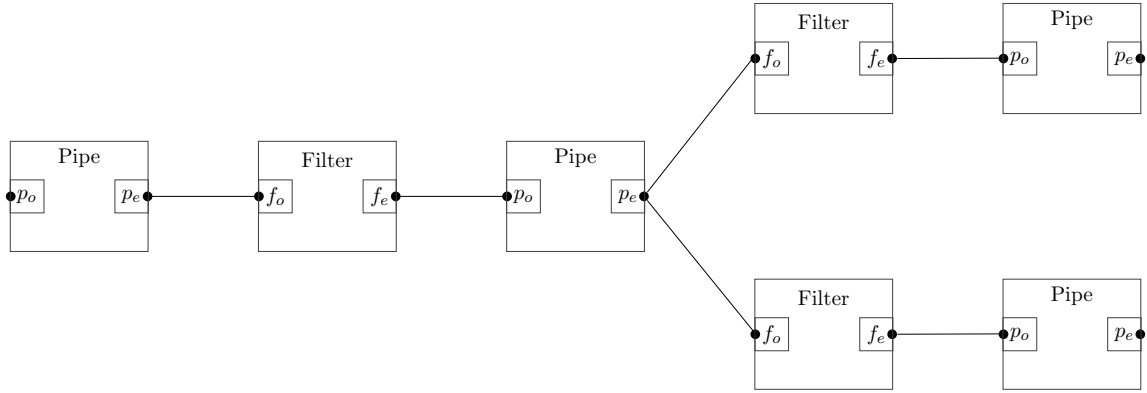
\begin{figure}[h]
\centering
\resizebox{1.0\linewidth}{!}{
\begin{tikzpicture}[>=stealth',shorten >=1pt,auto,node distance=1cm,baseline=(current bounding box.north)]
\tikzstyle{component}=[rectangle,ultra thin,draw=black!75,align=center,inner sep=9pt,minimum size=1.5cm,minimum width=2.5cm,minimum height=2cm]
\tikzstyle{port}=[rectangle,ultra thin,draw=black!75,minimum size=6mm]
\tikzstyle{bubble} = [fill,shape=circle,minimum size=5pt,inner sep=0pt]
\tikzstyle{type} = [draw=none,fill=none]

 \node [component] (a1) {};

\node[bubble] (a2) [right=-0.105cm of a1]   {};   
\node [port] (a3) [right=-0.605cm of a1]  {$p_e$};  

\node[bubble] (a4) [left=-0.105cm of a1]   {};   
\node [port] (a5) [left=-0.615cm of a1]  {$p_o$};  

\node[type]  [above=-0.6cm of a1]{Pipe};

\node [component] (b2) [right=2cm of a1]  {};
\node[bubble] (b3) [right=-0.105cm of b2]   {};   
\node [port] (b4) [right=-0.605cm of b2]  {$f_e$};  

\node[bubble] (b5) [left=-0.105cm of b2]   {};   
\node [port] (b6) [left=-0.615cm of b2]  {$f_o$};  

\node[type]  [above=-0.6cm of b2]{Filter};
\path[-]          (a3)  edge                  node {} (b6);

 \node [component] (c2) [right=2cm of b2]{};

\node[bubble] (c3) [right=-0.105cm of c2]   {};   
\node [port] (c4) [right=-0.605cm of c2]  {$p_e$};  

\node[bubble] (c5) [left=-0.105cm of c2]   {};   
\node [port] (c6) [left=-0.615cm of c2]  {$p_o$};  

\node[type]  [above=-0.6cm of c2]{Pipe};

\path[-]          (b4)  edge                  node {} (c6);

\node [component] (d2) [above right= 0.5cm and 2cm of c2]  {};
\node[bubble] (d3) [right=-0.105cm of d2]   {};  
\node [port] (d4) [right=-0.600cm of d2]  {$f_e$};  

\node[bubble] (d5) [left=-0.105cm of d2]   {};   
\node [port] (d6) [left=-0.615cm of d2]  {$f_o$};  
\node[] (i1) [above left=-0.3 cm and -0.20cm of d6]   {};

\node[type]  [above=-0.6cm of d2]{Filter};
\path[-]          (c3)  edge                  node {} (i1);

\node [component] (e2) [below right= 0.5cm and 2cm of c2]  {};
\node[bubble] (e3) [right=-0.105cm of e2]   {};   
\node [port] (e4) [right=-0.600cm of e2]  {$f_e$};  

\node[bubble] (e5) [left=-0.105cm of e2]   {};   
\node [port] (e6) [left=-0.615cm of e2]  {$f_o$}; 
\node[] (i2) [above left=-0.62 cm and -0.30cm of e6]   {};
\node[type]  [above=-0.6cm of e2]{Filter};
\path[-]          (c3)  edge                  node {} (i2);

 \node [component] (f2) [right=2cm of d2]{};

\node[bubble] (f3) [right=-0.105cm of f2]   {};   
\node [port] (f4) [right=-0.605cm of f2]  {$p_e$}; 

\node[bubble] (f5) [left=-0.105cm of f2]   {};   
\node [port] (f6) [left=-0.615cm of f2]  {$p_o$};  
\node[] (i3) [above left=-0.43 cm and -0.25cm of f6]   {};

\node[type]  [above=-0.6cm of f2]{Pipe};

\path[-]          (d3)  edge                  node {} (i3);

 \node [component] (g2) [right=2cm of e2]{};

\node[bubble] (g3) [right=-0.105cm of g2]   {};   
\node [port] (g4) [right=-0.605cm of g2]  {$p_e$};  

\node[bubble] (g5) [left=-0.105cm of g2]   {};   
\node [port] (g6) [left=-0.615cm of g2]  {$p_o$};  
\node[] (i4) [above left=-0.43 cm and -0.25cm of g6]   {};

\node[type]  [above=-0.6cm of g2]{Pipe};

\path[-]          (e3)  edge                  node {} (i4);

\end{tikzpicture}}
\caption{Pipes/Filters architecture.}
\label{p-f}
\end{figure}
\end{exa}

\begin{exa}
\label{repo}
\textbf{(Parametric Repository)} Repository architecture involves two types of components, namely repository and data accessor \cite{Cl:Do}. Repository component is unique and all data accessors are connected to it. No connection among data accessors exists. Both repository and data accessors have one port each denoted by $p_r,p_d$, respectively. Figure \ref{rep} shows an instantiation of the architecture with four data accessors. 

\begin{figure}[h]
\centering
\resizebox{0.6\linewidth}{!}{
\begin{tikzpicture}[>=stealth',shorten >=1pt,auto,node distance=1cm,baseline=(current bounding box.north)]
\tikzstyle{component}=[rectangle,ultra thin,draw=black!75,align=center,inner sep=9pt,minimum size=1.5cm]
\tikzstyle{port}=[rectangle,ultra thin,draw=black!75,minimum size=6mm]
\tikzstyle{bubble} = [fill,shape=circle,minimum size=5pt,inner sep=0pt]
\tikzstyle{type} = [draw=none,fill=none] 

\node [component] (a1) {};
\node [port] (a2) [below=-0.605cm of a1]  {$p_r$};
\node[bubble] (a3) [below=-0.105cm of a1]   {};

\node [component] (a4) [below left =2.5cm and 4cm of a1]  {};
\node [port] (a5) [above=-0.605cm of a4]  {$p_d$};
\node[] (i1) [above left=-0.15 cm and -0.73cm of a4]   {};
\node[bubble] (a6) [above=-0.105cm of a4]   {};

\path[-]          (a3)  edge                  node {} (i1);

\node [component] (b4) [right =2 cm of a4]  {};
\node [port] (b5) [above=-0.605cm of b4]  {$p_d$};
\node[] (i2) [above left=-0.25 cm and -0.30cm of b5]   {};
\node[bubble] (b6) [above=-0.105cm of b4]   {};

\path[-]          (a3)  edge                  node {} (i2);

\node [component] (c4) [right =2 cm of b4]  {};
\node [port] (c5) [above=-0.605cm of c4]  {$p_d$};
\node[] (i3) [above right=-0.20 cm and -0.38cm of c5]   {};
\node[bubble] (c6) [above=-0.105cm of c4]   {};

\path[-]          (a3)  edge                  node {} (i3);

\node [component] (d4) [right =2 cm of c4]  {};
\node [port] (d5) [above=-0.605cm of d4]  {$p_d$};
\node[] (i4) [above right=-0.20 cm and -0.27cm of d5]   {};
\node[bubble] (d6) [above=-0.105cm of d4]   {};

\path[-]          (a3)  edge                  node {} (i4);

\end{tikzpicture}}
\caption{Repository architecture.}
\label{rep}
\end{figure}
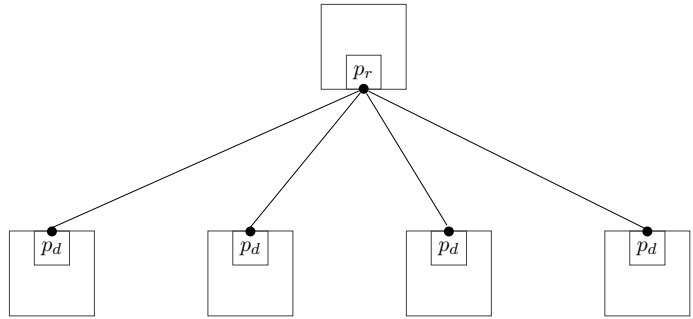

\noindent The subsequent FOEIL sentence $\psi$ characterizes the parametric Repository architecture. Variable set $\mathcal{X}^{(1)}$ refers to  instances of repository component and variable set $\mathcal{X}^{(2)}$ to instances of data accessor component. Then, the use of the
 universal concatenation operator combined with the existential one serves to encode that each of the data accessor instances interact with the repository instance
in consecutive order. 
$$ \psi= \exists x^{(1)} \forall^* x^{(2)}.  \#(p_r(x^{(1)}) \wedge
p_d(x^{(2)})). $$

\noindent 

\end{exa}

Observe that in Examples \ref{ma-sl}-\ref{repo} we use the concatenation quantifier to describe the execution of the interactions for the several component
instances. On the other hand, the shuffle quantifier would not provide the encoding of any different implementation of the corresponding architectures.
This is because the above parametric architectures do not impose order restrictions
on the execution of the permissible interactions. For this reason, we also omit the application of the iteration operator
in these examples. Hence, FOEIL can describe sufficiently parametric architectures with no order restrictions in the allowed interactions.
 Next, we provide three more examples of parametric architectures,
 namely Blackboard, Request/Response, and Publish/Subscribe, where the order of interactions constitutes a main feature.

\begin{exa}
\label{blackboard}
\textbf{(Parametric Blackboard)} The subsequent FOEIL sentence $\psi$ encodes the interactions of Blackboard architecture, described in Example \ref{ex_b_blackboard}, in the parametric setting. We consider three sets of variables, namely $\mathcal{X}^{(1)}, \mathcal{X}^{(2)}, \mathcal{X}^{(3)}$ for the instances of blackboard, controller, and knowledge sources components, respectively.  
\begin{multline*}
\psi=\Bigg(\exists x^{(1)} \exists x^{(2)}.\Bigg( \# (p_d(x^{(1)})\wedge p_r(x^{(2)})) * \bigg(\forall^{\shuffle} x^{(3)}.\# (p_d(x^{(1)})\wedge p_n(x^{(3)}))\bigg) * \\ \bigg(\exists^{\shuffle} y^{(3)}. \big(\#(p_l(x^{(2)})\wedge p_t(y^{(3)})) *  \# (p_e(x^{(2)})\wedge p_w(y^{(3)})\wedge p_a(x^{(1)}))\big) \bigg)^+\Bigg)\Bigg)^{+}.
\end{multline*}

\noindent We interpret the above FOEIL sentence as follows: There exists a blackboard and a controller instance ($\exists x^{(1)}, \exists x^{(2)}$) so that the former informs the
latter for the available data. In turn, the parentheses with the 
universal shuffle quantifier ($\forall^{\shuffle} x^{(3)}$) encodes that all the sources instances are informed from blackboard in arbitrary order. 
Finally, the parenthesis with the existential shuffle quantifier ($\exists^{\shuffle} y^{(3)}$) describes the triggering and writing process of some of the sources  instances, 
implemented with interleaving. Similarly to the non-parametric version of the architecture, the iteration operators describe the repetition
of the corresponding parametric interactions. An instantiation of  parametric Blackboard architecture with three sources is presented in Example \ref{ex_b_blackboard} of Subsection \ref{examples_EPIL}.

\end{exa}

\begin{exa}
\label{re-re}
\textbf{(Parametric Request/Response)} Next we present a FOEIL sentence $\psi$ for Request/Response architecture, described in Example \ref{b-r-r}, in the parametric setting. We consider the variable sets $\mathcal{X}^{(1)},\mathcal{X}^{(2)},\mathcal{X}^{(3)}$, and $\mathcal{X}^{(4)}$  referring to instances of service registry, service, client, and coordinator component, respectively. 
\begin{multline*}
\psi= \bigg( \exists x^{(1)}.\bigg(\big(\forall^{\shuffle} x^{(2)}. \# (p_e(x^{(1)})\wedge p_r(x^{(2)}))\big) *\\ \big(\forall^{\shuffle} x^{(3)}. (\# (p_l(x^{(3)})\wedge p_u(x^{(1)})) * \#(p_o(x^{(3)})\wedge p_t(x^{(1)}))) \big )\bigg)\bigg) * \\ \bigg(\exists^{\shuffle} y^{(2)}\exists x^{(4)}\exists^*y^{(3)}.\xi \wedge \bigg(\forall y^{(4)}\forall z^{(3)}\forall z^{(2)}.\big( \theta \vee \big(\forall t^{(3)}\forall t^{(2)}(z^{(2)}\neq t^{(2)}).\theta'\big)\big)\bigg)\bigg)^+
\end{multline*}
\noindent where the EPIL formulas $\xi$, $\theta$, and $\theta'$ are  given respectively, by:

\

$ \xi= \# (p_n(y^{(3)})\wedge p_m(x^{(4)})) * \# (p_q(y^{(3)})\wedge p_a(x^{(4)})\wedge p_g (y^{(2)})) *\# (p_c(y^{(3)})\wedge p_d(x^{(4)})\wedge p_s(y^{(2)}))$,

\

$\theta=\neg ( (p_q(z^{(3)})\wedge p_a(y^{(4)})\wedge p_g (z^{(2)}))\shuffle \mathrm{true})$,

\

\noindent and

\

$\theta' =( \# (p_q(z^{(3)})\wedge p_a(y^{(4)})\wedge p_g (z^{(2)}))\shuffle\mathrm{true}) \wedge \\
 \text{ \qquad \qquad \qquad \qquad \qquad \qquad \qquad \ \ \ \ } \neg(   (p_q(t^{(3)})\wedge p_a(y^{(4)})\wedge p_g (t^{(2)}))\shuffle\mathrm{true}).$

\

\noindent The first two lines of FOEIL sentence $\psi$ express the connections 
of every
service and client instance with the service registry instance, respectively. Moreover, the several instances of services as well as of clients
interact with the registry instance in arbitrary order and for this, we use the universal shuffle quantifiers $\forall^{\shuffle} x^{(2)}$ and $\forall^{\shuffle} x^{(3)}$, respectively. Then, the last line of $\psi$
captures that for some of the service instances ($\exists^{\shuffle} y{^{(2)}}$) there exist some client instances ($\exists^{\ast} y^{(3)}$) that are interested in the former, and hence are connected 
consecutively through the services' coordinator instance ($\exists x^{(4)}$). Recall that only a unique client instance is allowed to interact with a service instance which justifies the use of the concatenation quantifier $\exists^{\ast} y^{(3)}$. On the other hand, the interactions of the distinct service instances with the interested client instances are implemented with interleaving, since there are no order restrictions from the architecture, which is expressed by the shuffle quantifier $\exists^{\shuffle} y{^{(2)}}$. Then, the subformula $\forall y^{(4)}\forall z^{(3)} \forall z^{(2)} .\big( \theta \vee \big(\forall t^{(3)}\forall t^{(2)}(z^{(2)}\neq t^{(2)}).\theta'\big)\big)$ in $\psi$ serves as a constraint to ensure that a unique coordinator instance is assigned to each service instance. Finally, the application of the iteration operator allows the repetition of the permissible interactions in the parametric architecture. An instantiation of the parametric architecture for two clients and two services is discussed in Example \ref{b-r-r} of Subsection \ref{examples_EPIL}.
\end{exa}

\begin{exa}\textbf{(Parametric Publish/Subscribe)}
\label{pu-su}
We consider Publish/Subscribe architecture, described in Example \ref{b-pu-su}, in the parametric setting. In the subsequent FOEIL sentence $\psi$, we let variable sets $\mathcal{X}^{(1)},\mathcal{X}^{(2)},\mathcal{X}^{(3)}$ correspond to publisher, topic, and subscriber component instances, respectively. 
\begin{multline*}
\psi= \Bigg( \exists^{\shuffle} x^{(2)}. \Bigg( \bigg( \exists^{\shuffle} x^{(1)}.\big(\# (p_a(x^{(1)})\wedge p_n (x^{(2)}))* \# (p_t(x^{(1)})\wedge p_r(x^{(2)}))\big) \bigg)  
* \\ \bigg(\exists^{\shuffle} x^{(3)}. \big( \# (p_e(x^{(3)})\wedge p_c(x^{(2)}))*\#(p_g(x^{(3)})\wedge p_s (x^{(2)})) *  \#(p_d(x^{(3)})\wedge p_f (x^{(2)}))\big)\bigg)\Bigg)\Bigg)^+.
\end{multline*}

\noindent The FOEIL sentence $\psi$ is interpreted as follows: The big parenthesis with the existential shuffle
quantifier $\exists^{\shuffle} x^{(1)}$ preceded by $\exists^{\shuffle} x^{(2)}$ describes that given some topic instances, some of the 
publisher instances advertise and in turn transmit their messages to the former with interleaving. For the same topic instances, in turn, the parenthesis at the second line with the
existential shuffle quantifier $\exists^{\shuffle} x^{(3)}$ encodes the arbitrary order among some of the subscriber instances for executing three types of consecutive interactions, namely the connection with the
interested subscriber instance, the transfer of the message and their disconnection. The communication steps described above are implemented 
for the distinct topic instances with interleaving (captured by the first existential shuffle quantifier $\exists^{\shuffle} x^{(2)}$). Finally, the
iteration operator is applied to the whole sentence  in order to describe the recursion of the aforementioned interactions, 
and hence models the subsequent implementation of the architecture within the parametric system. Example \ref{b-pu-su} presented in Subsection \ref{examples_EPIL} is an instantiation of the parametric Publish/Subscribe architecture for two publishers, two topics, and three subscribers.

\end{exa}

In \cite{Ma:Co} the authors described a simpler version of Request/Response and Blackboard architectures. Though the resulting
sets of interactions do not depict the order in which they should be performed (cf. Remark \ref{remark}). Publish/Subscribe architecture has not been considered
 in the related work \cite{Bo:St,Ko:Pa,Ma:Co}. A weighted version of Publish/Subscribe architecture was described by a weighted propositional configuration logic formula in \cite{Pa:On}. Nevertheless, even if we consider that formula without weights, it is not possible to express the required order of the implementation of the interactions. Moreover,  the versions of the parametric architectures studied in \cite{Bo:St,Ko:Pa,Ma:Co,Pa:On} do not allow recursive interactions. Our Examples \ref{blackboard}, \ref{re-re}, and \ref{pu-su} show that for a parametric component-based system with any of the aforementioned architectures, the semantics of the corresponding FOEIL formula encodes the required order of recursive interactions.

\section{Decidability results for FOEIL}\label{sec_dec}
In this section, we prove decidability results for FOEIL  sentences. Specifically, we show that the equivalence and validity  problems for FOEIL sentences are decidable in doubly exponential time, whereas the satisfiability problem is decidable in exponential time. For this, we establish an effective translation of every FOEIL formula to an expressive equivalent finite automaton, and hence we take advantage of well-known computational results for finite automata. For the reader's convenience we briefly recall basic notions and results on finite automata. 

Let $A$ be an alphabet. A (nondeterministic) finite automaton (NFA for short) over $A$ is a five-tuple $\mathcal{A}=(Q, A, I, \Delta, F )$ where $Q$ is the finite state set, $I \subseteq Q$ is the set of initial states, $\Delta \subseteq Q\times A \times Q$ is the set of transitions, and $F$ is the final state set.  

Let $w=a_1 \ldots a_n \in A^*$. A path of $\mathcal{A}$ over $w$ is a sequence of transitions $((q_{i-1}, a_i, q_i))_{1 \leq i \leq n}$. The path is called successful if $q_0 \in I$ and $q_n \in F$. A word $w\in A^*$ is accepted (or recognized) by $\mathcal{A}$ if there a successful path of $\mathcal{A}$ over $w$. The language $L(\mathcal{A})$ of $\mathcal{A}$ is the set of all words accepted by $\mathcal{A}$. 

The finite automaton $\mathcal{A}$ is called deterministic (DFA for short) (resp. complete) if $I=\{q_0\}$ and for every $q\in Q$ and $a \in A$ there is at most (resp. exactly) one state $q'\in Q$ such that $(q,a,q') \in \Delta$. In this case we write $\mathcal{A}=(Q, A, q_0, \Delta, F)$. Two finite automata $\mathcal{A}$ and $\mathcal{A}'$ over $A$ are called equivalent if $L(\mathcal{A})=L(\mathcal{A}')$.  For our translation algorithm of FOEIL formulas to finite automata we shall need folklore results in automata theory. We collect them in the following proposition (cf. for instance \cite{Kh:Au, Sa:El, Si:In}). 

\begin{prop}\label{prop-rec}
\begin{itemize}
\hfill
\item[\emph{1)}] Let $\mathcal{A}=(Q, A, I, \Delta, F )$ be an \emph{NFA} over $A$. Then, we can construct an equivalent complete finite automaton $\mathcal{A}'$ with state set $\mathcal{P}(Q)$ over $A$. The run time of the  algorithm is exponential.
\item[\emph{2)}] Let $\mathcal{A}_1=(Q_1, A, I_1, \Delta_1, F_1 )$ and $\mathcal{A}_2=(Q_2, A, I_2, \Delta_2, F_2 )$ be two \emph{NFA's} over $A$. Then, the intersection  $L(\mathcal{A}_1) \cap L(\mathcal{A}_2)$ is accepted by the \emph{NFA} $\mathcal{A}=(Q_1\times Q_2, A, I_1 \times I_2, \Delta, F_1 \times F_2)$ where $\Delta=\{((q_1,q_2),a,(q'_1,q'_2)) \mid (q_1,a,q'_1) \in \Delta_1, (q_2,a,q_2')\in \Delta_2\}$. If $\mathcal{A}_1$ and $\mathcal{A}_2$ are \emph{DFA's}, then $\mathcal{A}$ is also a \emph{DFA}. The finite automaton $\mathcal{A}$ is called the product automaton of $\mathcal{A}_1$ and $\mathcal{A}_2$. \\
The union $L(\mathcal{A}_1) \cup L(\mathcal{A}_2)$ is accepted by the \emph{NFA} $\mathcal{A}'=(Q_1 \cup Q_2, A, I_1 \cup I_2, \Delta_1 \cup \Delta_2, F_1 \cup F_2)$ where without loss of generality we assume that $Q_1 \cap Q_2 = \emptyset$. The \emph{NFA} $\mathcal{A}'$ is called the disjoint union of $\mathcal{A}_1$ and $\mathcal{A}_2$.
\item[\emph{3)}] If $\mathcal{A}_1$,  $\mathcal{A}_2$, and $\mathcal{A}$ are finite automata over $A$, then we can construct, from $\mathcal{A}_1$, $\mathcal{A}_2$, and $\mathcal{A}$ \emph{NFA's} $\mathcal{B}$, $\mathcal{C}$, and $\mathcal{D}$ accepting respectively, the languages $L(\mathcal{A}_1)*L(\mathcal{A}_2)$, $L(\mathcal{A}_1) \shuffle L(\mathcal{A}_2)$, and $L(\mathcal{A})^+$. The run time for all the constructions is polynomial.
\item[\emph{4)}] Let $\mathcal{A}_1=(Q_1, A, q_{0,1}, \Delta_1, F_1 )$ and $\mathcal{A}_2=(Q_2, A, q_{0,2}, \Delta_2, F_2 )$ be two complete finite automata over $A$. Then, the union  $L(\mathcal{A}_1) \cup L(\mathcal{A}_2)$ is accepted by the complete finite automaton $\mathcal{A}=(Q_1\times Q_2, A, (q_{0,1},q_{0,2}), \Delta, (Q_1 \times F_2) \cup (F_1 \times Q_2))$ where $\Delta=\{((q_1,q_2),a,(q'_1,q'_2)) \mid (q_1,a,q'_1) \in \Delta_1, (q_2,a,q_2')\in \Delta_2\}$. 
\item[\emph{5)}] Let $\mathcal{A}=(Q,A,q_0,\Delta, F)$ be a \emph{DFA} over $A$. Then, we can construct an equivalent complete finite automaton $\mathcal{A}'=(Q \cup \{\bar{q}\}, A, q_0, \Delta', F)$ where $\bar{q}$ is a new state and $\Delta'=\Delta \cup \{(q,a,\bar{q}) \mid \text{ there is no state } q' \in Q \text{ such that } (q,a,q') \in \Delta \} \cup \{(\bar{q},a,\bar{q}) \mid a \in A\}$.
\item[\emph{6)}] Let $\mathcal{A}=(Q,A,q_0,\Delta, F)$ be a complete finite automaton over $A$. Then, the complement of $L(\mathcal{A})$ is accepted by the complete finite automaton $\bar{\mathcal{A}}=(Q,A,q_0,\Delta, Q\setminus F)$.
\item[\emph{7)}] Let $\mathcal{A}=(Q, A, I, \Delta, F )$ be an \emph{NFA} over $A$. Then, we can decide in linear time whether $L(\mathcal{A})=\emptyset$  or not (emptiness problem).
\item[\emph{8)}] Let $\mathcal{A}=(Q, A, I, \Delta, F )$ be an \emph{NFA} over $A$. Then, we can decide in exponential time  whether $L(\mathcal{A})=A^*$  or not (universality problem).
\end{itemize}
\end{prop}

\noindent Next we present the translation algorithm of  FOEIL sentences to finite automata. Our algorithm requires an exponential time  at its worst case.  Specifically, we state the following theorem.

\begin{thm}\label{sent_to_aut}
Let $\psi$ be a \emph{FOEIL} sentence over a set $p\B=\{B(i,j) \mid i\in [n], j\geq 1\}$ of parametric components and $r:[n] \rightarrow \mathbb{N}$. Then, we can effectively construct a finite automaton $\mathcal{A}_{\psi,r}$ over $I_{p\B(r)}$ such that $(r,w) \models \psi$ iff $w \in L(\mathcal{A}_{\psi,r})$ for every  $w \in I_{p\B(r)}^+$. The worst case run time for the translation algorithm is exponential and the best case is polynomial.    
\end{thm}

We shall prove Theorem \ref{sent_to_aut} using the subsequent proposition. For this, we need the following notations. Let $\mathcal{V} \subseteq \mathcal{X}$ be a finite set of variables. For every $i \in [n]$ and $x^{(i)} \in \mathcal{V}$, we define the set $P(i)(x^{(i)})= \{p(x^{(i)}) \mid p\in P(i) \text{ and } x^{(i)} \in \mathcal{V} \}$ and let $I_{p\B(\mathcal{V})}=\{ a \in I(P_{p\B(\mathcal{V})}) \mid  \vert a \cap P(i)(x^{(i)}) \vert \leq 1  \text{ for every } i\in[n] \text{ and } x^{(i)} \in \mathcal{V} \}$. Next let   $\sigma$  be a $(\mathcal{V}, r)$ assignment and $L$ a language  over $I_{p\B(\mathcal{V})}$. We shall denote by $\sigma(L)$ the language over $I(P_{p\B(r)})$\footnote{The language $\sigma(L)$ is not always over $I_{p\B(r)}$. For instance,  assume that $a \in L$, with $a \in I_{p\B(\mathcal{V})}$, $p(x^{(i)}), p'(y^{(i)}) \in a $ for some $i\in [n]$, $p,p' \in P(i)$, and $\sigma(x^{(i)})=\sigma(y^{(i)})$. Then  $\sigma(a) \notin I_{p\B(r)}$.} which is obtained by $L$ by replacing every  variable $x\in \mathcal{V}$ by $\sigma(x)$.

\begin{prop}\label{formula-aut}
Let $\psi$ be a \emph{FOEIL} formula over a set $p\B=\{B(i,j) \mid i\in [n], j\geq 1\}$ of parametric components. Let also $\mathcal{V} \subseteq \mathcal{X}$ be a finite set of variables containing $\mathrm{free}(\psi)$ and $r:[n] \rightarrow \mathbb{N}$. Then, we can effectively construct a finite automaton $\mathcal{A}_{\psi,r}$ over $I_{p\B(\mathcal{V})}$ such that for every $(\mathcal{V}, r)$-assignment $\sigma$ and $w \in I_{p\B(r)}^+$ we have $(r, \sigma, w) \models \psi$ iff $w \in \sigma(L(\mathcal{A}_{\psi,r})) \cap I^+_{p\B(r)}$. The worst case run time for the translation algorithm is exponential and the best case is polynomial.   
\end{prop}
\begin{proof}
We prove our claim by  induction on the structure of the FOEIL formula $\psi$. 
\begin{itemize}
\item[i)] If $\psi=\mathrm{true}$, then we consider the complete finite automaton $\mathcal{A}_{\psi,r}=(\{q\}, I_{p\B(\mathcal{V})}, q, \Delta , \{q\})$ where $\Delta = \{(q,a,q) \mid a \in I_{p\B(\mathcal{V})} \}$. 
\item[ii)] If $\psi=p(x^{(i)})$, then we construct the DFA $\mathcal{A}_{\psi,r}=(\{q_0, q_1\}, I_{p\B(\mathcal{V})}, q_0,  \Delta, \{q_1\})$ with $\Delta= \{(q_0,a,q_1) \mid a \in I_{p\B(\mathcal{V})} \text{ and }  p(x^{(i)}) \in a \}$.
\item[iii)] If $\psi=\neg \phi'$ or $\psi=\phi_1 \vee \phi_2$ where $\phi', \phi_1, \phi_2$ are PIL formulas, and $\mathcal{A}_{\phi',r}$, $\mathcal{A}_{\phi_1,r}$, and $\mathcal{A}_{\phi_2,r}$ are DFA's, then we construct $\mathcal{A}_{\psi,r}$ by applying Proposition \ref{prop-rec}(5),(6) and (4), respectively. Trivially, the finite automaton $\mathcal{A}_{\psi, r}$ is deterministic.
\item[iv)] If $\psi= \zeta_1 * \zeta_2$, then we construct $\mathcal{A}_{\psi,r}$ by taking the finite automaton accepting the Cauchy product of $L(\mathcal{A}_{\zeta_1,r})$ and $L(\mathcal{A}_{\zeta_2,r})$ (cf. Proposition \ref{prop-rec}(3)).
\item[v)] If $\psi= \zeta_1 \shuffle \zeta_2$, then we construct $\mathcal{A}_{\psi,r}$ by taking the finite automaton accepting the shuffle product of $L(\mathcal{A}_{\zeta_1,r})$ and $L(\mathcal{A}_{\zeta_2,r})$ (cf. Proposition \ref{prop-rec}(3)).
\item[vi)] If $\psi= \neg \zeta$, then we construct firstly, from $\mathcal{A}_{\zeta, r}$, an equivalent complete finite automaton $\mathcal{A}'$ (cf. Proposition \ref{prop-rec}(1)). Then, we get $\mathcal{A}_{\psi,r}$ from $\mathcal{A}'$ by applying Proposition \ref{prop-rec}(6).
\item[vii)] If $\psi=\varphi_1 \vee \varphi_2$ or $\psi=\varphi_1 \wedge \varphi_2$ or $\psi=\varphi_1 * \varphi_2$ or $\psi=\varphi_1 \shuffle \varphi_2$ or $\psi=\varphi^+$ where $\varphi_1, \varphi_2, \varphi$ are EPIL formulas, then we construct $\mathcal{A}_{\psi,r}$ by taking respectively, the disjoint union and the product automaton of $\mathcal{A}_{\varphi_1,r}$ and $\mathcal{A}_{\varphi_2,r}$, the NFA accepting the Cauchy product of $L(\mathcal{A}_{\varphi_1,r})$ and $L(\mathcal{A}_{\varphi_2,r})$, the NFA accepting the shuffle product of $L(\mathcal{A}_{\varphi_1,r})$ and $L(\mathcal{A}_{\varphi_2,r})$, and the NFA accepting the iteration of $L(\mathcal{A}_{\varphi,r})$.  
\item[viii)] If $\psi=x^{(i)}=y^{(i)}$, then we consider the DFA  $\mathcal{A}_{\psi,r}=(\{q\}, I_{p\B(\mathcal{V})}, q, \Delta , \{q\})$ where $\Delta = \{(q,a,q) \mid a \in I_{p\B(\mathcal{V})}  \text{ such that  for every } p,p' \in P(i) \text{ if } p(x^{(i)}) \in a, \text{ then } p'(y^{(i)}) \notin a \}$.
\item[ix)] If $\psi=\neg (x^{(i)}=y^{(i)})$, then we get firstly a complete finite automaton by the DFA corresponding to formula $x^{(i)}=y^{(i)}$ (Proposition \ref{prop-rec}(5)). Then, we obtain $\mathcal{A}_{\psi,r}$ by applying Proposition \ref{prop-rec}(6).
\item[x)] If  $\psi=\psi_1 \vee \psi_2$ or $\psi=\psi_1 \wedge \psi_2$ or $\psi=\psi_1 * \psi_2$ or $\psi=\psi_1 \shuffle \psi_2$ or $\psi=\psi'^+$ where $\psi_1, \psi_2, \psi'$ are FOEIL formulas, then we apply the same arguments as in (vii) for EPIL formulas.
\item[xi)] If $\psi=\exists x^{(i)}. \psi'$, then we get $\mathcal{A}_{\psi,r}$ as the disjoint union of the finite automata $\mathcal{A}_{\psi',r}^{(j)}$, $j\in [r(i)]$, where $\mathcal{A}_{\psi',r}^{(j)}$ is obtained by $\mathcal{A}_{\psi',r}$ by replacing $x^{(i)}$ by $j$ in $I_{p\B(\mathcal{V})}$.
\item[xii)] If $\psi=\forall x^{(i)}. \psi'$, then we get $\mathcal{A}_{\psi,r}$ as the product automaton of the finite automata $\mathcal{A}_{\psi',r}^{(j)}$, $j\in [r(i)]$, where $\mathcal{A}_{\psi',r}^{(j)}$ is obtained by $\mathcal{A}_{\psi',r}$ by replacing $x^{(i)}$ by $j$ in $I_{p\B(\mathcal{V})}$.
\item[xiii)] If $\psi=\exists ^* x^{(i)}.\psi'$, then we compute firstly all nonempty  subsets $J$ of $[r(i)]$. For every such subset $J=\{l_1, \ldots, l_k\}$, with $1 \leq k \leq r(i)$ and $1 \leq l_1 <  \ldots <  l_k \leq r(i)$,  we consider the NFA $\mathcal{A}_{\psi,r}^{(J)}$ accepting the Cauchy product of the languages $L(\mathcal{A}_{\psi',r}^{(l_1)}),\ldots, L(\mathcal{A}_{\psi',r}^{(l_k)})$ where $\mathcal{A}_{\psi',r}^{(j)}$, $j\in J$, is obtained by $\mathcal{A}_{\psi',r}$ by replacing $x^{(i)}$ by  $j$ in $I_{p\B(\mathcal{V})}$. Then, we get $\mathcal{A}_{\psi,r}$ as the disjoint union of all finite automata $\mathcal{A}_{\psi',r}^{(J)}$  with $\emptyset \neq J \subseteq [r(i)]$. 
\item[xiv)] If $\psi=\forall^* x^{(i)}. \psi'$, then we get $\mathcal{A}_{\psi,r}$ as the NFA accepting the Cauchy product of the languages $L(\mathcal{A}_{\psi',r}^{(1)}), \ldots, L(\mathcal{A}_{\psi',r}^{(r(i))})$ where $\mathcal{A}_{\psi'r}^{(j)}$, $j\in [r(i)]$, is obtained by $\mathcal{A}_{\psi',r}$ by replacing $x^{(i)}$ by  $j$ in $I_{p\B(\mathcal{V})}$. 
\item[xv)] If $\psi=\exists ^{\shuffle} x^{(i)}. \psi'$, then we compute firstly all nonempty subsets $J$ of $[r(i)]$. For every such subset $J=\{l_1, \ldots, l_k\}$, with $1 \leq k \leq r(i)$ and $1 \leq l_1 <  \ldots <  l_k \leq r(i)$,   we consider the NFA $\mathcal{A}_{\psi,r}^{(J)}$ accepting the shuffle product of the languages $L(\mathcal{A}_{\psi',r}^{(l_1)}),\ldots, L(\mathcal{A}_{\psi',r}^{(l_k)})$ where $\mathcal{A}_{\psi',r}^{(j)}$, $j\in J$, is obtained by $\mathcal{A}_{\psi',r}$ by replacing $x^{(i)}$ by  $j$ in $I_{p\B(\mathcal{V})}$.  Then, we get $\mathcal{A}_{\psi,r}$ as the disjoint union of all finite automata $\mathcal{A}_{\psi',r}^{(J)}$ with $\emptyset \neq J \subseteq [r(i)]$.
\item[xvi)] If $\psi=\forall^{\shuffle} x^{(i)}. \psi'$, then we get $\mathcal{A}_{\psi,r}$ as the NFA accepting the shuffle product of the languages $L(\mathcal{A}_{\psi',r}^{(1)}), \ldots, L(\mathcal{A}_{\psi',r}^{(r(i))})$ where $\mathcal{A}_{\psi',r}^{(j)}$, $j\in [r(i)]$, is obtained by $\mathcal{A}_{\psi',r}$ by replacing $x^{(i)}$ by  $j$ in $I_{p\B(\mathcal{V})}$.
\end{itemize}
By our constructions above, we immediately get that for every $(\mathcal{V}, r)$-assignment $\sigma$ and $w \in I_{p\B(r)}^+$ we have $(r, \sigma, w) \models \psi$ iff $w \in \sigma(L(\mathcal{A}_{\psi,r})) \cap I^+_{p\B(r)}$. Hence, it remains to deal with the time complexity of our translation algorithm. 
  
Taking into account the above induction steps, we show that the worst case run time for our translation algorithm is exponential. Indeed, if $\psi$ is a PIL formula, then the constructed finite automaton $\mathcal{A}_{\psi,r}$  in steps (i)-(iii) is a DFA and its state-size is polynomial in the size of $\psi$. Specifically, if $\psi=p(x^{(i)})$, then the state-size of $\mathcal{A}_{\psi,r}$ is the same as the size of $\psi$. If $\psi=\neg \zeta$, then by step (vi) we need an exponential time since we construct firstly a complete finite automaton equivalent to the one corresponding to $\zeta$. The cases (iv), (v), (vii)-(xii), (xiv), and (xvi) trivially require polynomial time constructions. Finally, the translations in steps (xiii) and (xv) require an exponential run time since we need to compute all nonempty subsets of $[r(i)]$; every such subset corresponds to a polynomial time construction of a finite automaton. Moreover, this is an upper bound for the complexity in that steps because of the following reasons. Firstly, the FOEIL formula $\psi'$ contains no EPIL subformulas of the form $\neg \zeta$. Secondly, if $\psi'$ contains a subformula of the form $\exists ^* x^{(i')}.\psi''$ or $\exists ^{\shuffle} x^{(i')}.\psi''$, then the computation of the subsets of $[r(i')]$ is independent of the computation of the subsets of $[r(i)]$. \\ 
On the other hand, the best case run time of the algorithm is polynomial since if $\psi$ needs no translation of steps (vi), (xiii), or (xv), then no exponential blow up occurs.
\end{proof}

\

\begin{proof}[Proof of Theorem \ref{sent_to_aut}] We apply Proposition \ref{formula-aut}. Since $\psi$ is a sentence it contains no free variables. Hence, we get a finite automaton $\mathcal{A}_{\psi,r}$ over $I_{p\B(r)}$ such that $(r,w) \models \psi$ iff $w \in L(\mathcal{A}_{\psi,r})$ for every $w \in I^+_{p\B(r)}$, and this concludes our proof.  
\end{proof}

\

We should note that several of the constructions described in the proof of Proposition \ref{formula-aut} can be simplified according to the form of the FOEIL sentence $\psi$. For instance consider two ports $p,p'$ and the EPIL formula $\varphi=\# (p \wedge p')$. Clearly $\varphi$ is satisfied only by the interaction $a=\{p,p'\}$ which in turns implies that we get in  a straightforward way a finite automaton for $\varphi$. We clarify this in the following example.

\begin{exa}
We consider the parametric Repository architecture  and its corresponding FOEIL sentence 
$ \psi= \exists x^{(1)} \forall^* x^{(2)}.  \#(p_r(x^{(1)}) \wedge
p_d(x^{(2)}))$ (cf. Example \ref{repo}). We let $r(1)=1$ and $r(2)=3$, i.e., we assume one repository component instance, as required, and three instances of the data accessor component. Hence, the set of available ports of the architecture w.r.t. $r$ is $\{p_r(1), p_d(1), p_d(2), p_d(3)  \}$. We consider the interactions $a_1=\{p_r(1), p_d(1)  \}$, $a_2=\{p_r(1), p_d(2)  \}$, and $a_3=\{p_r(1), p_d(3)  \}$ and  the DFA's $\mathcal{A}_1 =(\{q_{1,0}, q_{1,1} \}, \{a_1 \}, q_{1,0}, \{(q_{1,0},a_1, \\ q_{1,1})\}, \{q_{1,1} \} )$, $\mathcal{A}_2 =(\{q_{2,0}, q_{2,1} \}, \{a_2 \}, q_{2,0}, \{(q_{2,0},a_2, q_{2,1})\}, \{q_{2,1} \} )$, and $\mathcal{A}_3 =(\{q_{3,0}, q_{3,1} \}, \\ \{a_3 \}, q_{3,0}, \{(q_{3,0},a_3, q_{3,1})\}, \{q_{3,1} \} )$. Trivially we get $L(\mathcal{A}_1) =\{a_1\}$, $L(\mathcal{A}_2) =\{a_2\}$, and $L(\mathcal{A}_3) =\{a_3\}$, hence we can consider the normalized DFA's $\mathcal{A}_1$, $\mathcal{A}_2$, and $\mathcal{A}_3$ (cf. for instance \cite{Kh:Au, Sa:El, Si:In}) as the ones corresponding to EPIL formulas $\#(p_r(1) \wedge
p_d(1))$, $\#(p_r(1) \wedge
p_d(2))$, and $\#(p_r(1) \wedge
p_d(3))$, respectively. Then, we construct the DFA $\mathcal{A}_{\psi,r}$, depicted in Figure  \ref{exam_aut},  which corresponds to FOEIL sentence $\psi$ w.r.t. $r$.
\begin{figure}[h]
  \centering
 \begin{tikzpicture}[>=stealth',shorten >=1pt,auto,node distance=3cm]
  \node[initial,state,initial text={}]    (q10)                   {$q_{1,0}$};
  \node[state]            (q) [right of=q10]  {$q$};
  \node[state]            (q') [right of=q]  {$q'$};
  \node[state, accepting] (q31) [right of=q']  {$q_{3,1}$}; 
  \path[->] 
      (q10)  edge   node [above]  {$a_1$} (q)
      (q)   edge   node [above] {$a_2$} (q')
      (q')   edge   node [above] {$a_3$} (q31);
\end{tikzpicture}
\caption{The DFA $\mathcal{A}_{\psi,r}$ corresponding to sentence $\psi$ w.r.t. $r$.}
\label{exam_aut}
\end{figure} 
\end{exa}

\

Now we are ready to state the decidability result of the equivalence of FOEIL sentences. Specifically, we prove the following theorem.

\begin{thm}
Let $p\mathcal{B}= \{B(i,j) \mid i \in [n], j \geq 1 \} $ be a set of parametric components and  $r:[n] \rightarrow \mathbb{N}$ a mapping. Then, the equivalence  problem for \emph{FOEIL} sentences over $p\B$ w.r.t. $r$ is decidable in doubly exponential time. 
\end{thm} 
\begin{proof}
Let $\psi_1, \psi_2$ be FOEIL sentences over $p\B$. Then, by Theorem \ref{sent_to_aut} we  construct, in exponential time, finite automata $\mathcal{A}_{\psi_1,r}$ and $\mathcal{A}_{\psi_2,r}$ such that $(r,w) \models \psi_i$ iff $w \in L(\mathcal{A}_{\psi_i,r})$ for every $w \in I^+_{p\B(r)}$ and $i=1,2$. The finite automata $\mathcal{A}_{\psi_1,r}$ and $\mathcal{A}_{\psi_2,r}$ are in general nondeterministic, hence by Proposition \ref{prop-rec}(1) we construct complete finite automata $\mathcal{A}'_{\psi_1,r}$ and $\mathcal{A}'_{\psi_2,r}$  equivalent to $\mathcal{A}_{\psi_1,r}$ and $\mathcal{A}_{\psi_2,r}$, respectively. In this construction another exponential blow up occurs. Finally, the decidability of equivalence of the complete finite automata $\mathcal{A}'_{\psi_1,r}$ and $\mathcal{A}'_{\psi_2,r} $ requires  a linear time (cf. pages 143--145 in  \cite{Ah:Th}), and we are done.
\end{proof}

\

Next, we deal with the decidability of satisfiability and validity results for FOEIL sentences.  For this, we recall firstly these notions. More precisely, a FOEIL sentence $\psi$ over $p\B$ is called \emph{satisfiable w.r.t.} $r$ whenever there exists a $w  \in I_{p\B(r)}^+$ such that $(r,w) \models \psi$, and \emph{valid w.r.t.} $r$ whenever $(r,w) \models \psi$ for every $w  \in I_{p\B(r)}^+$.   

\begin{thm}
Let $p\mathcal{B}= \{B(i,j) \mid i \in [n], j \geq 1 \} $ be a set of parametric components and  $r:[n] \rightarrow \mathbb{N}$ a mapping. Then, the satisfiability problem for \emph{FOEIL} sentences over $p\B$ w.r.t. $r$ is decidable in  exponential time. 
\end{thm}
\begin{proof}
Let $\psi$ be a FOEIL sentence over $p\B$. By Theorem \ref{sent_to_aut} we  construct, in exponential time, an NFA $\mathcal{A}_{\psi,r}$ such that $(r,w) \models \psi$ iff $w \in L(\mathcal{A}_{\psi,r})$ for every $w \in I^+_{p\B(r)}$. Then, $\psi$ is satisfiable iff $L(\mathcal{A}_{\psi,r}) \neq \emptyset$ which is decidable in linear time (Proposition \ref{prop-rec}(7)), and this concludes our proof.    
\end{proof}

\begin{thm}
Let $p\mathcal{B}= \{B(i,j) \mid i \in [n], j \geq 1 \} $ be a set of parametric components and  $r:[n] \rightarrow \mathbb{N}$ a mapping. Then, the validity problem for \emph{FOEIL} sentences over $p\B$ w.r.t. $r$ is decidable in doubly exponential time. 
\end{thm}
\begin{proof}
Let $\psi$ be a FOEIL sentence over $p\B$. By Theorem \ref{sent_to_aut} we  construct, in exponential time, an NFA $\mathcal{A}_{\psi,r}$ such that $(r,w) \models \psi$ iff $w \in L(\mathcal{A}_{\psi,r})$ for every $w \in I^+_{p\B(r)}$. Then, $\psi$ is valid iff  $L(\mathcal{A}_{\psi,r}) = I^+_{p\B(r)} $ which is decidable in exponential time (Proposition \ref{prop-rec}(8)). Hence, we can decide whether $\psi$ is valid or not in doubly exponential time. 
\end{proof}

\section{Parametric weighted architectures}\label{weighted_part}

\subsection{Weighted EPIL and component-based systems}
\label{wEPIL_sec}
In this subsection, we introduce the notion of weighted EPIL over a set of ports $P$ and the semiring $K$. Furthermore, we  define weighted component-based systems and provide examples of architectures in the weighted setup.

\begin{defi}
Let $P$ be a finite set of ports. Then the syntax of weighted EPIL (wEPIL for short) formulas over $P$ and $K$ is given by the grammar
$$\tilde{\varphi}::= k \mid \varphi \mid \tilde{\varphi} \oplus  \tilde{\varphi} \mid \tilde{\varphi} \otimes  \tilde{\varphi} \mid \tilde{\varphi} \odot  \tilde{\varphi} \mid \tilde{\varphi} \varpi  \tilde{\varphi} \mid \tilde{\varphi}^{\oplus}  $$ 
where $k \in K$, $\varphi$ is an EPIL formula over $P$, and $\oplus$, $\otimes$, $\odot$, $\varpi$, and $^{\oplus}$ are the weighted disjunction, conjunction, concatenation, shuffle, and iteration operator, respectively.
\end{defi}

If $\tilde{\varphi}$ is composed by elements in $K$ and PIL formulas connected with $\oplus$ and $\otimes$ operators only, then it is called also a weighted PIL (wPIL for short) formula over $P$ and $K$ \cite{Pa:On,Pa:We} and it will be denoted also by $\tilde{\phi}$. The binding strength, in decreasing order, of the operators in wEPIL is the following: weighted iteration, weighted shuffle, weighted concatenation, weighted conjunction, and weighted disjunction. 

For the semantics of wEPIL formulas  we consider finite words $w$ over $I(P)$ and interpret wEPIL formulas as series in $K \left \langle \left \langle I(P)^+ \right \rangle \right \rangle $. 

\begin{defi}\label{wsem-epil}
Let $\e$ be a wEPIL formula over $P$ and $K$. Then the semantics of $\e$ is a series $\left \Vert \e \right \Vert \in K \left \langle \left \langle I(P)^+ \right \rangle \right \rangle$.  For every $w \in I(P)^+$ the value $\left \Vert \e \right \Vert (w)$ is defined inductively on the structure of $\e$ as follows: 
\begin{itemize}
\item[-] $\left \Vert k \right \Vert(w) = k$,

\item[-] $\left \Vert \varphi \right \Vert(w) = \left\{
\begin{array}
[c]{rl}%
1 & \textnormal{ if }w\models \varphi\\
0 & \textnormal{ otherwise}%
\end{array}
,\right.  $

\item[-] $\left \Vert \e_1 \oplus \e_2 \right \Vert(w)=\left \Vert \e_1  \right \Vert(w)+ \left \Vert  \e_2 \right \Vert(w)$,

\item[-] $\left \Vert \e_1 \otimes \e_2 \right \Vert(w)=\left \Vert \e_1  \right \Vert(w) \cdot \left \Vert  \e_2 \right \Vert(w)$,

\item[-] $\left \Vert \e_1 \odot \e_2 \right \Vert(w)=\sum\limits_{w=w_1w_2} ( \left \Vert \e_1  \right \Vert(w_1) \cdot \left \Vert  \e_2 \right \Vert(w_2))$,

\item[-] $\left \Vert \e_1 \varpi \e_2 \right \Vert(w)=\sum\limits_{w \in w_1 \shuffle w_2} ( \left \Vert \e_1  \right \Vert(w_1) \cdot \left \Vert  \e_2 \right \Vert(w_2))$,

\item[-] $\left \Vert \e^{\oplus} \right \Vert(w)=\sum\limits_{ \nu \geq 1} \left \Vert \e \right \Vert^{\nu}$.

\end{itemize}
\end{defi}

\noindent By definition, the series $\Vert \e \Vert$ is proper for every wEPIL formula $\e$ over $P$ and $K$.
Two wEPIL formulas $\e_1,\e_2$ over $P$ and $K$ are called equivalent and we write $\e_1\equiv \e_2$ if $\Vert \e_1\Vert=\Vert \e_2 \Vert$. Next we define weighted component-based systems. For this, we introduce the notion of a weighted atomic component.

\begin{defi} A weighted atomic component over $K$ is a pair $wB=(B, wt)$ where  $B=(Q,P,q_0,R)$ is an atomic component and $wt:R\rightarrow K$ is a weight mapping.
\end{defi}

\noindent Since every port in $P$ occurs in at most one transition in $R$, we consider, in the sequel, $wt$ as a mapping $wt:P \rightarrow K$.
If a port $p\in P$ occurs in no transition, then we set $wt(p)=0$.

We call a weighted atomic component $wB$ over $K$ a \emph{weighted component}, whenever we deal with several weighted atomic components and the semiring $K$ is understood. Let $w\B = \{wB(i) \mid   i \in [n] \} $ be a set of weighted components where $wB(i)=(B(i), wt(i))$ with $B(i)=(Q(i),P(i),q_{0}(i), R(i)) $ for every $i \in [n]$. The set of ports and the set of interactions of $w\B$  are the sets $P_{\mathcal{B}}$ and $I_{\mathcal{B}}$ respectively, of the underlying set of components $\B =\{B(i) \mid i \in [n] \}$. Let $a=\{p_{j_1},\ldots,p_{j_m}\}$ be an interaction in $I_{\B}$ such that $p_{j_l}\in P(j_l)$ for every $l\in [m]$. Then, the weighted monomial $\tilde{\phi}_a$ of $a$ is given by the wPIL formula 
\begin{align*}
\tilde{\phi}_a &=( wt(j_1)(p_{j_1})\otimes p_{j_1}) \otimes \ldots \otimes (wt(j_m)(p_{j_m}) \otimes p_{j_m}) \\
& \equiv (wt(j_1)(p_{j_1})\otimes   \ldots \otimes wt(j_m)(p_{j_m}))  \otimes  (p_{j_1}\otimes\ldots     \otimes p_{j_m}) \\
& \equiv (wt(j_1)(p_{j_1})\otimes  \ldots \otimes wt(j_m)(p_{j_m}))  \otimes  (p_{j_1} \wedge \ldots  \wedge  p_{j_m})
\end{align*}
where the first equivalence holds since $K$ is commutative and the second one since $p \otimes p' \equiv p\wedge p'$ for every $p,p' \in P_{\B}$.

\begin{defi}
A weighted component-based system (over $K$) is a pair $(w\B, \tilde{\varphi})$ where $w\B =\{wB(i) \mid i\in [n] \}$ is a set of weighted components and $\e$ is a wEPIL formula over $P_{\B}$ and $K$. 
\end{defi}

We should note that, as in the unweighted case,  the wEPIL formula $\e$ is defined over the set of ports $P_{\B}$ and  is interpreted as a series in $K \left \langle \left \langle I_{\B}^+ \right \rangle \right \rangle $.
 
\

Next we consider three examples of weighted component-based models whose architectures have ordered interactions encoded by wEPIL formulas. We recall from the Subsection \ref{examples_EPIL} the following macro EPIL formula. Let $P=\{p_1, \ldots , p_n\}$ be a set of ports. Then, for  $p_{i_1}, \ldots , p_{i_m} \in P$ with $m <n$ we let 
$$\#(p_{i_1} \wedge \ldots \wedge p_{i_m})::=p_{i_1}\wedge \ldots \wedge p_{i_m} \wedge \bigwedge_{p \in P \setminus \{p_{i_1}, \ldots, p_{i_m}\}}\neg p.$$ 
Let us now assume that we assign a  weight $k_{i_l} \in K$ to $p_{i_l}$ for every $l \in [m]$. We define the subsequent macro wEPIL formula 
 $$\#_w(p_{i_1} \otimes \ldots \otimes p_{i_m})::=(k_{i_1} \otimes p_{i_1}) \otimes \ldots \otimes (k_{i_m} \otimes p_{i_m}) \otimes \bigwedge_{p \in P \setminus \{p_{i_1}, \ldots, p_{i_m}\}}\neg p. $$
Then, we get \begin{align*}
\#_w(p_{i_1} \otimes \ldots \otimes p_{i_m})& \equiv (k_{i_1} \otimes \ldots \otimes k_{i_m})   \otimes (p_{i_1} \otimes \ldots \otimes p_{i_m}) \otimes \bigwedge_{p \in P \setminus \{p_{i_1}, \ldots, p_{i_m}\}}\neg p \\
& \equiv (k_{i_1} \otimes \ldots \otimes k_{i_m})   \otimes \bigg((p_{i_1} \wedge \ldots \wedge p_{i_m}) \wedge \bigwedge_{p \in P \setminus \{p_{i_1}, \ldots, p_{i_m}\}}\neg p \bigg)\\
& = (k_{i_1} \otimes \ldots \otimes k_{i_m})   \otimes \#(p_{i_1} \wedge \ldots \wedge p_{i_m}).
\end{align*}

\noindent Clearly the above macro formula $\#_w(p_{i_1} \otimes \ldots \otimes p_{i_m})$ depends on the values $k_{i_1}, \ldots, k_{i_m}$. Though, in order to simplify our wEPIL formulas, we make no special notation about this. If the macro formula is defined in a weighted component-based system, then the values $k_{i_1}, \ldots, k_{i_m}$ are unique in the whole formula. 

\begin{exa}\label{ex_b_wblackboard}\textbf{(Weighted Blackboard)}
Consider a weighted component-based system  $(w\B, \tilde{\varphi})$ with the Blackboard architecture described in Example \ref{ex_b_blackboard}. We assume the existence of three knowledge source weighted components. Therefore, we have  $w\B=\lbrace wB(1), wB(2),wB(3),wB(4), wB(5)\rbrace$ referring to blackboard,  controller, and the three source
weighted components, respectively. Figure \ref{wb_blackboard} depicts each of the weighted components in the system and a possible execution of the permissible interactions. The weight associated with each port in the system is shown at the outside of the port.

The allowed interactions range over $I_{\B}$, i.e., they are defined as in the corresponding (unweighted) component-based system $\B$ with Blackboard architecture. Then, the wEPIL formula $\e$ for the weighted Blackboard architecture with three source weighted components is 
\begin{multline*}\e=\Bigg(\#_w(p_d\otimes p_r) \odot \bigg(\#_w(p_d\otimes p_{n_1})\varpi \#_w(p_d\otimes p_{n_2})\varpi \#_w(p_d\otimes p_{n_3}) \bigg) \odot \\ \bigg(\e_1\oplus \e_2\oplus \e_3\oplus(\e_1 \varpi \e_2) \oplus (\e_1\varpi \e_3)\oplus (\e_2\varpi \e_3)\oplus (\e_1\varpi\e_2\varpi\e_3)\bigg)^{\oplus}\Bigg)^{\oplus}
\end{multline*}
where 
$$\e_i= \#_w(p_l\otimes p_{t_i}) \odot \#_w(p_e\otimes p_{w_i}\otimes p_a)$$
for $i=1,2,3$. The first wPIL subformula expresses the cost for the connection of blackboard and controller. The wEPIL subformula between the two weighted concatenation operators represents the cost of the connection of the three  sources to blackboard in order to be informed for existing data. The last part of $\e$ captures the cost of applying the connection of some of the three sources with controller and blackboard for the triggering and writing process. The weighted iteration operators describe the cost of executing recursive interactions in $\e$. Let $w_1,w_2 \in I^{+}_{\B}$ encode two distinct sequences of interactions permitted in the given weighted Blackboard architecture. The values $\left \Vert \e \right \Vert(w_1)$ and $\left \Vert \e \right \Vert(w_2)$ represent the cost for executing these interactions with the order encoded by $w_1$ and $w_2$, respectively. Then, $\left \Vert \e \right \Vert(w_1)+\left \Vert \e \right \Vert(w_2)$  is the `total' cost for implementing $w_1$ and $w_2$. For example, if we consider
the max-plus semiring, then the value $\max \lbrace \left \Vert \e \right \Vert(w_1),\left \Vert \e \right \Vert(w_2)\rbrace$ gives information for
the communication with the maximum cost. Such information would be important in systems with restricted resources such as battery capacity for instance, in order for the systems to opt for implementing the least `expensive' sequence of interactions.

\definecolor{harlequin}{rgb}{0.25, 1.0, 0.0}
\definecolor{ao}{rgb}{0.0, 0.0, 1.0}

\definecolor{darkorange}{rgb}{1.0, 0.55, 0.0}
\definecolor{harlequin}{rgb}{0.25, 1.0, 0.0}
\definecolor{ao}{rgb}{0.0, 0.0, 1.0}

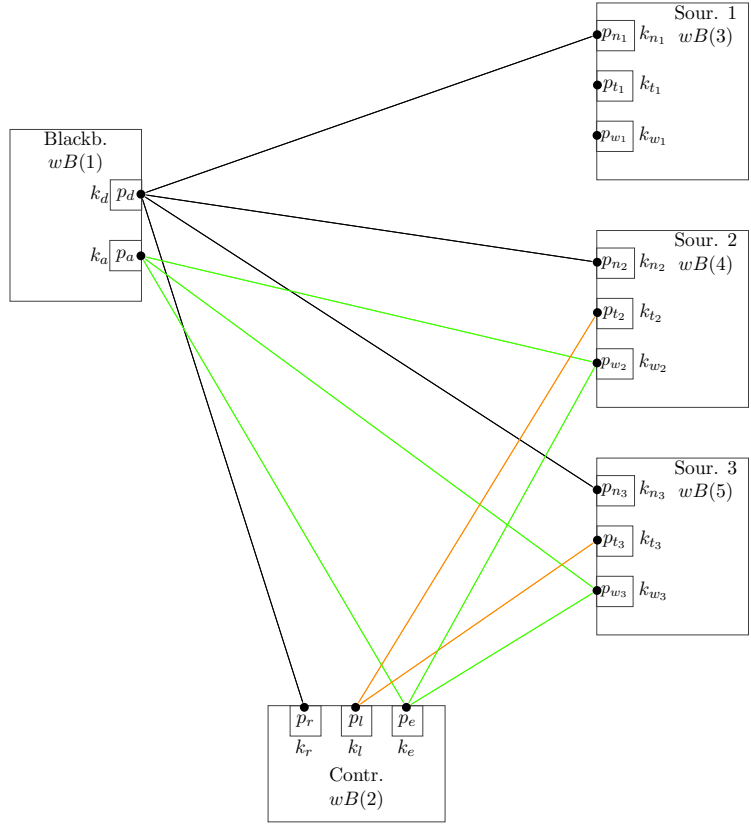
\begin{figure}[h]
\centering
\resizebox{0.65\linewidth}{!}{
\begin{tikzpicture}[>=stealth',shorten >=1pt,auto,node distance=1cm,baseline=(current bounding box.north)]
\tikzstyle{component}=[rectangle,ultra thin,draw=black!75,align=center,inner sep=9pt,minimum size=1.5cm,minimum height=3.4
cm,minimum width=2.6cm]
\tikzstyle{port}=[rectangle,ultra thin,draw=black!75,minimum size=6mm]
\tikzstyle{bubble} = [fill,shape=circle,minimum size=5pt,inner sep=0pt]
\tikzstyle{type} = [draw=none,fill=none] 

\node [component,align=center] (a1)  {};
\node [port] (a2) [above right= -1.6cm and -0.62cm of a1]  {$p_d$};
\node[bubble] (a3) [above right=-1.35cm and -0.08cm of a1]   {};

\node[]       (k1) [left= 0.4cm of a3]{$k_d$};

\node [port] (a4) [above right= -2.8cm and -0.63cm of a1]  {$p_a$};
\node[bubble] (a5) [above right=-2.57cm and -0.08cm of a1]   {};

\node[]       (k2) [left= 0.4cm of a5]{$k_a$};

\node[type] (a6) [above=-0.45cm of a1]{Blackb.};
\node            [below=-0.1cm of a6]{$wB(1)$};

\tikzstyle{control}=[rectangle,ultra thin,draw=black!75,align=center,inner sep=9pt,minimum size=1.5cm,minimum height=2.3cm,minimum width=3.5cm]

\node [control,align=center] (b1)  [below right=8cm and 2.5cm of a1]{};
\node [port] (b2) [above left=-0.605cm  and -1.05 of b1]  {$p_r$};
\node[bubble] (b3) [above left=-0.1cm and -0.35cm of b2]   {};

\node[]       (k3) [below= 0.4cm of b3]{$k_r$};

\node [port] (b4) [right=0.4 of b2]  {$p_l$};
\node[bubble] (b5) [above left=-0.1cm and -0.35cm of b4]   {};

\node[]       (k4) [below= 0.4cm of b5]{$k_l$};

\node [port] (b6) [right=0.4 of b4]  {$p_e$};
\node[bubble] (b7) [above left=-0.1cm and -0.35cm of b6]   {};

\node[]       (k5) [below= 0.4cm of b7]{$k_e$};

\node[type] (b8) [below=-1.2cm of b1]{Contr.};
\node[]      [below=-0.1cm of b8]{$wB(2)$};

\tikzstyle{source}=[rectangle,ultra thin,draw=black!75,align=center,inner sep=9pt,minimum size=1.5cm,minimum height=3.5cm,minimum width=3cm]

\tikzstyle{port}=[rectangle,ultra thin,draw=black!75,minimum size=6mm,inner sep=3pt]

\node [source,align=center] (c1) [above right=-1cm and 9cm of a1] {};
\node [port] (c2) [above left= -0.95cm and -0.75cm of c1]  {$p_{n_1}$};
\node[bubble] (c3) [above left=-0.35cm and -0.08cm of c2]   {};

\node[] (k6) [right= 0.6cm of c3]{$k_{n_1}$};

\tikzstyle{port}=[rectangle,ultra thin,draw=black!75,minimum size=6mm,inner sep=3.5pt]
\node [port] (c4) [above left= -1.6cm and -0.72cm of c2]  {$p_{t_1}$};
\node[bubble] (c5) [above left=-0.35cm and -0.08cm of c4]   {};

\node[] (k7) [right= 0.6cm of c5]{$k_{t_1}$};

\tikzstyle{port}=[rectangle,ultra thin,draw=black!75,minimum size=6mm,inner sep=2.7pt]
\node [port] (c6) [above left= -1.6cm and -0.72cm of c4]  {{\small $p_{w_1}$}};
\node[bubble] (c7) [above left=-0.35cm and -0.08cm of c6]   {};

\node[] (k8) [right= 0.6cm of c7]{$k_{w_1}$};

\node[type,align=center] (c8) [above right=-0.45cm and -1.6cm of c1]{Sour. $1$};
\node[] [below=-0.1cm of c8]{$wB(3)$};

\tikzstyle{port}=[rectangle,ultra thin,draw=black!75,minimum size=6mm,inner sep=3pt]
\node [source,align=center] (d1) [below= 1cm of c1] {};
\node [port] (d2) [above left= -0.95cm and -0.75cm of d1]  {$p_{n_2}$};
\node[bubble] (d3) [above left=-0.35cm and -0.08cm of d2]   {};

\node[] (k9) [right= 0.6cm of d3]{$k_{n_2}$};

\tikzstyle{port}=[rectangle,ultra thin,draw=black!75,minimum size=6mm,inner sep=3.5pt]
\node [port] (d4) [above left= -1.6cm and -0.72cm of d2]  {$p_{t_2}$};
\node[bubble] (d5) [above left=-0.35cm and -0.08cm of d4]   {};

\node[] (k10) [right= 0.6cm of d5]{$k_{t_2}$};

\tikzstyle{port}=[rectangle,ultra thin,draw=black!75,minimum size=6mm,inner sep=2.7pt]
\node [port] (d6) [above left= -1.6cm and -0.72cm of d4]  {{\small $p_{w_2}$}};
\node[bubble] (d7) [above left=-0.35cm and -0.08cm of d6]   {};

\node[] (k11) [right= 0.6cm of d7]{$k_{w_2}$};

\node[type] (d8) [above right=-0.45cm and -1.6cm of d1]{Sour. $2$};
\node[] [below=-0.1cm of d8]{$wB(4)$};

\tikzstyle{port}=[rectangle,ultra thin,draw=black!75,minimum size=6mm,inner sep=3pt]
\node [source,align=center] (e1) [below= 1cm of d1] {};
\node [port] (e2) [above left= -0.95cm and -0.75cm of e1]  {$p_{n_3}$};
\node[bubble] (e3) [above left=-0.35cm and -0.08cm of e2]   {};

\node[] (k12) [right= 0.6cm of e3]{$k_{n_3}$};

\tikzstyle{port}=[rectangle,ultra thin,draw=black!75,minimum size=6mm,inner sep=3.5pt]
\node [port] (e4) [above left= -1.6cm and -0.72cm of e2]  {$p_{t_3}$};
\node[bubble] (e5) [above left=-0.35cm and -0.08cm of e4]   {};

\node[] (k13) [right= 0.6cm of e5]{$k_{t_3}$};

\tikzstyle{port}=[rectangle,ultra thin,draw=black!75,minimum size=6mm,inner sep=2.7pt]
\node [port] (e6) [above left= -1.6cm and -0.72cm of e4]  {{\small $p_{w_3}$}};
\node[bubble] (e7) [above left=-0.35cm and -0.08cm of e6]   {};

\node[] (k14) [right= 0.6cm of e7]{$k_{w_3}$};

\node[type] (e8) [above right=-0.45cm and -1.6cm of e1]{Sour. $3$};
\node[] [below=-0.1cm of e8]{$wB(5)$};
 
 
\path[-]          (a3)  edge                  node {} (c3);
 \path[-]          (a3)  edge                  node {} (d3);
 \path[-]          (a3)  edge                  node {} (e3);

 \path[-]          (a3)  edge                  node {} (b3);
  \path[-]          (b3)  edge                  node {} (a3);
 
 
 \path[-]          (c3)  edge                  node {} (a3);
 \path[-]          (d3)  edge                  node {} (a3);
 \path[-]          (e3)  edge                  node {} (a3);


 \path[-]          (a5)  edge  [harlequin]                node {} (d7);
 \path[-]          (a5)  edge  [harlequin]                node {} (e7);

 \path[-]          (a5)  edge  [harlequin]                node {} (b7);
 

 \path[-]          (d7)  edge   [harlequin]               node {} (a5);
 \path[-]          (e7)  edge   [harlequin]               node {} (a5);

 \path[-]          (b7)  edge   [harlequin]              node {} (a5);
 
 
  \path[-]          (b5)  edge  [darkorange]                node {} (d5);
   \path[-]          (b5)  edge [darkorange]                node {} (e5);

  \path[-]          (d5)  edge  [darkorange]               node {} (b5);
  \path[-]          (e5)  edge  [darkorange]               node {} (b5);

  \path[-]          (b7)  edge  [harlequin]               node {} (d7);
 \path[-]          (d7)  edge  [harlequin]               node {} (b7);
 
    \path[-]          (b7)  edge  [harlequin]               node {} (e7);
 \path[-]          (e7)  edge  [harlequin]               node {} (b7);
   
\end{tikzpicture}}
\caption{A possible execution of the interactions in a weighted Blackboard architecture.}
\label{wb_blackboard}
\end{figure}
\end{exa}

\begin{exa} \textbf{(Weighted Request/Response)}
\label{w-b-r-r}
Let $(w\B,\widetilde{\varphi})$ be a weighted component-based system with the Request/Response architecture presented in Example \ref{b-r-r}. Our weighted system, shown in Figure \ref{w_b_re_re}, consists of seven weighted components, and specifically, the service registry, two services  with their associated coordinators, and two clients. Therefore, we have that $w\B=\lbrace wB(1),\ldots, wB(7)\rbrace$ referring to each of the aforementioned weighted components, respectively. The allowed interactions range over $I_{\B}$, and the wEPIL formula $\e$ describing the weighted Request/Response architecture is
\begin{multline*}
\e=\big(\#_w(p_e\otimes p_{r_1})\varpi \#_w(p_e\otimes p_{r_2})\big) \odot \big(\x_1 \varpi \x_2\big) \odot \\ \Bigg( \bigg(\e_{11}  \oplus \e_{21} \oplus (\e_{11} \odot \e_{21}) \oplus (\e_{21} \odot \e_{11})\bigg)^{\oplus} \bigoplus  \bigg(\e_{12}  \oplus \e_{22} \oplus (\e_{12} \odot \e_{22}) \oplus (\e_{22} \odot \e_{12}) \bigg)^{\oplus} \bigoplus \\ \bigg(\big(\e_{11}  \oplus \e_{21} \oplus (\e_{11} \odot \e_{21}) \oplus (\e_{21} \odot \e_{11})\big)^{\oplus} \varpi 
 \\  \big(\e_{12}  \oplus \e_{22} \oplus (\e_{12} \odot \e_{22}) \oplus (\e_{22} \odot \e_{12}) \big)^{\oplus}\bigg)^{\oplus}\Bigg)^{\oplus}
\end{multline*} 
where
\begin{itemize}
 \item[-]$\x_1=\#_w(p_{l_1}\otimes p_u) \odot \#_w(p_{o_1}\otimes p_t)$
 \item[-]$\x_2=\#_w(p_{l_2}\otimes p_u) \odot \#_w(p_{o_2}\otimes p_t)$
\end{itemize}
capture the cost of the interactions of the two clients with the service registry and
\begin{itemize} 
\item[-]$\e_{11}=\#_w(p_{n_1}\otimes p_{m_1}) \odot \#_w(p_{q_1}\otimes p_{a_1}\otimes p_{g_1})\odot \#_w(p_{c_1}\otimes p_{d_1}\otimes p_{s_1})$
\item[-] $\e_{12}=\#_w(p_{n_1}\otimes p_{m_2}) \odot \#_w(p_{q_1}\otimes p_{a_2}\otimes p_{g_2})\odot \#_w(p_{c_1}\otimes p_{d_2}\otimes p_{s_2})$
\item[-]$\e_{21}=\#_w(p_{n_2}\otimes p_{m_1}) \odot \#_w(p_{q_2}\otimes p_{a_1}\otimes p_{g_1})\odot \#_w(p_{c_2}\otimes p_{d_1}\otimes p_{s_1})$
\item[-] $\e_{22}=\#_w(p_{n_2}\otimes p_{m_2}) \odot \#_w(p_{q_2}\otimes p_{a_2}\otimes p_{g_2}) \odot \#_w(p_{c_2}\otimes p_{d_2}\otimes p_{s_2})$
\end{itemize}
\noindent describe the connection's cost of each of the two clients with the two services through their coordinators.

\noindent The two wEPIL subformulas at the left of the first two weighted concatenation operators encode the cost for the connections of the two services and the two clients with registry, respectively. Then, each of the three wEPIL subformulas connected with $\bigoplus$ represent the cost for the connection of either one of the two clients or both of them (one at each time) with the first service only, the second service only, or both of the services, respectively. Also, the weighted iteration operators in $\e$ express the cost of repeating the corresponding permissible interactions in the architecture.

Let $w_1,w_2 \in I^{+}_{\B}$ encode two distinct sequences of interactions in the presented weighted Request/Response architecture. Then, the value $\min \lbrace \left \Vert \e \right \Vert(w_1),\left \Vert \e \right \Vert(w_2)\rbrace$ expresses in min-plus semiring for instance,
the communication with the minimum cost. Furthermore, if we interpret the cost as the energy consumption required for implementing the architecture in a network, then we would be able to derive which pattern of interactions, $w_1$ or $w_2$, requires the minimum energy.

\definecolor{darkorange}{rgb}{1.0, 0.55, 0.0}
\definecolor{harlequin}{rgb}{0.25, 1.0, 0.0}
\definecolor{ao}{rgb}{0.0, 0.0, 1.0}

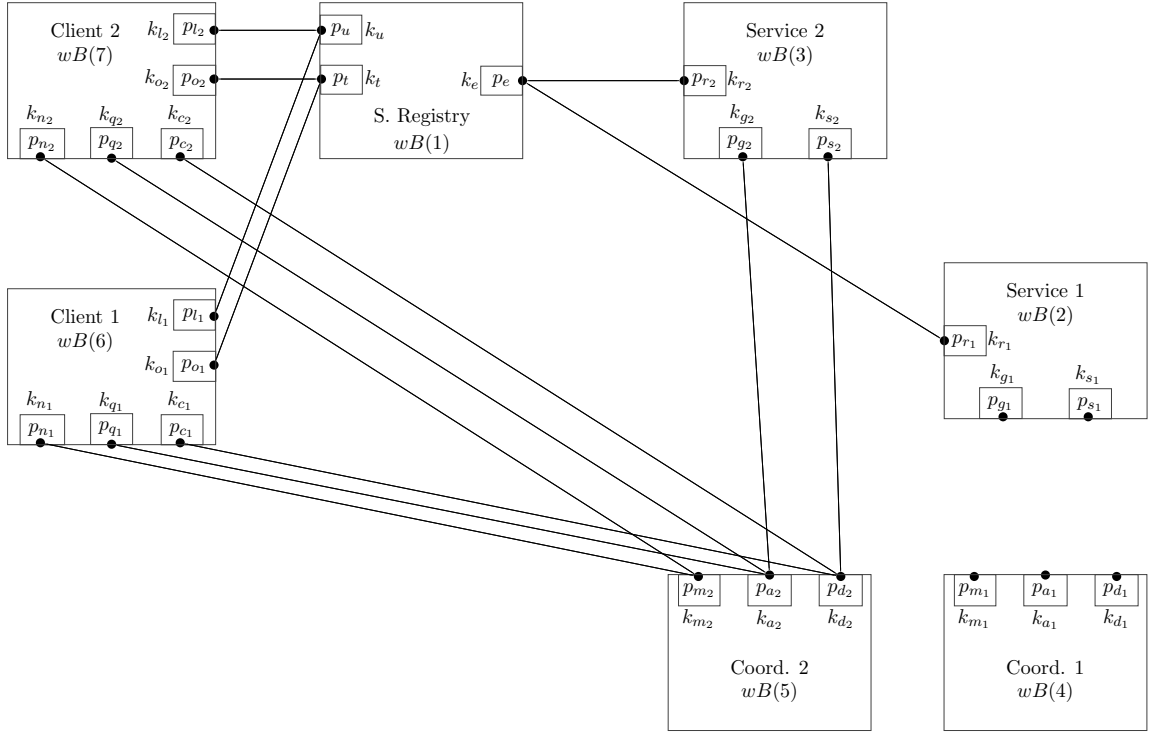
\begin{figure}[h]
\centering
\resizebox{1.0\linewidth}{!}{
\begin{tikzpicture}[>=stealth',shorten >=1pt,auto,node distance=3cm,baseline=(current bounding box.north)]
\tikzstyle{component}=[rectangle,ultra thin,draw=black!75,align=center,inner sep=9pt,minimum size=2.5cm,minimum height=3cm, minimum width=3.9cm]
\tikzstyle{port}=[rectangle,ultra thin,draw=black!75,minimum size=6mm]
\tikzstyle{bubble} = [fill,shape=circle,minimum size=5pt,inner sep=0pt]
\tikzstyle{type} = [draw=none,fill=none]

\node [component] (a1) {};

\tikzstyle{port}=[rectangle,ultra thin,draw=black!75,minimum size=6mm,inner sep=2.5pt]
\node[bubble] (a2)  [above left=-0.105cm and -0.65cm of a1]   {};   
\node [port]  (a3)  [above left=-0.605cm and -1.0cm of a1]  {$p_{m_1}$};  

\node []  (k1)  [below=0.4cm of a2]  {$k_{m_1}$}; 

\tikzstyle{port}=[rectangle,ultra thin,draw=black!75,minimum size=6mm,inner sep=4.5pt]
\node[bubble] (a4)  [above=-0.105cm of a1]   {};   
\node [port]  (a5)  [above=-0.605cm and -2.0cm of a1]  {$p_{a_1}$};  

\node []  (k2)  [below=0.45cm of a4]  {$k_{a_1}$}; 

\node[bubble] (a6)  [above right=-0.105cm and -0.65cm of a1]   {};   
\node [port]  (a7)  [above right=-0.605cm and -1.0cm of a1]  {$p_{d_1}$};  

\node []  (k3)  [below=0.4cm of a6]  {$k_{d_1}$}; 

\node[type]   (a8)      [below=-1.5cm of a1]{Coord. $1$};
\node                   [below=-0.1cm of a8]{$wB(4)$};

\node [component] (h1) [left=1.4 of a1]{};

\tikzstyle{port}=[rectangle,ultra thin,draw=black!75,minimum size=6mm,inner sep=2.5pt]
\node[bubble] (h2)  [above left=-0.105cm and -0.65cm of h1]   {};   
\node [port]  (h3)  [above left=-0.605cm and -1.0cm of h1]  {$p_{m_2}$};  

\node []  (k4)  [below=0.4cm of h2]  {$k_{m_2}$}; 

\tikzstyle{port}=[rectangle,ultra thin,draw=black!75,minimum size=6mm,inner sep=4.5pt]
\node[bubble] (h4)  [above=-0.105cm of h1]   {};   
\node [port]  (h5)  [above=-0.605cm and -2.0cm of h1]  {$p_{a_2}$};  

\node []  (k4)  [below=0.45cm of h4]  {$k_{a_2}$}; 

\node[bubble] (h6)  [above right=-0.105cm and -0.65cm of h1]   {};   
\node [port]  (h7)  [above right=-0.605cm and -1.0cm of h1]  {$p_{d_2}$};  

\node []  (k5)  [below=0.4cm of h6]  {$k_{d_2}$}; 

\node[type]   (h8)      [below=-1.5cm of h1]{Coord. $2$};
\node                   [below=-0.1cm of h8]{$wB(5)$};

\tikzstyle{client}=[rectangle,ultra thin,draw=black!75,align=center,inner sep=9pt,minimum size=2.5cm,minimum height=3.0cm, minimum width=4cm]

\tikzstyle{port}=[rectangle,ultra thin,draw=black!75,minimum size=6mm,inner sep=4.5pt,minimum height=0.4pt]
\node [client] (b1) [above left=8.0cm and 14.0cm of a1]{};

\node[bubble]     (b2) [below left=-0.105cm and -0.7cm of b1]   {};   
\node [port]      (b3) [below left=-0.58cm and -1.1cm of b1]  {$p_{n_2}$}; 
 
\node []  (k6)  [above=0.4cm of b2]  {$k_{n_2}$};

\tikzstyle{port}=[rectangle,ultra thin,draw=black!75,minimum size=6mm,inner sep=4.5pt, minimum height=2pt, minimum width= 0.4cm]
\node[bubble]     (b4) [below=-0.105cm of b1]   {};  
\node [port]      (b5) [below=-0.60cm and -2.0cm of b1]  {$p_{q_2}$}; 

\node []  (k7)  [above=0.4cm of b4]  {$k_{q_2}$};

\tikzstyle{port}=[rectangle,ultra thin,draw=black!75,minimum size=6mm,inner sep=4.5pt, minimum height=1pt, minimum width= 0.4cm]
\node[bubble]     (b6) [below right=-0.105cm and -0.75cm of b1]   {};   
\node [port]      (b7) [below right=-0.58cm and -1.05cm of b1]  {$p_{c_2}$};  
 
 \node []  (k8)  [above=0.4cm of b6]  {$k_{c_2}$};

\tikzstyle{port}=[rectangle,ultra thin,draw=black!75,minimum size=6mm,inner sep=4.5pt, minimum height=2.4pt, minimum width=0.8cm]
\node[bubble] (b8)  [above right=-0.6cm and -0.1cm of b1]   {};   
\node [port]  (b9)  [above right=-0.8cm and -.81cm of b1]  {$p_{l_2}$};  

\node []  (k9)  [left=0.6cm of b8]  {$k_{l_2}$};

\node[bubble] (b10)  [below right=-1.6cm and -0.1cm of b1]   {};   
\node [port]  (b11)  [below right=-1.8 cm and -.82cm of b1]  {$p_{o_2}$};

\node []  (k10)  [left=0.6cm of b10]  {$k_{o_2}$};

\node[type]    (b12)        [above left=-0.8cm and -2.3cm of b1]{Client $2$};
\node                  [below=-0.1cm of b12]{$wB(7)$};

\node [component] (g1) [right=2.0 of b1]{};

\tikzstyle{port}=[rectangle,ultra thin,draw=black!75,minimum size=6mm,inner sep=4.5pt, minimum height=2pt, minimum width=0.8cm]
\node[bubble] (g2)  [right=-0.1cm of g1]   {};   
\node [port]  (g3)  [right=-.81cm of g1]  {$p_{e}$};  

\node []  (k11)  [left=0.6cm of g2]  {$k_{e}$};

\node[bubble] (g4)  [above left=-0.6 and -0.1cm of g1]   {};   
\node [port]  (g5)  [above left=-0.8cm and -.81cm of g1]  {$p_{u}$};  

\node []  (k12)  [right=0.6cm of g4]  {$k_{u}$};

\node[bubble] (g6)  [below left=-1.6cm and -0.1cm of g1]   {};   
\node [port]  (g7)  [below left=-1.8 cm and -.82cm of g1]  {$p_{t}$};  

\node []  (k13)  [right=0.6cm of g6]  {$k_{t}$};

\node[type]   (g8)      [above right=-2.5cm and -3cm of g1]{S. Registry};
\node                   [below=-0.1cm of g8]{$wB(1)$};

\tikzstyle{port}=[rectangle,ultra thin,draw=black!75,minimum size=6mm,inner sep=4.5pt, minimum height=1pt, minimum width= 0.4cm]
\node [client] (c1) [below=2.5cm of b1]{};

\node[bubble]     (c2) [below left=-0.105cm and -0.7cm of c1]   {};   
\node [port]      (c3) [below left=-0.58cm and -1.1cm of c1]  {$p_{n_1}$}; 
 
\node []  (k14)  [above=0.4cm of c2]  {$k_{n_1}$}; 

\tikzstyle{port}=[rectangle,ultra thin,draw=black!75,minimum size=6mm,inner sep=4.5pt, minimum height=2pt, minimum width= 0.4cm]
\node[bubble]     (c4) [below=-0.105cm of c1]   {};  
\node [port]      (c5) [below=-0.60cm and -2.0cm of c1]  {$p_{q_1}$}; 

\node []  (k15)  [above=0.4cm of c4]  {$k_{q_1}$}; 

\tikzstyle{port}=[rectangle,ultra thin,draw=black!75,minimum size=6mm,inner sep=4.5pt, minimum height=1pt, minimum width= 0.4cm]
\node[bubble]     (c6) [below right=-0.105cm and -0.75cm of c1]   {};   
\node [port]      (c7) [below right=-0.58cm and -1.05cm of c1]  {$p_{c_1}$};  

\node []  (k16)  [above=0.4cm of c6]  {$k_{c_1}$}; 

\tikzstyle{port}=[rectangle,ultra thin,draw=black!75,minimum size=6mm,inner sep=4.5pt, minimum height=2.4pt, minimum width=0.8cm]
\node[bubble] (c8)  [above right=-0.6 and -0.1cm of c1]   {};   
\node [port]  (c9)  [above right=-0.8cm and -.81cm of c1]  {$p_{l_1}$};

\node []  (k17)  [left=0.6cm of c8]  {$k_{l_1}$};

\node[bubble] (c10)  [below right=-1.6cm and -0.1cm of c1]   {};   
\node [port]  (c11)  [below right=-1.8 cm and -.82cm of c1]  {$p_{o_1}$};
 
\node []  (k18)  [left=0.6cm of c10]  {$k_{o_1}$};

\node[type]    (c12)        [above left=-0.8cm and -2.3cm of c1]{Client $1$};
\node                       [below=-0.1cm of c12]{$wB(6)$};

\tikzstyle{port}=[rectangle,ultra thin,draw=black!75,minimum size=6mm,inner sep=4.5pt, minimum height=1pt, minimum width= 0.4cm]
\node [component] (d1) [right=9cm of b1]{};
\node[bubble] (d2) [below left=-0.105cm and -1.2cm of d1]   {};   
\node [port] (d3) [below left=-0.60cm and -1.5cm of d1]  {$p_{g_2}$};  

\node []  (k19)  [above=0.4cm of d2]  {$k_{g_2}$};

\node[bubble] (d4) [below right=-0.105cm and -1.2cm of d1]   {};   
\node [port] (d5) [below right=-0.58cm and -1.5cm of d1]  {$p_{s_2}$};  

\node []  (k20)  [above=0.4cm of d4]  {$k_{s_2}$};

\node[bubble] (d6) [left=-0.1cm of d1]   {};   
\node [port] (d7) [left=-.81cm of d1]  {$p_{r_2}$}; 

\node []  (k20)  [right=0.6cm of d6]  {$k_{r_2}$};

\node[type] (d8) [above=-0.8cm of d1]{Service $2$};
\node            [below=-0.1cm of d8]{$wB(3)$};

\node [component] (e1) [below right=2cm and 1.1cm of d1]{};
\node[bubble] (e2) [below left=-0.105cm and -1.2cm of e1]   {};   
\node [port] (e3) [below left=-0.60cm and -1.5cm of e1]  {$p_{g_1}$};  

\node []  (k21)  [above=0.4cm of e2]  {$k_{g_1}$};

\node[bubble] (e4) [below right=-0.105cm and -1.2cm of e1]   {};   
\node [port] (e5) [below right=-0.58cm and -1.5cm of e1]  {$p_{s_1}$}; 

\node []  (k22)  [above=0.4cm of e4]  {$k_{s_1}$};

\node[bubble] (e6) [left=-0.1cm of e1]   {};   
\node [port] (e7) [left=-.81cm of e1]  {$p_{r_1}$};  

\node []  (k23)  [right=0.6cm of e6]  {$k_{r_1}$};

\node[type] (e8) [above=-0.8cm of e1]{Service $1$};
\node        [below=-0.1cm of e8]{$wB(2)$};

\path[-]          (g2)  edge                  node {} (d6);
\path[-]          (d6)  edge                  node {} (g2);

\path[-]          (g2)  edge                  node {} (e6);
\path[-]          (e6)  edge                 node {} (g2);
\path[-]          (g4)  edge [darkorange]                 node {} (c8);
\path[-]          (c8)  edge [darkorange]                 node {} (g4);

\path[-]          (g6)  edge   [darkorange]                node {} (c10);
\path[-]          (c10) edge    [darkorange]               node {} (g6);

\path[-]          (c2)  edge                  node {} (h2);
\path[-]          (c4)  edge                  node {} (h4);
\path[-]          (c6)  edge                  node {} (h6);

\path[-]          (h2)  edge  [harlequin]                 node {} (c2);
\path[-]          (h4)  edge  [harlequin]                node {} (c4);
\path[-]          (h6)  edge  [harlequin]                node {} (c6);


\path[-]          (g4)  edge  [darkorange]                node {} (b8);
\path[-]          (b8)  edge  [darkorange]                node {} (g4);

\path[-]          (g6)  edge   [darkorange]               node {} (b10);
\path[-]          (b10)  edge  [darkorange]                  node {} (g6);

\path[-]          (b2)  edge  [harlequin]                 node {} (h2);
\path[-]          (h2)  edge  [harlequin]                node {} (b2);

\path[-]          (b4)  edge  [harlequin]                 node {} (h4);
\path[-]          (h4)  edge  [harlequin]                node {} (b4);

\path[-]          (b6)  edge  [harlequin]                node {} (h6);
\path[-]          (h6)  edge  [harlequin]                 node {} (b6);

\path[-]          (d2)  edge [harlequin]                 node {} (h4);
\path[-]          (h4)  edge [harlequin]                 node {} (d2);

\path[-]          (d4)  edge [harlequin]                 node {} (h6);
\path[-]          (h6)  edge [harlequin]                 node {} (d4);

\end{tikzpicture}}
\caption{A possible execution of the interactions in a weighted Request/Response architecture.}
\label{w_b_re_re}
\end{figure}
\end{exa}

\begin{exa} \textbf{(Weighted Publish/Subscribe)}
\label{w-b-pu-su} Let $(w\B,\widetilde{\varphi})$ be a weighted component-based system with the Publish/Subscribe architecture described in Example \ref{b-pu-su}. Our weighted system, shown in Figure \ref{w-b-p-s}, is comprised of two publisher, two topic, and three subscriber weighted components, hence  $w\B=\lbrace wB(1),\ldots, wB(7)\rbrace$ referring to these components, respectively. The permissible interactions range over $I_{\B}$, and the wEPIL formula $\e$ for the weighted Publish/Subscribe architecture is $\e=(\e_1\oplus \e_2\oplus (\e_1\varpi \e_2))^{\oplus}$ with
\begin{multline*}
\e_1=\bigg(\big( \x_1 \odot \e_{11}\big)\oplus \big( \x_1 \odot \e_{12}\big)\oplus \big(  \x_1 \odot \e_{13}\big)\oplus \\
\big(  \x_1 \odot (\e_{11}\varpi \e_{12})\big)\oplus \big( \x_1 \odot (\e_{11}\varpi \e_{13})\big)\oplus \\ \big( \x_1 \odot (\e_{12}\varpi \e_{13})\big)\oplus \big( \x_1 \odot (\e_{11}\varpi\e_{12}\varpi \e_{13})\big)\bigg)
\end{multline*}
and 
\begin{multline*}
\e_2=\bigg(\big(  \x_2 \odot \e_{21}\big)\oplus \big( \x_2 \odot \e_{22}\big)\oplus \big(  \x_2 \odot \e_{23}\big)\oplus \\
\big(  \x_2 \odot (\e_{21}\varpi \e_{22})\big)\oplus \big( \x_2 \odot (\e_{21}\varpi \e_{23})\big)\oplus \\ \big( \x_2 \odot (\e_{22}\varpi \e_{23})\big)\oplus \big(  \x_2 \odot (\e_{21}\varpi \e_{22}\varpi \e_{23})\big)\bigg)
\end{multline*}
\noindent where  the following auxiliary subformulas: 
\begin{itemize}
    \item[-] $\x_1=\x_{11}\oplus \x_{12}\oplus (\x_{11}\varpi \x_{12})$
    \item[-] $\x_2=\x_{21}\oplus \x_{22}\oplus (\x_{21}\varpi \x_{22})$
\end{itemize}
\noindent represent the cost for the connection of each of the two topics with the first publisher $wB(1)$, or the second publisher $wB(2)$ or with both of them, and
\begin{itemize}
\item[-] $\x_{11}=\#_w(p_{n_1}\otimes p_{a_1}) \odot \#_w(p_{r_1}\otimes p_{t_1})$
\item[-] $\x_{12}=\#_w(p_{n_1}\otimes p_{a_2}) \odot \#_w(p_{r_1}\otimes p_{t_2})$
\item[-] $\x_{21}=\#_w(p_{n_2}\otimes p_{a_1}) \odot \#_w(p_{r_2}\otimes p_{t_1})$
\item[-] $\x_{22}=\#_w(p_{n_2}\otimes p_{a_2}) \odot \#_w(p_{r_2}\otimes p_{t_2})$
\end{itemize}
\noindent describe the cost of the interactions of the two topics with each of the two publishers, and 
\begin{itemize}
\item[-] $\e_{11}= \#_w(p_{c_1}\otimes p_{e_1}) \odot \#_w(p_{s_1}\otimes p_{g_1}) \odot \#_w(p_{f_1}\otimes p_{d_1})$
\item[-] $\e_{12}= \#_w(p_{c_1}\otimes p_{e_2}) \odot \#_w(p_{s_1}\otimes p_{g_2}) \odot \#_w(p_{f_1}\otimes p_{d_2})$
\item[-] $\e_{13}= \#_w(p_{c_1}\otimes p_{e_3})\odot \#_w(p_{s_1}\otimes p_{g_3}) \odot \#_w(p_{f_1}\otimes p_{d_3})$
\item[-] $\e_{21}= \#_w(p_{c_2}\otimes p_{e_1}) \odot \#_w(p_{s_2}\otimes p_{g_1}) \odot \#_w(p_{f_2}\otimes p_{d_1})$
\item[-] $\e_{22}= \#_w(p_{c_2}\otimes p_{e_2})\odot \#_w(p_{s_2}\otimes p_{g_2}) \odot \#_w(p_{f_2}\otimes p_{d_2})$
\item[-] $\e_{23}= \#_w(p_{c_2}\otimes p_{e_3})\odot \#_w(p_{s_2}\otimes p_{g_3}) \odot \#_w(p_{f_2}\otimes p_{d_3})$
\end{itemize}
\noindent describe the cost of the connections of the two topics with each of the three subscribers.

\noindent Each of $\e_1$ and $\e_2$ capture the cost for implementing the  interactions of topics $wB(3)$ and $wB(4)$, respectively, with some of the publisher
and subscriber components.  Then, $\e_1\oplus \e_2\oplus (\e_1\varpi \e_2)$ expresses the cost for the participation of some of the topics
in the architecture, and in turn the use of the weight iteration operator in $\e$ allows to encode the overall cost of applying these interactions with recursion. Consider for instance the words $w_1,w_2 \in I^{+}_{\B}$ encoding two interactions in the given weighted Publish/Subscribe architecture. Then, in
the Viterbi semiring the value $\max \lbrace \left \Vert \e \right \Vert(w_1),\left \Vert \e \right \Vert(w_2)\rbrace$ shows the sequence of interactions executed with the maximum probability. 

\definecolor{darkorange}{rgb}{1.0, 0.55, 0.0}
\definecolor{harlequin}{rgb}{0.25, 1.0, 0.0}
\definecolor{ao}{rgb}{0.0, 0.0, 1.0}

\begin{figure}[h]
\centering
\resizebox{0.7\linewidth}{!}{
\begin{tikzpicture}[>=stealth',shorten >=1pt,auto,node distance=1cm,baseline=(current bounding box.north)]
\tikzstyle{component}=[rectangle,ultra thin,draw=black!75,align=center,inner sep=9pt,minimum size=1.5cm,minimum height=3.9cm,minimum width=3.3cm]
\tikzstyle{port}=[rectangle,ultra thin,draw=black!75,minimum size=7mm]
\tikzstyle{bubble} = [fill,shape=circle,minimum size=5pt,inner sep=0pt]
\tikzstyle{type} = [draw=none,fill=none]

\node [component,align=center] (a1)  {};
\node [port] (a2) [above right=-1.66 and -0.78cm of a1]  {$p_{a_1}$};
\node[bubble] (a3) [below right=-0.40 and -0.10cm of a2]   {};

\node []  (k1)  [left=0.5cm of a3]  {$k_{a_1}$}; 

\node [port] (a4) [below right=0.86cm and -0.73cm of a2]  {$p_{t_1}$};
\node[bubble] (a5) [below right=-0.37 and -0.10cm of a4]   {};

\node []  (k2)  [left=0.5cm of a5]  {$k_{t_1}$};

\node[type] (a6) [above=-0.45cm of a1]{Publ. $1$};
\node[type]  [below=-0.05cm of a6]{$wB(1)$};

\node [component] (b1) [right=3cm of a1] {};
\node [port] (b2) [below right=-1.07cm and -0.77cm of b1]  {$p_{f_1}$}; 
\node[bubble] (b3) [below right=-0.40cm and -0.10cm of b2]   {};   

\node []  (k3)  [left=0.5cm of b3]  {$k_{f_1}$}; 

\node [port] (b4) [above right=0.25cm and -0.75cm of b2]  {$p_{s_1}$}; 
\node[bubble] (b5) [above right=-0.40cm and -0.10cm of b4]   {};   

\node []  (k4)  [left=0.5cm of b5]  {$k_{s_1}$}; 

\node [port] (b6) [above =0.3cm of b4]  {$p_{c_1}$}; 
\node[bubble] (b7) [above right=-0.40cm and -0.10cm of b6]   {};   

\node []  (k5)  [left=0.5cm of b7]  {$k_{c_1}$}; 

\node [port] (b8) [above left= -1.7cm and -0.79cm of b1]  {$p_{n_1}$};
\node[bubble] (b9) [above left=-0.4cm and -0.10cm of b8]   {};

\node []  (k6)  [right=0.6cm of b9]  {$k_{n_1}$}; 

\node [port] (b10) [below left=0.85 and -0.76cm of b8]  {$p_{r_1}$};
\node[bubble] (b11) [below left=-0.40 and -0.10cm of b10]   {};
 
 \node []  (k7)  [right=0.6cm of b11]  {$k_{r_1}$}; 
 
\node[type] (b12) [above left=-0.55cm and -2.4cm of b1]{Topic $1$};
\node[type]  [below=-0.1cm of b12]{$wB(3)$};

\node [component] (c1) [above right= -1.6cm and 3cm of b1] {};
\node [port] (c2) [below left=-1.07cm and -0.77cm of c1]  {$p_{d_1}$}; 
\node[bubble] (c3) [below left=-0.40cm and -0.10cm of c2]   {};  

 \node []  (k8)  [right=0.6cm of c3]  {$k_{d_1}$}; 

\node [port] (c4) [above=0.25cm and -0.75cm of c2]  {$p_{g_1}$}; 
\node[bubble] (c5) [above left=-0.40cm and -0.10cm of c4]   {};   

 \node []  (k9)  [right=0.6cm of c5]  {$k_{g_1}$}; 

\node [port] (c6) [above=0.3cm of c4]  {$p_{e_1}$}; 
\node[bubble] (c7) [above left=-0.40cm and -0.10cm of c6]   {};  

 \node []  (k10)  [right=0.6cm of c7]  {$k_{e_1}$}; 
 
\node[type] (c8) [above right=-0.55cm and -2.1cm of c1]{Subs. $1$};
\node[type]  [below=-0.1cm of c8]{$wB(5)$};

\node [component] (d1) [below right=-0.8cm and 3cm of b1] {};
\node [port] (d2) [below left=-1.07cm and -0.77cm of d1]  {$p_{d_2}$}; 
\node[bubble] (d3) [below left=-0.40cm and -0.10cm of d2]   {};  

 \node []  (k11)  [right=0.6cm of d3]  {$k_{d_2}$}; 

\node [port] (d4) [above=0.25cm and -0.75cm of d2]  {$p_{g_2}$}; 
\node[bubble] (d5) [above left=-0.40cm and -0.10cm of d4]   {};   

\node []  (k12)  [right=0.6cm of d5]  {$k_{g_2}$}; 

\node [port] (d6) [above=0.3cm of d4]  {$p_{e_2}$}; 
\node[bubble] (d7) [above left=-0.40cm and -0.10cm of d6]   {};   

 \node []  (k13)  [right=0.6cm of d7]  {$k_{e_2}$}; 

\node[type] (d8) [above right=-0.55cm and -2.1cm of d1]{Subs. $2$};
\node[type]  [below=-0.1cm of d8]{$wB(6)$};


\node [component,align=center] (e1) [below=4.5cm of a1] {};
\node [port] (e2) [above right=-1.66 and -0.78cm of e1]  {$p_{a_2}$};
\node[bubble] (e3) [below right=-0.40 and -0.10cm of e2]   {};

 \node []  (k14)  [left=0.55cm of e3]  {$k_{a_2}$}; 

\node [port] (e4) [below right=0.86cm and -0.73cm of e2]  {$p_{t_2}$};
\node[bubble] (e5) [below right=-0.37 and -0.10cm of e4]   {};

 \node []  (k15)  [left=0.5cm of e5]  {$k_{t_2}$}; 

\node[type] (e6) [above=-0.45cm of e1]{Publ. $2$};
\node[type]  [below=-0.05cm of e6]{$wB(2)$};

2nd topic

\node [component] (f1) [right=3cm of e1] {};
\node [port] (f2) [below right=-1.07cm and -0.77cm of f1]  {$p_{f_2}$}; 
\node[bubble] (f3) [below right=-0.40cm and -0.10cm of f2]   {};   

 \node []  (k16)  [left=0.5cm of f3]  {$k_{f_2}$}; 

\node [port] (f4) [above right=0.25cm and -0.75cm of f2]  {$p_{s_2}$}; 
\node[bubble] (f5) [above right=-0.40cm and -0.10cm of f4]   {};   

 \node []  (k17)  [left=0.5cm of f5]  {$k_{s_2}$}; 

\node [port] (f6) [above =0.3cm of f4]  {$p_{c_2}$}; 
\node[bubble] (f7) [above right=-0.40cm and -0.10cm of f6]   {};   

 \node []  (k18)  [left=0.5cm of f7]  {$k_{c_2}$}; 

\node [port] (f8) [above left= -1.7cm and -0.79cm of f1]  {$p_{n_2}$};
\node[bubble] (f9) [above left=-0.4cm and -0.10cm of f8]   {};

 \node []  (k19)  [right=0.6cm of f9]  {$k_{n_2}$}; 

\node [port] (f10) [below left= 0.85cm and -0.76cm of f8]  {$p_{r_2}$};
\node[bubble] (f11) [below left=-0.40 and -0.10cm of f10]   {};
 
  \node []  (k20)  [right=0.6cm of f11]  {$k_{r_2}$}; 
 
\node[type] (f12) [above left=-0.55cm and -2.4cm of f1]{Topic $2$};
\node[type]  [below=-0.1cm of f12]{$wB(4)$};

\node [component] (g1) [right= 3cm of f1] {};
\node [port] (g2) [below left=-1.07cm and -0.77cm of g1]  {$p_{d_3}$}; 
\node[bubble] (g3) [below left=-0.40cm and -0.10cm of g2]   {};  

 \node []  (k21)  [right=0.6cm of g3]  {$k_{d_3}$}; 

\node [port] (g4) [above left=0.25cm and -0.75cm of g2]  {$p_{g_3}$}; 
\node[bubble] (g5) [above left=-0.40cm and -0.10cm of g4]   {};   

 \node []  (k22)  [right=0.6cm of g5]  {$k_{g_3}$}; 

\node [port] (g6) [above=0.3cm of g4]  {$p_{e_3}$}; 
\node[bubble] (g7) [above left=-0.40cm and -0.10cm of g6]   {};  

 \node []  (k23)  [right=0.6cm of g7]  {$k_{e_3}$}; 
 
\node[type] (g8) [above right=-0.55cm and -2.1cm of g1]{Subs. $3$};
\node[type]  [below=-0.1cm of g8]{$wB(7)$};

\path[-]          (a3)  edge                  node {} (b9);
\path[-]          (b9)  edge                  node {} (a3);

\path[-]          (a5)  edge                  node {} (b11);
\path[-]          (b11)  edge                  node {} (a5);

\path[-]          (b7)  edge  [darkorange]                node {} (c7);
\path[-]          (c7)  edge  [darkorange]                 node {} (b7);

\path[-]          (b7)  edge  [darkorange]                 node {} (g7);
\path[-]          (g7)  edge  [darkorange]                 node {} (b7);

\path[-]          (b5)  edge   [darkorange]                node {} (c5);
\path[-]          (c5)  edge  [darkorange]                 node {} (b5);

\path[-]          (b5)  edge  [darkorange]                 node {} (g5);
\path[-]          (g5)  edge  [darkorange]                 node {} (b5);

\path[-]          (b3)  edge  [darkorange]                 node {} (c3);
\path[-]          (c3)  edge [darkorange]                  node {} (b3);

\path[-]          (b3)  edge   [darkorange]                node {} (g3);
\path[-]          (g3)  edge   [darkorange]                node {} (b3);


\path[-]          (e3)  edge                  node {} (f9);
\path[-]          (f9)  edge                  node {} (e3);

\path[-]          (e5)  edge                  node {} (f11);
\path[-]          (f11)  edge                  node {} (e5);

\path[-]          (a3)  edge                  node {} (f9);
\path[-]          (f9)  edge                  node {} (a3);

\path[-]          (a5)  edge                  node {} (f11);
\path[-]          (f11)  edge                  node {} (a5);


\path[-]          (f7)  edge  [darkorange]                 node {} (d7);
\path[-]          (d7)  edge   [darkorange]                node {} (f7);

\path[-]          (f5)  edge   [darkorange]                 node {} (d5);
\path[-]          (d5)  edge   [darkorange]                 node {} (f5);

\path[-]          (f3)  edge  [darkorange]                  node {} (d3);
\path[-]          (d3)  edge  [darkorange]                  node {} (f3);

\end{tikzpicture}}
\caption{A possible execution of the interactions in a weighted Publish/Subscribe architecture.}
\label{w-b-p-s}
\end{figure}
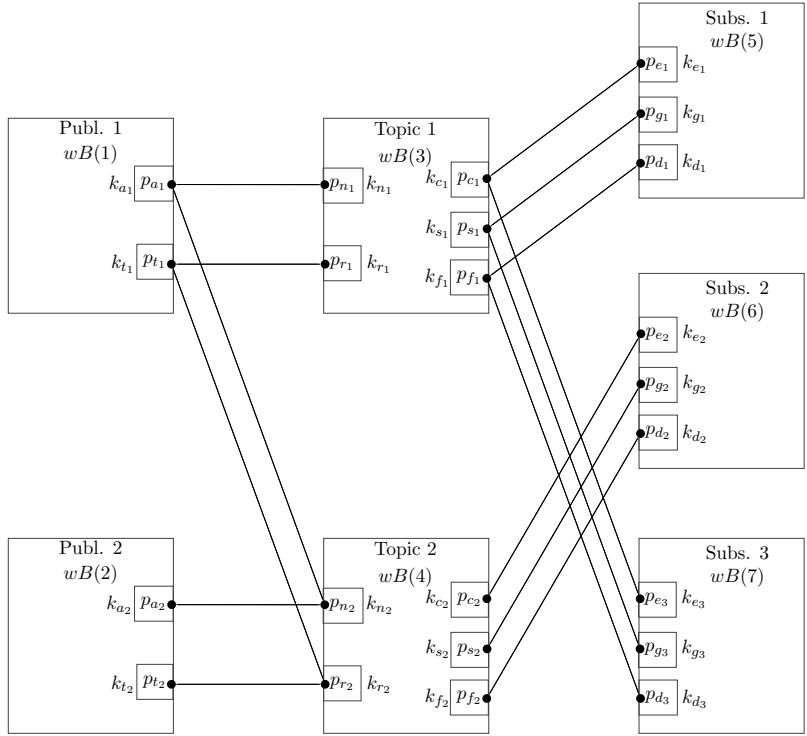
\end{exa}

\subsection{Parametric weighted component-based systems}

In this subsection we deal with the parametric extension of weighted component-based systems, i.e., with systems comprised of a finite number of \emph{weighted component types} where the cardinality 
of the \emph{instances} of each type is a parameter for the system. 

Let $w\B=\{wB(i) \mid i \in [n]  \}$ be a set of weighted component types. For every $i \in [n] $ and $j \geq 1$ we consider a weighted component   $wB(i,j)=(B(i,j), wt(i))$, where $B(i,j)=(Q(i,j),P(i,j),q_{0}(i,j),R(i,j))$ is the $j$-th instance of $B(i)$, and it is called a \emph{parametric weighted component} or a \emph{weighted component instance}. We set $pw\mathcal{B}= \{wB(i,j) \mid i \in [n],  j \geq 1 \} $ and call it a set of \emph{parametric weighted components}. We impose on $pw\B$ the same  assumptions as for $p\B$. Abusing notations, we denote by $wt(i)$, $i \in [n]$, the weight mapping of $wB(i,j)$, $j\geq 1$, meaning that it assigns the value $wt(i)(p)$ to every port $p(j)\in P(i,j)$. 

Now we can introduce the  weighted first-order extended interaction logic as a modelling language for describing the weight of the interactions in parametric weighted component-based systems.

As in FOEIL we equip wEPIL formulas with variables. Let $pw\mathcal{B}= \{wB(i,j) \mid i \in [n], j \geq 1 \} $ be a set of parametric weighted components. We consider pairwise disjoint countable sets of first-order variables  $\mathcal{X}^{(1)}, \ldots , \mathcal{X}^{(n)}$ referring to  instances of weighted component types $wB(1), \ldots, wB(n)$, respectively. 

\begin{defi} \label{def_wfOEIL}
Let $pw\mathcal{B}= \{wB(i,j) \mid i \in [n], j \geq 1 \} $ be a set of parametric weighted components. Then the syntax of weighted first-order extended interaction logic (wFOEIL for short) formulas $\ps$ over $pw\mathcal{B}$ and $K$ is given by the grammar
\begin{multline*}
\ps  ::= k \mid \psi   \mid  \ps \oplus \ps \mid \ps \otimes \ps \mid   \ps \odot \ps   \mid \ps\varpi \ps \mid \ps^{\oplus}\mid  {\textstyle\sum x^{(i)}}. \ps \mid  {\textstyle\prod x^{(i)}} . \ps \mid  \\ {\textstyle\sum\nolimits{^{\odot}}} x^{(i)}. \ps \mid {\textstyle\prod\nolimits{^{\odot}}x^{(i)}}.\ps \mid {\textstyle\sum\nolimits{^{\varpi}} x^{(i)}}.\ps \mid {\textstyle\prod\nolimits{^{\varpi}}x^{(i)}}.\ps  
\end{multline*}
where $k \in K$, $\psi$ is a FOEIL formula over $p\B$,  $x^{(i)}, y^{(i)}$ are first-order variables in $\mathcal{X}^{(i)}$, ${\textstyle\sum}$ (resp. ${\textstyle\prod}$) denotes the weighted existential (resp. universal) quantifier, ${\textstyle\sum\nolimits{^{\odot}}}$ (resp. $ {\textstyle\prod\nolimits{^{\odot}}}$) denotes the weighted existential (resp. universal) concatenation quantifier, and ${\textstyle\sum\nolimits{^{\varpi}}}$ (resp. ${\textstyle\prod\nolimits{^{\varpi}}}$) the weighted existential (resp. universal) shuffle quantifier. Furthermore, we assume that when $\ps$ contains a subformula of the form $ {\textstyle\sum\nolimits{^{\odot}}} x^{(i)}. \ps'$ or $ {\textstyle\sum\nolimits{^{\varpi}} x^{(i)}}.\ps'$, and $\ps'$ contains a FOEIL formula $\psi$, then the application of negation in $\psi$ is permitted only in PIL formulas, and formulas of the form $x^{(j)}=y^{(j)}$.
\end{defi}

Let $\ps$ be a wFOEIL formula over $pw\B$ and $r:[n] \rightarrow \mathbb{N}$ a mapping. As for (unweighted) parametric systems the value $r(i)$, for every $i \in [n]$, intends to represent the finite number of instances of the weighted component type $wB(i)$ in the parametric system. 
We let $pw\B(r)= \{wB(i,j) \mid i \in [n], j \in [r(i)] \} $ and call it  the \emph{instantiation of} $pw\B$ w.r.t. $r$. The set of ports and the set of interactions of $pw\B(r)$ are the same as the corresponding ones in $p\B(r)$, hence we use for simplicity
the same symbols $P_{p\B(r)}$ and $I_{p\B(r)}$, respectively.

We interpret wFOEIL formulas $\ps$ as series $\Vert \ps  \Vert$ over triples consisting of a mapping $r:[n] \rightarrow \mathbb{N}$, a $(\mathcal{V},r)$-assignment $\sigma$, and a word $w \in I_{p\B(r)}^+$, and $K$. Intuitively, the use of weighted existential and universal concatenation (resp. shuffle) quantifiers ${\textstyle\sum\nolimits{^{\odot}}} x^{(i)}. \ps$ and ${\textstyle\prod\nolimits{^{\odot}}x^{(i)}}.\ps$ (resp. ${\textstyle\sum\nolimits{^{\varpi}} x^{(i)}}.\ps$ and ${\textstyle\prod\nolimits{^{\varpi}}x^{(i)}}.\ps $) serves to compute the weight of the partial and whole participation of the weighted component instances, determined by the application of the assignment $\sigma$ to $x^{(i)}$, in sequential (resp. interleaving) interactions.

\begin{defi}
Let $\ps$ be a wFOEIL formula over a set $pw\mathcal{B}= \{wB(i,j) \mid i \in [n], j \geq 1 \} $ of parametric weighted components and $K$, and $\mathcal{V} \subseteq \mathcal{X}$ a finite set containing $\mathrm{free}(\psi)$. Then for every $r:[n] \rightarrow \mathbb{N}$, $(\mathcal{V},r)$-assignment $\sigma$, and $w  \in I_{p\B(r)}^+$ we define the value   $\Vert\ps\Vert(r,\sigma,w)$, inductively on the structure of $\ps$ as follows:

\begin{itemize}

     \item[-] $\Vert k \Vert(r,\sigma,w) = k$, 
     
     \item[-] $\Vert\psi\Vert(r,\sigma,w)\left\{
\begin{array}
[c]{rl}%
1 & \textnormal{ if }(r, \sigma,w) \models \psi\\
0 & \textnormal{ otherwise}%
\end{array}
,\right.  $

     \item[-]  $\Vert \ps_1 \oplus \ps_2 \Vert(r,\sigma, w) =\Vert \ps_1 \Vert(r,\sigma, w) + \Vert  \ps_2 \Vert(r,\sigma, w)$, 
 
\item[-]  $\Vert \ps_1 \otimes \ps_2 \Vert(r,\sigma, w) =\Vert \ps_1 \Vert(r,\sigma, w) \cdot \Vert  \ps_2 \Vert(r,\sigma, w)$,

 \item[-]  $\Vert \ps_1 \odot \ps_2 \Vert(r,\sigma, w) =\sum\nolimits_{w=w_1w_2}(\Vert\ps_1\Vert(r,\sigma,w_1) \cdot \Vert\ps_2\Vert(r,\sigma, w_2)) $,

\item[-]  $\Vert \ps_1 \varpi \ps_2 \Vert(r,\sigma, w) =\sum\nolimits_{w\in w_1\shuffle w_2}(\Vert\ps_1\Vert(r,\sigma,w_1) \cdot \Vert\ps_2\Vert(r, \sigma,w_2))$,

\item[-]  $\Vert \ps^{\oplus} \Vert(r,\sigma, w) =\sum\nolimits_{\nu \geq 1}\Vert\ps\Vert^{\nu}(r,\sigma,w) $,

     \item[-]   $\left\Vert {\textstyle\sum x^{(i)}}. \ps \right \Vert(r,\sigma, w) = \sum\limits_{j \in [r(i)]}  \Vert \ps   \Vert (r, \sigma[x^{(i)}\rightarrow j],w)$,

     \item[-]   $\left\Vert {\textstyle\prod x^{(i)}}. \ps \right \Vert(r,\sigma, w) = \prod\limits_{j \in [r(i)]}  \Vert \ps   \Vert (r, \sigma[x^{(i)}\rightarrow j],w)$,

 \item[-]   $\left\Vert{\textstyle\sum\nolimits{^{\odot}}} x^{(i)}. \ps\right\Vert(r,\sigma,w) = \sum\limits_{1 \leq t \leq r(i)}\sum\limits_{1 \leq l_1 < \ldots <l_t \leq r(i)}\sum\limits_{w=w_{l_1}\ldots w_{l_t}}\prod\limits_{j=l_1, \ldots,l_t}  \Vert \ps   \Vert (r, \sigma[x^{(i)}\rightarrow j],w_j)$,

 \item[-]   $\left\Vert{\textstyle\prod\nolimits{^{\odot}}x^{(i)}}. \ps\right\Vert(r,\sigma,w) = \sum\limits_{w=w_1\ldots w_{r(i)}}\prod\limits_{1\leq j \leq r(i)}  \Vert \ps   \Vert (r, \sigma[x^{(i)}\rightarrow j],w_j)$,

\item[-]   $\left\Vert{\textstyle\sum\nolimits{^{\varpi}}} x^{(i)}. \ps\right\Vert(r,\sigma,w) = \sum\limits_{1 \leq t \leq r(i)}\sum\limits_{1 \leq l_1 < \ldots <l_t \leq r(i)}\sum\limits_{w\in w_{l_1} \shuffle \ldots \shuffle w_{l_t}}\prod\limits_{j=l_1, \ldots,l_t}  \Vert \ps   \Vert (r, \sigma[x^{(i)}\rightarrow j],w_j)$,

\item[-]   $\left\Vert{\textstyle\prod\nolimits{^{\varpi}}} x^{(i)}. \ps\right\Vert(r,\sigma,w) = \sum\limits_{w\in w_1 \shuffle \ldots \shuffle w_{r(i)}}\prod\limits_{1 \leq j \leq r(i)}  \Vert \ps   \Vert (r, \sigma[x^{(i)}\rightarrow j],w_j)$.    
\end{itemize}
\end{defi}

\noindent Similarly to the unweighted setup, all instances of each weighted component type in parametric weighted component-based systems are identical. Hence the order specified above in the semantics of  $\textstyle\sum\nolimits{^{\odot}},\textstyle\prod\nolimits{^{\odot}}, {\textstyle\sum\nolimits{^{\varpi}}},\textstyle\prod\nolimits{^{\varpi}}$ weighted quantifiers causes no restriction in the resulting architecture.

 If $\ps$ is a wFOEIL sentence over $pw\B$ and $K$, then we simply write $\Vert\ps\Vert(r, w)$. Let also $\ps'$ be a wFOEIL sentence over $pw\mathcal{B}$ and $K$. Then,  $\ps$ and $\ps'$ are called \emph{equivalent w.r.t.} $r$ whenever $\Vert\ps\Vert(r, w)=\Vert\ps'\Vert(r, w)$, for every $w  \in I_{p\B(r)}^+$.

Now we are ready to formally define the concept of a parametric weighted component-based system.

\begin{defi}
A parametric weighted component-based system over $K$ is a pair $(pw\B, \ps)$ where $pw\B=\{wB(i,j) \mid i\in [n], j \geq 1 \}$ is a set of parametric weighted components and $\ps$ is a wFOEIL sentence over $pw\B$ and $K$.  
\end{defi}

In the sequel, for simplicity we refer to parametric weighted component-based systems by parametric weighted systems, and we often omit the term ``weighted" when we refer to instances.

For our examples in the sequel, we need the following macro wFOEIL formula. Let $pw\B=\{wB(i,j) \mid i\in [n], j \geq 1\}$ and $1 \leq i_1, \ldots, i_m \leq n$ be pairwise different indices. Let $p_{i_1} \in P(i_1), \ldots, p_{i_m} \in P(i_m)$ and $k_{i_1}, \ldots, k_{i_m}$ denote the weights in $K$ assigned to $p_{i_1}, \ldots, p_{i_m}$, respectively, i.e., $k_{i_1} =wt(i_1)(p_{i_1}), \ldots,  k_{i_m}=wt(i_m)(p_{i_m})$. We set
\begin{multline*}
\#_w\big(p_{i_1}(x^{(i_1)}) \otimes \ldots \otimes p_{i_m}(x^{(i_m)})\big)::= \big ((k_{i_1}\otimes p_{i_1}(x^{(i_1)})) \otimes  \ldots \otimes (k_{i_m}\otimes p_{i_m}(x^{(i_m)}))\big) \otimes \\ 
\bigg(\bigg(\bigwedge_{j =i_1, \ldots, i_m } \bigwedge_{p \in P(j)\setminus \{ p_j\}} \neg p(x^{(j)}) \bigg) \wedge 
\bigg(\bigwedge_{j=i_1, \ldots, i_m} \forall y^{(j)}(y^{(j)} \neq x^{(j)}).\bigwedge_{p \in P(j)} \neg p(y^{(j)}) \bigg ) \wedge \\
 \bigg( \bigwedge_{\tau \in [n]\setminus \{i_1, \ldots, i_m\} }\bigwedge_{p \in P(\tau)}\forall x^{(\tau)} .   \neg p(x^{(\tau)}) \bigg )\bigg) .
\end{multline*}

The weighted conjunctions in the right-hand side of the first line, in the above formula,   express that the ports appearing in the argument of $\#_w$ participate in the interaction with their corresponding weights. In the second line, the   
  double indexed conjunctions in the first pair of big parentheses disable all the other ports of the participating instances of weighted components of type $i_1, \ldots, i_m$ described by variables $x^{(i_1)}, \ldots, x^{(i_m)}$, respectively; conjunctions in the second pair of parentheses disable all ports of remaining instances of weighted component types $i_1, \ldots, i_m$. Finally, the last conjunct in the third line ensures that no ports in instances of remaining weighted component types
participate in the interaction. Then we get 
\begin{multline*}
\#_w\big(p_{i_1}(x^{(i_1)}) \otimes \ldots \otimes p_{i_m}(x^{(i_m)})\big)\equiv \big (k_{i_1}\otimes \ldots \otimes k_{i_m}\big) \otimes \bigg(\big(p_{i_1}(x^{(i_1)}) \wedge  \ldots \wedge p_{i_m}(x^{(i_m)})\big) \\ \wedge 
\bigg(\bigwedge_{j =i_1, \ldots, i_m } \bigwedge_{p \in P(j)\setminus \{ p_j\}} \neg p(x^{(j)}) \bigg) \wedge 
\bigg(\bigwedge_{j=i_1, \ldots, i_m} \forall y^{(j)}(y^{(j)} \neq x^{(j)}).\bigwedge_{p \in P(j)} \neg p(y^{(j)}) \bigg ) \wedge \\
 \bigg( \bigwedge_{\tau \in [n]\setminus \{i_1, \ldots, i_m\} }\bigwedge_{p \in P(\tau)}\forall x^{(\tau)} .   \neg p(x^{(\tau)}) \bigg )\bigg) .
\end{multline*} 

Next we present three examples of wFOEIL sentences describing concrete parametric architectures with ordered interactions in the weighted setup. We note that as for the examples of FOEIL sentences,
 whenever is defined a unique instance for a weighted component type we may also consider the corresponding set of variables as a singleton.

\begin{exa}
\label{wblackboard}
\textbf{(Parametric weighted Blackboard)} We consider weighted Blackboard architecture, described in Example \ref{ex_b_wblackboard}, in the parametric setting. Therefore, blackboard's instance interacts with controller's instance and all the source instances for presenting the current state of the problem. Then, some of the source instances are triggered and logged in the controller in arbitrary order, and in turn these triggered source instances add the new information on blackboard through the controller, again in any order. Hence we need to describe the overall weight for any possible scenario of the aforementioned interactions among the several instances. 
We consider three sets of variables, namely $\mathcal{X}^{(1)}, \mathcal{X}^{(2)}, \mathcal{X}^{(3)}$ for the instances of blackboard, controller, and knowledge source weighted components, respectively. Therefore,  the wFOEIL sentence $\ps$ that encodes the cost of the interactions of parametric weighted Blackboard architecture is
\begin{multline*}
\ps=\bigg({\textstyle\sum x^{(1)}} {\textstyle\sum x^{(2)}}.\bigg( \#_w (p_d(x^{(1)})\otimes p_r(x^{(2)})) \odot \bigg({\textstyle\prod^{\varpi} x^{(3)}}.\#_w (p_d(x^{(1)})\otimes p_n(x^{(3)}))\bigg) \odot \\ \bigg({\textstyle\sum^{\varpi} y^{(3)}}. \big(\#_w(p_l(x^{(2)})\otimes p_t(y^{(3)})) \odot  \#_w (p_e(x^{(2)})\otimes p_w(y^{(3)})\otimes p_a(x^{(1)}))\big) \bigg)^{\oplus}\bigg)\bigg)^{\oplus}.
\end{multline*}
The weighted subformula $\#_w (p_d(x^{(1)})\otimes p_r(x^{(2)}))$ with the associated quantifiers ${\textstyle\sum x^{(1)}}, {\textstyle\sum x^{(2)}}$  expresses the weight for the connection of blackboard with the controller in order for the latter to be informed for the available information. The subformula with the weighted  universal shuffle quantifier ${\textstyle\prod^{\varpi} x^{(3)}}$ captures the weight for the respective connections of blackboard with each of the source instances, in arbitrary order. Then, the last subformula with the weighted existential shuffle quantifier ${\textstyle\sum^{\varpi} y^{(3)}}$ encodes the weight for the triggering process of some source instances, and in turn for updating the information in blackboard through the controller. These interactions may be applied in any order for some instances of the source weighted component, which justifies the use of ${\textstyle\sum^{\varpi}}$ for the variables $y^{(3)}$. The application of the weighted iteration operators allows computing the cost of the repetition of the respective interactions in the architecture. An instantiation of the parametric weighted architecture for three sources is presented in Figure \ref{wb_blackboard} of Example \ref{ex_b_wblackboard}.
\end{exa}

\begin{exa}
\label{w-re-re}
\textbf{(Parametric weighted Request/Response)} Next we consider the weighted Request/Response architecture, described in Example \ref{w-b-r-r}, in the parametric setting. 
We recall that all the service instances interact with registry through interleaving for their enrollment, and in turn all the client instances interact with registry in any order for taking the address of the service instances. Then, for some of the service instances there exist some client instances that interact with each of them through their associated coordinator instance. The interactions of distinct client instances take place sequentially for the same service instance and with interleaving, otherwise. Next we present a wFOEIL sentence $\ps$ encoding the weight of any possible combination of the permissible interactions in parametric weighted Request/Response architecture described above.
We consider the variable sets $\mathcal{X}^{(1)},\mathcal{X}^{(2)},\mathcal{X}^{(3)}$, and $\mathcal{X}^{(4)}$  referring to weighted component instances of service registry, service, client, and coordinator component, respectively. Therefore,
\begin{multline*}
\ps= \bigg( {\textstyle\sum x^{(1)}}.\bigg(\big({\textstyle\prod^{\varpi} x^{(2)}}. \#_w (p_e(x^{(1)})\otimes p_r(x^{(2)}))\big) \odot \\ \big({\textstyle\prod^{\varpi} x^{(3)}}. (\#_w (p_l(x^{(3)})\otimes p_u(x^{(1)})) \odot \#_w(p_o(x^{(3)})\otimes p_t(x^{(1)}))) \big )\bigg)\bigg) \odot \\ \bigg({\textstyle\sum^{\varpi} y^{(2)}}{\textstyle\sum x^{(4)}}{\textstyle\sum^{\odot}y^{(3)}}.\x \otimes \bigg(\forall y^{(4)}\forall z^{(3)}\forall z^{(2)}.\big( \theta \vee \big(\forall t^{(3)}\forall t^{(2)}(z^{(2)}\neq t^{(2)}).\theta'\big)\big)\bigg)\bigg)^{\oplus}
\end{multline*}
\noindent where the wEPIL formula $\x$ is  given by:

\

$ \x= \#_w (p_n(y^{(3)})\otimes p_m(x^{(4)})) \odot \#_w (p_q(y^{(3)})\otimes p_a(x^{(4)})\otimes p_g (y^{(2)})) \odot  \\ \text{\qquad \qquad \qquad \qquad \qquad \qquad \qquad \qquad \qquad\qquad \qquad } \#_w (p_c(y^{(3)})  \otimes p_d(x^{(4)}) \otimes p_s(y^{(2)}))$,

\

\noindent and the EPIL formulas $\theta$ and $\theta'$ are given respectively, by:

$\theta=\neg ( (p_q(z^{(3)})\wedge p_a(y^{(4)})\wedge p_g (z^{(2)}))\shuffle \mathrm{true})$,

\

\noindent and

\

$\theta' =( \# (p_q(z^{(3)})\wedge p_a(y^{(4)})\wedge p_g (z^{(2)}))\shuffle\mathrm{true}) \wedge \\
 \text{ \qquad \qquad \qquad \qquad \qquad \qquad \qquad \ \ \ \ } \neg(  (p_q(t^{(3)})\wedge p_a(y^{(4)})\wedge p_g (t^{(2)}))\shuffle\mathrm{true}).$

\

\noindent The wFOEIL subformula at the first line of $\ps$ with the weighted universal shuffle quantifier $ {\textstyle\prod^{\varpi} x^{(2)}}$ encodes the weight of the interleaving interactions between registry instance ($ {\textstyle\sum x^{(1)}}$) and each of the service instances for their enrollment. The wFOEIL subformula in the second line of $\ps$ captures the weight of the interactions between registry and each of the client instances in order for the latter to connect and take the services' address from registry. These interactions take place in arbitrary order for each of the distinct client instances, hence their weight should be computed accordingly, which is ensured by the use of the weighted universal shuffle quantifier $ {\textstyle\prod^{\varpi} x^{(3)}}$. The wFOEIL  subformula in the third line of $\ps$ expresses the weight of the connections among client and service instances through their coordinator instance, applying the quantifiers ${\textstyle\sum^{\varpi} y^{(2)}}$,${\textstyle\sum  x^{(4)}}$, and ${\textstyle\sum^{\odot} y^{(3)}}$. The use of quantifier ${\textstyle\sum^{\odot} y^{(3)}}$ is justified by the fact that only one client instance should interact with each service instance, and hence the respective weight is computed analogously. On the other hand, for different instances of services interleaving among several client instances is permitted, and hence we derive the respective weight by the wFOEIL subformula quantified by ${\textstyle\sum^{\varpi} y^{(2)}}$. The EPIL subformula $\forall y^{(4)}\forall z^{(3)} \forall z^{(2)} .\big( \theta \vee \big(\forall t^{(3)}\forall t^{(2)}(z^{(2)}\neq t^{(2)}).\theta'\big)\big)$ in $\ps$ serves as a constraint to ensure that a unique coordinator instance is assigned to each service instance. Finally, the weighted iteration operator returns the cost of implementing the corresponding interactions with recursion. An instantiation of the parametric weighted architecture for two clients and services is presented in Figure \ref{w_b_re_re} of Example \ref{w-b-r-r}.   
\end{exa}

\begin{exa}\textbf{(Parametric weighted Publish/Subscribe)}
\label{w-pu-su}
We consider weighted Publish/Subscribe architecture, described in Example \ref{w-b-pu-su}, in the parametric setting. We have that for some of the topic instances there are some publisher instances that advertise and transmit their messages. In turn, the same topic instances perform three consecutive interactions with some of the subscriber instances, in order for the latter to express their interest in some messages, receive the requested messages, and disconnect from the topic instances. In each case the interactions among the distinct instances are executed with interleaving. In the subsequent wFOEIL sentence $\ps$ we encode the weight of all the possible cases for the aforementioned permissible interactions. We let $\mathcal{X}^{(1)},\mathcal{X}^{(2)},\mathcal{X}^{(3)}$ denote the variable sets that correspond to publisher, topic, and subscriber weighted component instances, respectively.
\begin{multline*}
\ps= \Bigg({\textstyle\sum^{\varpi} x^{(2)}}. \Bigg( \bigg( {\textstyle\sum^{\varpi} x^{(1)}}.\big(\#_w (p_a(x^{(1)})\otimes p_n (x^{(2)})) \odot \#_w (p_t(x^{(1)})\otimes p_r(x^{(2)}))\big) \bigg)  
\odot \\ \bigg({\textstyle\sum^{\varpi} x^{(3)}}. \big( \#_w (p_e(x^{(3)})\otimes p_c(x^{(2)})) \odot \#_w(p_g(x^{(3)})\otimes p_s (x^{(2)})) \odot  \#_w(p_d(x^{(3)})\otimes p_f (x^{(2)}))\big)\bigg)\Bigg)\Bigg)^{\oplus}.
\end{multline*}
\end{exa}

We interpret the wFOEIL sentence $\ps$ as follows. The wFOEIL subformula in the first line encodes the weight for the interactions among some of the topic instances (weighted existential shuffle quantifier ${\textstyle\sum^{\varpi} x^{(2)}}$) and some publisher instances (weighted existential shuffle quantifier ${\textstyle\sum^{\varpi} x^{(1)}}$) in order for the latter to advertise and in turn transmit their messages to the former. These interactions may take place in any order for the distinct instances, hence the associated weight needs to be computed accordingly. Then, in the second line the subformula with the weighted existential shuffle quantifier ${\textstyle\sum^{\varpi} x^{(3)}}$ captures the weight of three sequential interactions between each of the topic instances, specified in the first line, and some of the subscriber instances. Since there is no restriction for the execution order of these interactions with respect to the distinct instances, we describe their weight with the corresponding weighted shuffle quantifiers. Finally, the use of the weighted iteration operator serves for computing the cost of the ongoing implementation of the architecture. An instantiation of the parametric weighted architecture is shown in Figure \ref{w-b-p-s} of Example \ref{w-b-pu-su}.

Parametric architectures of the above examples were considered in \cite{Bo:St,Ma:Co}, in the qualitative setting, and in the weighted  setup in \cite{Pa:On}. Though, neither the execution order nor the recursion of interactions was assumed in the work of \cite{Bo:St,Ma:Co,Pa:On}.  

Note that similarly to FOEIL, its weighted counterpart, wFOEIL, can also be applied for the quantitative modelling of parametric architectures without order restrictions, i.e., for Examples \ref{ma-sl}-\ref{repo}. Next we present the wFOEIL sentences describing the parametric weighted Master/Slave and Repository, while the rest of the architectures can be described in the weighted setup analogously.

\begin{exa}
\label{w-ma-sl}
\textbf{(Parametric weighted Master/Slave)} We consider the parametric Master/Slave architecture, described in Example \ref{ma-sl}, in the weighted setting. We let $\mathcal{X}^{(1)}, \mathcal{X}^{(2)}$ denote the sets of variables of master and slave weighted component instances, respectively.  Then, the wFOEIL sentence $\ps$ representing parametric weighted Master/Slave architecture is
$$\ps={\textstyle\prod\nolimits{^{\odot}}x^{(2)}} {\textstyle\sum x^{(1)}}. \#_w( p_m(x^{(1)}) \otimes p_s(x^{(2)})). $$
\noindent In the above sentence the weighted universal concatenation quantifier
 accompanied with the existential one encodes the weight of the interactions between every slave instance and a master instance. These connections  
are applied consecutively for the distinct slave instances, hence their weight is computed accordingly.
 \noindent Consider the instantiation of the architecture with two master and two slave weighted component instances. We let $w_1$, $w_2$, $w_3$, and $w_4$ correspond to the 
four possible connections for the components in the system, defined as in the unweighted case and shown in Figure \ref{m-s}. Then, the values
$\Vert \ps \Vert (r,w_1)$, $\Vert \ps \Vert (r,w_2)$, $\Vert \ps \Vert (r,w_3)$,
 and $\Vert \ps \Vert (r,w_4)$
return the cost of the implementation of each of the four possible connections in the architecture, according to the underlying semiring.
In turn, the `sum' $\Vert \ps \Vert (r,w_1)+\Vert \ps \Vert (r,w_2)+\Vert \ps \Vert (r,w_3)
 +\Vert \ps \Vert (r,w_4)$
equals in the semiring of natural numbers for instance, to the total cost for executing
these connections in the architecture. 
\end{exa}

\begin{exa}
\label{w-repo}
\textbf{(Parametric weighted Repository)} The subsequent wFOEIL sentence $\ps$ characterizes the parametric Repository architecture, presented in Example \ref{repo}, with weighted features. We let $\mathcal{X}^{(1)}$ and $\mathcal{X}^{(2)}$ denote the variable sets that refer to instances of repository and data accessor weighted components, respectively. Then, 
$$ \ps= {\textstyle\sum x^{(1)}} {\textstyle\prod\nolimits{^{\odot}} x^{(2)}}.  \#_w(p_r(x^{(1)}) \otimes
p_d(x^{(2)})). $$
\noindent Let $a_1$, $a_2$, $a_3$, $a_4$ represent each of the four interactions for the architecture instantiation of Figure \ref{rep}, in the weighted setup. Then, the value
 $\Vert \ps \Vert (r,w)$ for $w=a_1a_2a_3a_4$ is the `total' cost for implementing the interactions in the system.
\end{exa}

\subsection{Decidability results for wFOEIL}\label{sec_dec_weighted}
In this subsection, we state an effective translation of wFOEIL sentences to weighted automata. For this, we use the corresponding result of Theorem \ref{sent_to_aut}, namely for every FOEIL sentence we can effectively construct, in exponential time, an expressively equivalent finite automaton. Then, we show that the equivalence of wFOEIL sentences over specific semirings is decidable. Firstly, we briefly recall basic notions and results on weighted automata. 

Let $A$ be an alphabet. A (nondeterministic) weighted finite automaton (WFA for short) over $A$ and $K$ is a quadruple $\mathcal{A}=(Q, in, wt, ter)$ where $Q$ is the finite state set, $in: Q \rightarrow K$ is the initial distribution, $wt:Q\times A \times Q \rightarrow K$ is the mapping assigning weights to transitions of $\mathcal{A}$, and $ter:Q \rightarrow K$ is the terminal distribution.  

Let $w=a_1 \ldots a_n \in A^*$. A path $P_w$ of $\mathcal{A}$ over $w$ is a sequence of transitions $P_w=((q_{i-1}, a_i, q_i))_{1 \leq i \leq n}$. The weight of $P_w$ is given by 
$weight(P_w)=in(q_0)\cdot\prod_{1 \leq i \leq n}wt(q_{i-1},a_i,q_i) \cdot ter(q_n)$. The behavior of $\mathcal{A}$ is the series $\Vert \mathcal{A} \Vert:A^* \rightarrow K$ which is determined by $\Vert \mathcal{A} \Vert(w)= \sum_{P_w}weight(P_w)$. 

Two WFA $\mathcal{A}$ and $\mathcal{A}'$ over $A$ and $K$ are called equivalent if $\Vert\mathcal{A}\Vert=\Vert\mathcal{A}'\Vert$.  For our translation algorithm, of wFOEIL formulas to WFA, we shall need folklore results in WFA theory. We collect them in the following proposition (cf. for instance \cite{Dr:Ha, Sa:El}). 

\begin{prop}\label{w-prop-rec} Let $\mathcal{A}_1=(Q_1, in_1, wt_1, ter_1)$, $\mathcal{A}_2=(Q_2, in_2, wt_2, ter_2)$ and $\mathcal{A}=(Q, in, wt, ter)$ be three \emph{WFA's} over $A$ and $K$. Then, we can construct in polynomial time \emph{WFA's} $\mathcal{B}, \mathcal{C}, \mathcal{D}, \mathcal{E}$ over $A$ and $K$ accepting the sum, and the Hadamard, Cauchy and shuffle product of $\Vert \mathcal{A}_1\Vert$ and  $\Vert \mathcal{A}_2\Vert$, respectively. Moreover, if $\Vert\mathcal{A}\Vert$ is proper, then we can construct in polynomial time \emph{WFA} $\mathcal{A}'$ over $A$ and $K$ accepting the iteration of $\Vert\mathcal{A}\Vert$. \end{prop}

Next we present the translation algorithm of  wFOEIL formulas to WFA's. Our algorithm requires a doubly exponential time  at its worst case. Specifically, we prove the following theorem.

\begin{thm}\label{wsent_to_waut}
Let $pw\B=\{wB(i,j) \mid i\in [n], j\geq 1\}$ be a set of parametric weighted components over a commutative semiring $K$,  and $r:[n] \rightarrow \mathbb{N}$. Then, for every  \emph{wFOEIL} sentence $\ps$ over $pw\B$ and $K$ we can effectively construct a \emph{WFA} $\mathcal{A}_{\ps,r}$ over $I_{p\B(r)}$ and $K$ such that $\Vert \ps \Vert (r,w)=\Vert\mathcal{A}_{\ps,r}\Vert(w)$ for every  $w \in I_{p\B(r)}^+$. The worst case run time for the translation algorithm is doubly exponential and the best case is exponential.    
\end{thm}

We shall prove Theorem \ref{wsent_to_waut} using the subsequent Proposition \ref{wformula-waut}. For this, we need to slightly modify the corresponding result of Proposition \ref{formula-aut}. More precisely, we state the next proposition.
\begin{prop}\label{form-aut}
Let $\psi$ be a \emph{FOEIL} formula over a set $p\B=\{B(i,j) \mid i\in [n], j\geq 1\}$ of parametric components. Let also $\mathcal{V} \subseteq \mathcal{X}$ be a finite set of variables containing $\mathrm{free}(\psi)$, $r:[n] \rightarrow \mathbb{N}$, and $\sigma$ a  $(\mathcal{V}, r)$-assignment. Then, we can effectively construct a finite automaton $\mathcal{A}_{\psi,r, \sigma}$ over $I_{p\B(r)}$ such that $(r, \sigma, w) \models \psi$ iff $w \in L(\mathcal{A}_{\psi,r, \sigma})$ for every $w \in I_{p\B(r)}^+$. The worst case run time for the translation algorithm is exponential and the best case is polynomial. 
\end{prop}
\begin{proof}
We modify the proof of Proposition \ref{formula-aut} as follows.
If $\psi=\mathrm{true}$, then we consider the complete finite automaton $\mathcal{A}_{\psi,r, \sigma}=(\{q\}, I_{p\B(r)},q, \Delta, \{q\})$ with $\Delta=\{(q,a,q) \mid   a \in I_{p\B(r)} \}$. If $\psi=p(x^{(i)})$, then we consider the deterministic finite automaton  $\mathcal{A}_{\psi, r, \sigma}=(\{q_0,q_1\}, I_{p\B(r)}, q_0, \Delta, \{q_1\})$ with $\Delta=\{(q_0,a,q_1) \mid   p(\sigma(x^{(i)})) \in a\}$. Then, we follow accordingly the same induction steps. Concerning the complexity of the translation we use the same arguments (we do not take into account the trivial case $\psi = \mathrm{true}$ where the complexity of the translation is constant). 
\end{proof}

\begin{prop}\label{wformula-waut}
Let $\ps$ be a \emph{wFOEIL} formula over a set $pw\B=\{wB(i,j) \mid i\in [n], j\geq 1\}$ of parametric weighted components and $K$. Let also $\mathcal{V} \subseteq \mathcal{X}$ be a finite set of variables containing $\mathrm{free}(\psi)$, $r:[n] \rightarrow \mathbb{N}$ and $\sigma$  a $(\mathcal{V}, r)$-assignment. Then, we can effectively construct a \emph{WFA} $\mathcal{A}_{\ps,r,\sigma}$ over $I_{p\B(r)}$ and $K$ such that $\Vert \ps \Vert(r, \sigma, w) = \Vert \mathcal{A}_{\ps,r, \sigma}\Vert(w)$ for every $w \in I_{p\B(r)}^+$. The worst case run time for the translation algorithm is doubly exponential and the best case is exponential.   
\end{prop}
\begin{proof}
We prove our claim by  induction on the structure of the wFOEIL formula $\ps$. 
\begin{itemize}
\item[i)] If $\ps=k$, then we consider the WFA $\mathcal{A}_{\ps,r,\sigma}=(\{q\}, in, wt, ter)$ over $I_{p\B(r)}$ and $K$ with $in(q)=k$, $wt(q,a,q)=1$ for every $a \in I_{p\B(r)}$, and $ter(q)=1$. 
\item[ii)] If $\ps=\psi$, then we consider the finite automaton $\mathcal{A}_{\psi,r,\sigma}$ derived in Proposition \ref{form-aut}. Next, we construct, in exponential time, an equivalent complete finite automaton $\mathcal{A}'_{\psi,r,\sigma}=(Q,I_{p\B(r)},q_0,\Delta,F)$. Then, we construct, in linear time,  the WFA $\mathcal{A}_{\ps,r,\sigma}=(Q,in,wt,ter)$ where $in(q)=1$ if $q=q_0$ and $in(q)=0$ otherwise, for every $q \in Q$, $wt(q,a,q')=1$ if $(q,a,q') \in \Delta$ and $wt(q,a,q')=0$ otherwise, for every $(q,a,q') \in Q\times A \times Q$, and $ter(q)=1$ if $q \in F$ and $ter(q)=0$ otherwise, for every $q\in Q$.
\item[iii)] If $\ps=\ps_1 \oplus \ps_2$ or $\ps=\ps_1 \otimes \ps_2$ or $\ps=\ps_1 \odot \ps_2$ or $\ps=\ps_1 \varpi \ps_2$, then we rename firstly variables in $\mathrm{free}(\ps_1) \cap \mathrm{free}(\ps_2)$ as well variables which are free in $\ps_1$ (resp. $\ps_2$) and bounded (i.e., not free) in $\ps_2$ (resp. in $\ps_1$). Then, we extend $\sigma$ on $\mathrm{free}(\ps_1) \cup \mathrm{free}(\ps_2)$ in the obvious way, and construct   $\mathcal{A}_{\ps,r,\sigma}$ from   $\mathcal{A}_{\ps_1,r,\sigma}$ and $\mathcal{A}_{\ps_2,r,\sigma}$  by applying Proposition \ref{w-prop-rec}. 
\item[iv)] If $\ps=\ps'{^{\oplus}}$, then we get $\mathcal{A}_{\ps,r,\sigma}$ as the WFA for the iteration of the series $\Vert \mathcal{A}_{\ps',r,\sigma}\Vert$ by applying Proposition \ref{w-prop-rec}.
\item[v)] If $\ps={\textstyle\sum x^{(i)}}. \ps' $, then we get $\mathcal{A}_{\ps,r,\sigma}$ as the WFA for the sum of the series $\Vert\mathcal{A}_{\ps',r,\sigma[x^{(i)}\rightarrow j]}\Vert$, $j\in [r(i)]$ (Proposition \ref{w-prop-rec}).
\item[vi)] If $\ps={\textstyle\prod x^{(i)}}. \ps'$, then we get $\mathcal{A}_{\ps,r,\sigma}$ as the WFA for the Hadamard product of the series $\Vert\mathcal{A}_{\ps',r,\sigma[x^{(i)}\rightarrow j]}\Vert$, $j\in [r(i)]$ (Proposition \ref{w-prop-rec}).
\item[vii)] If $\ps={\textstyle\sum ^{\odot} x^{(i)}}. \ps' $, then we compute firstly all nonempty  subsets $J$ of $[r(i)]$. For every such subset $J=\{l_1, \ldots, l_t\}$, with $1 \leq t \leq r(i)$ and $1 \leq l_1 <  \ldots <  l_t \leq r(i)$, we consider the WFA $\mathcal{A}_{\ps,r,\sigma}^{(J)}$ accepting the Cauchy product of the series $\Vert\mathcal{A}_{\ps',r,\sigma[x^{(i)}\rightarrow l_1]}\Vert,\ldots, \Vert\mathcal{A}_{\ps',r,\sigma[x^{(i)} \rightarrow l_t]}\Vert$. Then, we get $\mathcal{A}_{\ps,r,\sigma}$ as the WFA for the sum of the series  $\Vert\mathcal{A}_{\ps,r,\sigma}^{(J)}\Vert$  with $\emptyset \neq J \subseteq [r(i)]$ (Proposition \ref{w-prop-rec}). 
\item[viii)] If $\ps={\textstyle\prod ^{\odot} x^{(i)}}. \ps' $, then we get $\mathcal{A}_{\ps,r,\sigma}$ as the WFA for the Cauchy product of the series $\Vert\mathcal{A}_{\ps',r,\sigma[x^{(i)}\rightarrow j]}\Vert$, $j\in [r(i)]$ (Proposition \ref{w-prop-rec}).
\item[ix)] If  $\ps={\textstyle\sum ^{\varpi} x^{(i)}}. \ps' $, then we compute firstly all nonempty  subsets $J$ of $[r(i)]$. For every such subset $J=\{l_1, \ldots, l_t\}$, with $1 \leq t \leq r(i)$ and $1 \leq l_1 <  \ldots <  l_t \leq r(i)$, we consider the WFA $\mathcal{A}_{\ps,r,\sigma}^{(J)}$ accepting the shuffle product of the series $\Vert\mathcal{A}_{\ps',r,\sigma[x^{(i)}\rightarrow l_1]}\Vert,\ldots, \Vert\mathcal{A}_{\ps',r,\sigma[x^{(i)} \rightarrow l_t]}\Vert$. Then, we get $\mathcal{A}_{\ps,r,\sigma}$ as the WFA for the sum of the series  $\Vert\mathcal{A}_{\ps,r,\sigma}^{(J)}\Vert$  with $\emptyset \neq J \subseteq [r(i)]$ (Proposition \ref{w-prop-rec}). 
\item[x)] If $\ps={\textstyle\prod ^{\varpi} x^{(i)}}. \ps' $, then we get $\mathcal{A}_{\ps,r,\sigma}$ as the WFA for the shuffle product of the series $\Vert\mathcal{A}_{\ps',r,\sigma[x^{(i)}\rightarrow j]}\Vert$, $j\in [r(i)]$ (Proposition \ref{w-prop-rec}).
\end{itemize}
By our constructions above, we immediately get  $\Vert \ps\Vert(r, \sigma, w) = \Vert \mathcal{A}_{\ps,r,\sigma} \Vert(w)$ for every $w \in I_{p\B(r)}^+$ . Hence, it remains to deal with the time complexity of our translation algorithm. 
  
Taking into account the above induction steps, we show that the worst case run time for our translation algorithm is doubly exponential. Indeed, if $\ps'=\psi$ is a FOEIL formula, then our claim holds by (ii) and Proposition \ref{form-aut}. Then the constructions in steps (iii)-(vi), (viii) and (x) require a polynomial time (cf. Proposition \ref{w-prop-rec}). Finally, the translations in steps (vii) and (ix) require at most a doubly exponential run time because of the following reasons. Firstly, we need to compute all nonempty subsets of $[r(i)]$ which requires an exponential time. Then, due to our restrictions for $\ps'$ in $\ps={\textstyle\sum ^{\odot} x^{(i)}}. \ps' $ and $\ps={\textstyle\sum ^{\varpi} x^{(i)}}. \ps' $, and Proposition \ref{form-aut} (recall also the proof of Proposition \ref{formula-aut}), if a FOEIL subformula $\psi$ occurs in $\ps'$, then we need a polynomial time to translate it to a finite automaton  and by (ii) an exponential time to translate it to a WFA. We should note that if $\ps'$ contains a subformula of the form $\exists ^* x^{(i')}.\psi''$ or $\exists ^{\shuffle} x^{(i')}.\psi''$ or ${\textstyle\sum ^{\odot} x^{(i')}}. \ps'' $ or ${\textstyle\sum ^{\varpi} x^{(i')}}. \ps'' $, then the computation of the subsets of $[r(i')]$ is independent of the computation of the subsets of $[r(i)]$. 
On the other hand, the best case run time of the algorithm is exponential. Indeed, if in step (ii) we get $\mathcal{A}_{\psi, r, \sigma}$ in polynomial time (cf. Proposition \ref{form-aut}) and we need  no translations of steps (vii) and (ix), then the required time is exponential.
\end{proof}

\

Now we are ready to state the proof of Theorem \ref{wsent_to_waut}.

\begin{proof}[Proof of Theorem \ref{wsent_to_waut}] We apply Proposition \ref{wformula-waut}. Since $\ps$ is a weighted sentence it contains no free variables. Hence, we get a WFA $\mathcal{A}_{\ps,r}$ over $I_{p\B(r)}$ and $K$ such that $\Vert\ps\Vert(r,w)=\Vert\mathcal{A}_{\ps,r}\Vert(w)$ for every $w \in I^+_{p\B(r)}$, and this concludes our proof. The worst case run time for the translation algorithm is doubly exponential and the best case   is exponential.
\end{proof}

\

Next we prove the decidability of the equivalence of wFOEIL sentences over (subsemirings of) skew fields. It is worth noting that the complexity remains the same with the one for the decidability of equivalence for FOEIL sentences (see Section \ref{sec_dec}).

\begin{thm}
Let $K$ be  a (subsemiring of a) skew field, $pw\mathcal{B}= \{wB(i,j) \mid i \in [n], j \geq 1 \} $  a set of parametric weighted components over   $K$,  and  $r:[n] \rightarrow \mathbb{N}$ a mapping. Then, the equivalence  problem for \emph{wFOEIL} sentences over $pw\B$ and $K$ w.r.t. $r$ is decidable in doubly exponential time. 
\end{thm} 
\begin{proof}
It is well known that the equivalence problem for weighted automata, with weights taken in  (a subsemiring of) a skew field, is decidable in cubic time (cf. Theorem 4.10 in \cite{Sa:El}, \cite{Sa:Ra}). Hence, we conclude our result by Theorem \ref{wsent_to_waut}.
\end{proof}

\begin{cor}
Let $pw\mathcal{B}= \{wB(i,j) \mid i \in [n], j \geq 1 \} $ be a set of parametric weighted components over $\mathbb{Q}$ and  $r:[n] \rightarrow \mathbb{N}$ a mapping. Then, the equivalence  problem for \emph{wFOEIL} sentences over $pw\B$ and $\mathbb{Q}$ w.r.t. $r$ is decidable in doubly exponential time. 

\end{cor}

\section{Conclusion}
Efficient modelling of architectures plays a key role in component-based systems in order to be well-defined.
In this paper we propose the formal study of architectures for parametric component-based systems that
consist of a finite number of component types with an unknown number of instances. 
Specifically, we introduce a propositional logic, EPIL, which augments PIL from \cite{Ma:Co} with
a concatenation, a shuffle, and an iteration operator. We then interpret EPIL formulas on finite words over the set of
interactions defined for a given set of ports.
We also study FOEIL, the first-order level of EPIL, as a modelling language for the architectures of parametric systems. 
EPIL and FOEIL are proved expressive enough to
return the permissible interactions and the order restrictions of architectures, as well as to encode 
recursive interactions.
Several examples are presented for modelling parametric architectures with or without ordered interactions by FOEIL sentences.
Also, we show that the equivalence and validity problems for FOEIL sentences are decidable in doubly exponential time,
and the satisfiability problem for FOEIL sentences is decidable in exponential time. Moreover, we show the robustness
of our theory by extending our results 
for the quantitative modelling of parametric architectures. For this, we introduce and study wEPIL and wFOEIL over a commutative semiring. Our weighted logics maintain the qualitative attributes of EPIL and FOEIL, and also model the quantitative properties of architectures, such as the `total' cost or the maximum probability of implementing concrete interactions. We show
 that the equivalence problem for wFOEIL sentences over (a subsemiring
of) a skew field is decidable in doubly exponential time, hence the complexity remains the same with the one for the decidability of FOEIL sentences. Furthermore, we apply wFOEIL
for describing the quantitative aspects of several parametric
architectures.

Work in progress involves the study of the second-order level of EPIL and wEPIL in order to model parametric architectures in the qualitative and weighted setup, respectively,  that cannot be formalized by first-order logics such as Ring, Linear, and  Grid \cite{De:Pa, Ma:Co}. Future work is also investigating the verification problem of parametric systems against formal properties \cite{Ab:Pa,Bo:St}, and
specifically the application of architectures modelled by our logics for studying the behavior and proving properties (such as deadlock-freedom)
in parametric systems. Moreover, it would be interesting to extend our logic-based framework for 
the investigation of dynamic and reconfigurable architectures \cite{Bo:Mo,Ci:Fo,Ma:Sp} as well as for addressing
the architecture composition problem \cite{At:Ge,Bl:Ve}. Another research direction is the 
study of our logics over alternative weight structures, found in applications, like for
instance valuation monoids \cite{Dr:Re,Ka:We}. Formal approaches for architectures are often extended
with tools or graphical languages  for supporting architectures' specification (cf. \cite{Am:RE,Ki:An,Ko:Pa,Ma:Ar,Me:Cl}). 
Therefore, in addition to theoretical directions, future work includes also the development of a tool and a language
in order to facilitate the architecture modelling and identification process of parametric (weighted) systems. 
Finally, in a subsequent work we investigate parametric systems and their architectures in the fuzzy framework
in order to address uncertainty and imprecision resulting from the components' communication.

\section*{Acknowledgment}
  \noindent We are deeply grateful to Simon Bliudze for valuable discussions on a previous version of the paper. Also, we thank anonymous referees for their constructive comments and suggestions that brought the paper in its current form.


\bibliographystyle{alpha}
\bibliography{P_R_special_issue_lmcs_FOEIL_accepted_clean}

\end{document}